\documentclass[12pt,a4paper]{amsart}
\usepackage{comment}
\usepackage{amsmath,amsfonts,amssymb,extarrows}
\usepackage{mathtools}
\usepackage{color}
\usepackage[hmargin=3cm,vmargin=3cm]{geometry}
\usepackage{hyperref}
\usepackage{nameref,zref-xr}                    
\usepackage[all]{xy}           
\usepackage{cancel}
\usepackage{longtable}
\usepackage{dsfont}
\usepackage{multicol}
\usepackage{bbold}
\numberwithin{equation}{section}
\newtheorem{theorem}{Theorem}[section]   
\newtheorem{definition}[theorem]{Definition}
\newtheorem{proposition}[theorem]{Proposition}
\newtheorem{lemma}[theorem]{Lemma}
\newtheorem{corollary}[theorem]{Corollary}

\newtheorem{example-notation}[theorem]{Example-Notation}
\newtheorem{remark}[theorem]{Remark}

\def\d{\partial}

\def\f{\frac}

\def\inw{\in\{1,\dots,n\}}

\def\bs{\boldsymbol}
\newcommand{\eqa}{\begin{eqnarray}}
\newcommand{\eeqa}{\end{eqnarray}}
\newcommand{\beq}{\begin{equation}}
\newcommand{\eeq}{\end{equation}}

\allowdisplaybreaks

\begin{document}
\title[Hamiltonian formalism for non-diagonalisable systems]{Hamiltonian formalism for non-diagonalisable systems of hydrodynamic type}
\author{Paolo Lorenzoni}
\address{P.~Lorenzoni:\newline Dipartimento di Matematica e Applicazioni, Universit\`a di Milano-Bicocca, \newline
Via Roberto Cozzi 55, I-20125 Milano, Italy and INFN sezione di Milano-Bicocca}
\email{paolo.lorenzoni@unimib.it}
\author{Sara Perletti}
\address{S.~Perletti:\newline Dipartimento di Matematica e Applicazioni, Universit\`a di Milano-Bicocca, \newline
Via Roberto Cozzi 55, I-20125 Milano, Italy and INFN sezione di Milano-Bicocca}
\email{sara.perletti1@unimib.it}
\author{Karoline van Gemst}
\address{K.~van Gemst:\newline Dipartimento di Matematica e Applicazioni, Universit\`a di Milano-Bicocca, \newline
Via Roberto Cozzi 55, I-20125 Milano, Italy and INFN sezione di Milano-Bicocca}
\email{karoline.vangemst@unimib.it}

\begin{abstract}
We study the system of first  order PDEs for pseudo-Riemannian metrics governing the Hamiltonian formalism  for systems of hydrodynamic type. In the diagonal setting the integrability conditions ensure the compatibility of this system and, thanks to a classical theorem of Darboux, the existence of a family of solutions depending  on functional parameters. In this paper we study the generalisation of this result to a class of non-diagonalisable systems  of hydrodynamic type that naturally generalises Tsarev's integrable diagonal systems. 
\end{abstract}    

\maketitle

\tableofcontents

\section*{Introduction}
\hspace{-1em}This  paper is  the third of a series of papers (the first two being \cite{LPVG} and \cite{LPVGhodograph}) devoted to the generalisation of Tsarev's theory of integrability  
to non-diagonalisable systems of hydrodynamic type 
\beq\label{SHT-intro}
u^i_t=V^i_j(u)u^j_x,\qquad i\inw,
\eeq
defined by a block-diagonal matrix $V=\text{diag}(V_{(1)} , \dots, V_{(r)})$, $r\leq n$, whose generic $\alpha^{\text{th}}$ block, of size $m_\alpha$, is of the lower-triangular Toeplitz form
\begin{equation}\label{toeplitz}
	V_{(\alpha)}=
	\begin{bmatrix}
		v^{1(\alpha)} & 0 & \dots & 0\cr
		v^{2(\alpha)} & v^{1(\alpha)} & \dots & 0\cr
		\vdots & \ddots & \ddots & \vdots\cr
		v^{m_\alpha(\alpha)} & \dots & v^{2(\alpha)} & v^{1(\alpha)}
	\end{bmatrix}.
\end{equation}
Non-diagonalisable systems of hydrodynamic type \eqref{SHT-intro} arise in numerous areas. Among them we mention generalised integrable  systems  of hydrodynamic type associated with the Frölicher–Nijenhuis bicomplex  \cite{LM}, non-semisimple (bi-)flat F-manifolds and Dubrovin-Frobenius manifolds (see, for instance, \cite{LPR,LP,LP22,LP23} and references therein) and Hamiltonian and bihamiltonian structures of hydrodynamic type and their deformations (see \cite{DVLS}). It is also worth mentioning \cite{FP,KK,KO,VF, XF}. 
\newline
\newline
Block-diagonal systems of the form \eqref{SHT-intro} naturally appear in the theory of regular F-manifolds. In this setting they can be written as
\beq\label{Fsys}
u^i_t=c^i_{jk}X^ju^k_x,\qquad i\inw,
\eeq
where  $X^j$ are the components of a vector field and $c^i_{jk}$ are  the structure functions of a commutative associative product satisfying Hertling-Manin conditions \cite{HM}. In the regular case there exists a distinguished set of coordinates, introduced in \cite{DH}, such that
\begin{equation}\label{strconsts_reg}
	c^{i(\alpha)}_{j(\beta)k(\gamma)}=c^{m_1+\dots+m_{\alpha-1}+i}_{m_1+\dots+m_{\beta-1}+j,m_1+\dots+m_{\gamma-1}+k}=\delta^\alpha_\beta\delta^\alpha_\gamma\delta^i_{j+k-1},
\end{equation}
for all $\alpha,\beta,\gamma\in\{1,\dots,r\}$ and  $i\in\{1,\dots,m_\alpha\}$, $j\in\{1,\dots,m_\beta\}$, ${k\in\{1,\dots,m_\gamma\}}$. In these coordinates, the system \eqref{Fsys} takes a block diagonal form with each block being of lower triangular Toeplitz type, \eqref{toeplitz}.
\newline
\newline
In \cite{LPVG} it was proved that, in the regular setting, the existence of  $n$ commuting flows of the form \eqref{Fsys} allows one to  introduce a connection $\nabla$ satisfying
the two equivalent conditions 
\beq\label{shc-intro}
R^s_{lmi}c^j_{ks}+R^s_{lik}c^j_{ms}+R^s_{lkm}c^j_{is}=0,\qquad R^j_{skl}c^s_{mi}+R^j_{smk}c^s_{li}+R^j_{slm}c^s_{ki}=0.
\eeq
This leads to the study of a class of F-manifolds introduced in \cite{LPR} (see also \cite{LP}) which goes by the name of \emph{F-manifolds with compatible connection (and flat unit)}. These are  F-manifolds equipped with a torsionless  connection $\nabla$ which is  compatible with the product and satisfying the condition \eqref{shc-intro}. In other words, in the regular case,  an integrable system of hydrodynamic type locally defines the structure of an F-manifold with compatible connection (and flat unit).
\newline
\newline
By exploiting this relation  we generalised, in \cite{LPVGhodograph},  Tsarev's  integration method to regular non-diagonalisable systems showing that  solutions  of  the  
algebraic system
\[x\,\delta^{i}_1+t\,v^{i(\alpha)}=w^{i(\alpha)},\qquad\alpha\in\{1 , \dots, r\},\,i\in\{1 , \dots, m_{\alpha}\},\]
where $w^{i(\alpha)}$ are the entries  of  the block-diagonal matrix $W=\text{diag}(W_{(1)} , \dots, W_{(r)})$ defining a symmetry of the system \eqref{SHT-intro}, define  solutions to the system \eqref{SHT-intro}.
\newline
\newline
In  the same paper we observed that, in  order to fully extend Tsarev's theory,  one  is forced 
 to  restrict to  the subclass of  systems  defined by functions,   $v^{i(\alpha)}$, depending  on a \textit{subset} of the variables $(u^1 , \dots, u^n)$. More precisely, 
\[v^{i(\alpha)}=v^{i(\alpha)}(u^{1(1)} , \dots, u^{k_1(1)} , \dots, u^{1(r)} , \dots, u^{k_r(r)}),\]
where $k_s=m_s$ if $i>m_s$ and  $k_s=i$ otherwise.
Under this additional assumption, the  linear system for the symmetries of \eqref{SHT-intro} can be reduced to a sequence of subsystems satisfying the hypothesis of a classical theorem of  Darboux. It turns out that the symmetries depends on  $n$  arbitrary functions  of  a single variable. Moreover, like in the diagonal case, locally any solution of  \eqref{SHT-intro} can be defined in this way.
\newline
\newline
This  paper  is devoted to the study of the Hamiltonian formalism for systems of the form
\eqref{SHT-intro}. According to Dubrovin and Novikov, local Hamiltonian operators for systems of hydrodynamic type are defined by a flat pseudo-Riemannian metric and by the associated  Christoffel symbols. The metrics defining the local Hamiiltonian operators for the system \eqref{SHT-intro} are  solutions of the  system
\begin{eqnarray}
\label{DN1intro}
g_{ik}V^k_j&=&g_{jk}V^k_i,\\
\label{DN2intro}
\tilde\nabla_iV^k_j&=&\tilde\nabla_jV^k_i,
\end{eqnarray}
where $\tilde\nabla$ is  the Levi-Civita connection of the metric $g$. Non-flat solutions of
 the  above system define  non-local Hamiltonian operators  provided that  $g$ satisfies Gauss-Peterson-Mainardi-Codazzi type conditions. In the diagonal case,  assuming the characteristic velocities to be  pairwise distinct, the above system implies that the metric $g$ is diagonal with it's components satisfying
\begin{equation}\label{metrics_intro}
\partial_j g_{ii}=2\f{\d_jv^i}{v^j-v^i}g_{ii},\qquad i\ne j.
\end{equation}
Remarkably,  compatibility of the system \eqref{metrics_intro} coincides with Tsarev's integrability conditions, also called semi-Hamiltonian conditions.
\newline
\newline
In the non-diagonalisable regular case, the Dubrovin-Novikov-Tsarev system (\ref{DN1intro}, \ref{DN2intro})  does not admit solutions in general. However, we show the following. 
\newline
\begin{itemize}
\item
For the $n$-component systems \eqref{Fsys}  associated with F-manifolds with compatible connection and flat unity, the Dubrovin-Novikov-Tsarev system is equivalent  to the system
\begin{equation}\label{eq_for_g_intro}
c^r_{hj}\nabla_i\theta_r+c^r_{hi}\nabla_j\theta_r-c^r_{ij}\nabla_h\theta_r=c^r_{ij}c^s_{hr}\nabla_e\theta_s-\theta_r\nabla_hc^r_{ij}.
\end{equation}
\newline
\item Under the additional Darboux assumptions ensuring the completeness of the symmetries,
 in the case of  $r$ blocks  of size $m_1 , \dots, m_r$, the metrics obtained from the Dubrovin-Novikov-Tsarev system depend on $r$ arbitrary functions of a single variable
\[f_{1(1)}(u^{1(1)}) , \dots, f_{1(r)}(u^{1(r)})\]
and $n-r$ functions of two variables 
\[\quad\quad\quad f_{2(1)}(u^{1(1)},u^{2(1)}) ,..., f_{m_1(1)}(u^{1(1)},u^{2(1)}) , ..., f_{2(r)}(u^{1(r)},u^{2(r)}),..., f_{m_r(r)}(u^{1(r)},u^{2(r)}),\]
for which a block only contributes if its size is greater than 1.
\end{itemize}
The paper is organised as follows. In Section 1, we introduce F-manifolds with compatible connection and the associated integrable hierarchies. This section is mainly based on the results of \cite{LPVG, LPVGhodograph} (in the appendix we prove that one of the result of \cite{LPVG} can be improved removing a technical assumption). In Section 2, we study the associated system for densities of conservation laws $h$
\begin{equation}\label{dcl_intro}
\partial_j\partial_ih=\Gamma_{ji}^s\partial_sh+c^s_{ij}e^t\nabla_s\partial_th,
	 \end{equation} 
and we show that the right-hand side of the system satisfies the compatibility condition
\[\partial_k(\Gamma_{ji}^s\partial_sh+c^s_{ij}e^t\nabla_s\partial_th)=\partial_j(\Gamma_{ki}^s\partial_sh+c^s_{ik}e^t\nabla_s\partial_th)\]
on  the  solutions of the system. In Section  3, following \cite{LPVGhodograph}, we introduce a class of systems, which we call of \emph{Darboux-Tsarev} type, satisfying  Darboux's  conditions  ensuring the completeness of the symmetries. Under these  additional assumptions, we show that  the general solution of the system \eqref{dcl_intro} depends on $n$  arbitrary functions of a single variable.  Section 4 is devoted to the Dubrovin-Novikov-Tsarev system. We first recall its form in the diagonal case, then we discuss  the general regular case and we prove its compatibility under Darboux's  conditions. The proof is given in canonical coordinates and requires the use of the additional properties of the Christoffel symbols coming from Darboux's  conditions. Section 4 also contains examples illustrating the theory.
\newline
\newline
\noindent{\bf Acknowledgements}. The authors are supported by funds of INFN (Istituto  Nazionale di Fisica Nucleare) by IS-CSN4 Mathematical Methods of Nonlinear  Physics. Authors are also thankful to GNFM (Gruppo Nazionale di Fisica Matematca) for supporting activities that contributed to the research reported in this paper. This research has received funding by the Italian PRIN 2022 (2022TEB52W)  \emph{The charm of integrability: from nonlinear waves to random matrices}.
 
\section{F-manifolds with compatible connection and  associated integrable hierarchies}\label{SectionFmnfs}
\subsection{F-manifolds}
F-manifolds were introduced by Hertling and Manin in \cite{HM} and are defined as follows.
\begin{definition}\label{defFmani}
An \emph{F-manifold} is a manifold $M$ equipped with
\begin{itemize}
\item[(i)] a commutative associative bilinear product  $\circ$  on the module of (local) vector fields, satisfying the following identity:
\begin{equation}\label{HMeq1free}
\mathcal{L}_{X\circ Y} \circ=X\circ (\mathcal{L}_Y \circ) +Y\circ (\mathcal{L}_X\circ ),
\end{equation}
for all local vector fields $X,Y$.
\newline
\item[(ii)] a distinguished vector field $e$ on $M$ such that 
\[e\circ X=X\] 
for all local vector fields $X$. 
\end{itemize}
\end{definition}
Condition \eqref{HMeq1free} is known as the \emph{Hertling-Manin condition}.

\subsection{F-manifolds with compatible connection}
F-manifolds are usually equipped with additional structures. 
\begin{definition}\label{defFwithE}
An \emph{F-manifold with Euler vector field} is an F-manifold $M$ equipped with a vector field $E$ satisfying
\begin{equation}
	\mathcal{L}_E \circ=\circ.\label{Euler}
\end{equation}
\end{definition}

\hspace{-1em}Following \cite{LPR}, we now introduce the notion of F-manifold with compatible connection (and flat unit).
\begin{definition}\label{Fmnfwcompatconn_flatunit}
An \emph{F-manifold with compatible connection and flat unit} is a manifold $M$ equipped with a product 
\[\circ : TM \times TM \rightarrow TM,\] 
with structure functions $c^i_{jk}$, a connection $\nabla$ with Christoffel symbols 
$\Gamma^i_{jk}$ and a distinguished vector field $e$ such that
\begin{itemize}
\item[(i)] the one-parameter family of  connections $\{\nabla_{\lambda}\}_{\lambda}$ with Christoffel symbols
$$\Gamma^i_{jk}-\lambda c^i_{jk},$$
gives a torsionless connection for any choice of $\lambda$ for which the Riemann  tensor coincides
 with the Riemann tensor of $\nabla \equiv \nabla_0$
\beq\label{curvlambda}
R_{\lambda}(X,Y)(Z)=R(X,Y)(Z),
\eeq
and satisfies the condition
\beq\label{rc-intri}
Z\circ R(W,Y)(X)+W\circ R(Y,Z)(X)+Y\circ R(Z,W)(X)=0,
\eeq
for all local vector fields $X$, $Y$, $Z$, $W$;
\item[(ii)] $e$ is the unit of the product;
\item[(iii)] $e$ is flat: $\nabla e=0$.
\end{itemize}
\end{definition}
\hspace{-1em}For a given $\lambda$, the torsion and curvature are respectively given by
\begin{eqnarray*}
T^{(\lambda)k}_{ij}&=&\Gamma^k_{ij}-\Gamma^k_{ji}+\lambda(c^k_{ij}-c^k_{ji}),\\
R^{(\lambda)k}_{ijl}&=&R^k_{ijl}+\lambda(\nabla_i c^k_{jl}-\nabla_j c^k_{il})+\lambda^2(c^k_{im}c^m_{jl}-c^k_{jm}c^m_{il}),
\end{eqnarray*}
where $R^k_{ijl}$ is the Riemann tensor of $\nabla$. Thus, condition (i) is equivalent to
\begin{enumerate}
\item the vanishing of the torsion of $\nabla$;
\item the commutativity of the product  $\circ$;
\item the symmetry  in the lower indices of the tensor field $\nabla_l c^k_{ij}$;
\item the associativity of the product $\circ$.
\end{enumerate}

\begin{remark}
Since the Hertling-Manin condition \eqref{HMeq1free} follows from the symmetry  in the lower indices of the tensor field $\nabla_l c^k_{ij}$, as shown in \cite{He02},  F-manifolds with compatible connection and flat unit constitute a special class of F-manifolds.
\end{remark}

\begin{remark}
The  operator  of multiplication by the Euler vector field has vanishing Nijenhuis torsion. 
By virtue of this property, families of examples were constructed in \cite{LP23,LPVG}.
\end{remark}

\begin{remark}	
Using Bianchi  identity condition \eqref{rc-intri} can be written in the equivalent  form  
\begin{equation}\label{rc-intri-2}
R(Y,Z)(X\circ W)+R(X,Y)(Z\circ W)+R(Z,X)(Y\circ W)=0,
\end{equation}
for all local vector fields $X$, $Y$, $Z$, $W$ (see \cite{LPR} for  details).
\end{remark}
\hspace{-1em}In local coordinates conditions \eqref{rc-intri} and \eqref{rc-intri-2} read  \eqref{shc-intro}.

\subsection{Regularity} Following \cite{DH}, we present the notion of a regular F-manifold.
\begin{definition}[\cite{DH}]\label{DavidHertlingdef}
An F-manifold with Euler vector field $(M,\circ, e,E)$ is called \emph{regular} if for each $p\in M$ the matrix representing the endomorphism
$$L_p := E_p\circ : T_pM \to T_pM$$
has exactly one Jordan block for each distinct eigenvalue.
\end{definition}
\hspace{-1em}In this paper, we will use the following important result of \cite{DH} regarding the existence of non-semisimple canonical coordinates for regular F-manifolds with Euler vector field.
\begin{theorem}[\cite{DH}]\label{DavidHertlingth}
Let $(M, \circ, e, E)$ be a regular F-manifold of dimension $n \geq 2$ with an Euler vector field $E$. Furthermore, assume that locally around a point $p\in M$, the Jordan canonical form of the operator $L$ has $r$ Jordan blocks of sizes $m_1 , \dots, m_r$ with distinct eigenvalues. Then
there exists locally around $p$ a distinguished system of coordinates $\{u^1, \dots, u^{m_1+\dots+m_r}\}$ such that
 \begin{align}
	e^{i(\alpha)}&=\delta^i_1,\qquad
	E^{i(\alpha)}=u^{i(\alpha)},\qquad
	c^{i(\alpha)}_{j(\beta)k(\gamma)}=\delta^\alpha_\beta\delta^\alpha_\gamma\delta^i_{j+k-1},\notag
\end{align}
for all $\alpha,\beta,\gamma\in\{1,\dots,r\}$ and  $i\in\{1,\dots,m_\alpha\}$, $j\in\{1,\dots,m_\beta\}$, ${k\in\{1,\dots,m_\gamma\}}$.
\end{theorem}
\hspace{-1em}The coordinates defining this special system are called \emph{David-Hertling coordinates}, and in such coordinates the operator  of multiplication by the Euler vector field takes the lower-triangular Toeplitz form \eqref{toeplitz} mentioned in the introduction. We will also often refer to David-Hertling coordinates as \emph{canonical coordinates} since they are a generalisation of Dubrovin's canonical coordinates in the semisimple case.

\subsection{Integrable hierarchies and  F-manifolds with  compatible connection}\label{SubsectionIntHsAndFmnfsCC}
According to the results of \cite{LPVG}, given  an $n$-dimensional regular F-manifold with Euler vector field $(M,\circ,e,E)$ and a vector field $X$, under mild technical assumptions, there exists a unique torsionless connection $\nabla$ defined by the conditions
\begin{equation}\label{cond1}
\nabla_j e^i=0,\qquad i,j\inw,
\end{equation}
\begin{equation}\label{cond2}
(d_{\nabla}X\circ)^i_{jk}=0,\qquad i,j,k\inw.
\end{equation}
 More precisely, we recall the following.
\begin{proposition}\label{LPVG_Prop3.10}(\cite{LPVG})
	For $(M,\circ,e)$ being an $n$-dimensional regular F-manifold and $X$ being a local vector field realising
	\begin{itemize}
		\item $X^{1(\alpha)}\neq X^{1(\beta)}$ for $\alpha\neq\beta$, $\alpha,\beta\in\{1,\dots,r\}$,
		\item $X^{2(\alpha)}\neq0$, $\alpha\in\{1,\dots,r\}$,
	\end{itemize}
	there exists a unique torsionless connection $\nabla$ satisfying \eqref{cond1} and \eqref{cond2}.
\end{proposition}

\begin{remark}\label{ChristoffelExplicit}
The proof of Proposition \ref{LPVG_Prop3.10} provides explicit formulas for the Christoffel symbols $\{\Gamma^i_{jk}\}_{i,j,k\in\{1,\dots,n\}}$ of $\nabla$, in terms of $\{\partial_j X^i\}_{i,j\in\{1,\dots,n\}}$, which we shall now recall. Let us fix $\alpha,\beta,\gamma\in\{1,\dots,r\}$, $\alpha\neq\beta\neq\gamma\neq\alpha$, and $i\in\{1,\dots,m_\alpha\}$. The Christoffel symbols of the unique torsionless connection satisfying $\nabla e=0$ and $d_\nabla V=0$, for $V=X\circ$, are determined by the following formulas:\small
	\begin{align}
		 \quad&\Gamma^{i(\alpha)}_{j(\beta)k(\gamma)}=0,\qquad j\in\{1,\dots,m_\beta\},\,k\in\{1,\dots,m_\gamma\};\label{Chr_a}\\
		 \label{Chr_b}\quad&\Gamma^{i(\alpha)}_{j(\beta)k(\alpha)}=-\frac{1}{X^{1(\alpha)}-X^{1(\beta)}}\Bigg(\partial_{j(\beta)}X^{(i-k+1)(\alpha)}+\overset{m_\alpha}{\underset{s=k+1}{\sum}}\Gamma^{i(\alpha)}_{j(\beta)s(\alpha)}X^{(s-k+1)(\alpha)}\\&\quad\quad\qquad\,\,\,\,-\overset{m_\beta}{\underset{s=j+1}{\sum}}\Gamma^{i(\alpha)}_{k(\alpha)s(\beta)}X^{(s-j+1)(\beta)}\Bigg), \qquad k\in\{1,\dots,m_\alpha\},\,j\in\{1,\dots,m_\beta\};\notag\\
		\quad&\Gamma^{i(\alpha)}_{j(\beta)1(\beta)}=-\Gamma^{i(\alpha)}_{j(\beta)1(\alpha)}\label{Chr_nablae_1},\qquad j\in\{1,\dots,m_\beta\},\\
		&\Gamma^{i(\alpha)}_{j(\alpha)1(\alpha)}=-\overset{}{\underset{\sigma\neq\alpha}{\sum}}\,\Gamma^{i(\alpha)}_{j(\alpha)1(\sigma)},\qquad j\in\{1,\dots,m_\alpha\};\label{Chr_nablae_2}\\
		\quad&\Gamma^{i(\alpha)}_{k(\alpha)j(\alpha)}=\frac{1}{X^{2(\alpha)}}\bigg(\partial_{(j-1)(\alpha)}X^{(i-k+1)(\alpha)}+\overset{m_\alpha}{\underset{s=k+1}{\sum}}\Gamma^{i(\alpha)}_{(j-1)(\alpha)s(\alpha)}X^{(s-k+1)(\alpha)}\label{Chr_d}\\&\quad\quad\quad\,\,\,\,-\partial_{k(\alpha)}X^{(i-j+2)(\alpha)}-\overset{m_\alpha}{\underset{s=j+1}{\sum}}\Gamma^{i(\alpha)}_{k(\alpha)s(\alpha)}X^{(s-j+2)(\alpha)}\bigg),\notag\\&\text{for } j, k\in\{2,\dots,m_\alpha\},\text{ where }j\leq k;\notag\\
		&\Gamma^{i(\alpha)}_{k(\beta)j(\beta)}=\frac{1}{X^{2(\beta)}}\Bigg(\overset{m_\beta}{\underset{s=k+1}{\sum}}\Gamma^{i(\alpha)}_{(j-1)(\beta)s(\beta)}X^{(s-k+1)(\beta)}-\overset{m_\beta}{\underset{s=j+1}{\sum}}\Gamma^{i(\alpha)}_{k(\beta)s(\beta)}X^{(s-j+2)(\beta)}\Bigg),\label{Chr_c}\\&\text{for }k, j\in\{2,\dots,m_\beta\},\text{ where }j\leq k.\notag
	\end{align}\normalsize
\end{remark}
\hspace{-1em}The above connection satisfies the following lemmas, which are either proved in \cite{LPVGhodograph} or are immediate consequences of results therein. 
\vspace{-2em}
\newline
\newline
\begin{lemma}[Lemma 4.1 in \cite{LPVGhodograph}] 
    \begin{align}
          & \,   \Gamma^{i(\alpha)}_{j(\beta)k(\alpha)} = \begin{cases}
                \Gamma^{(i-k+1)(\alpha)}_{j(\beta)1(\alpha)}, & \text{ if }\, i \geq k,\\
                0, & \text{ otherwise,}
            \end{cases} \label{eq:lemmaold1diag}\\
       & \,      \Gamma^{i(\alpha)}_{j(\beta)k(\beta)} = \begin{cases}
                \Gamma^{i(\alpha)}_{(j+k-1)(\beta) 1(\beta)}, & \text{ if }\, j+k \leq m_{\beta} + 1,\\
                0, & \text{ otherwise,}
            \end{cases}
            \label{eq:lemmaold1hori}\\
            & \, 
            \Gamma^{i(\alpha)}_{1(\alpha)j(\alpha)} = \begin{cases}
                \Gamma^{(i-j+1)(\alpha)}_{1(\alpha) 1(\alpha)}, & \text{ if }\, i \geq j,\\
                0, & \text{ otherwise.}
            \end{cases}
            \label{eq:lemmaold1diagsameblock}
    \end{align}
    \label{lemma:old1}
\end{lemma}
\hspace{-1em}Note that  \eqref{eq:lemmaold1diagsameblock} is not included in \cite{LPVGhodograph} but can be easily obtained from the condition $\nabla_jc^i_{ks}-\nabla_kc^i_{js}=0$. Indeed, in canonical coordinates such a condition reads 
	\begin{align}\label{SymmNablac}
		\Gamma^{i(\alpha)}_{j(\beta)t(\tau)}c^{t(\tau)}_{k(\gamma)s(\sigma)}-\Gamma^{t(\tau)}_{j(\beta)s(\sigma)}c^{i(\alpha)}_{k(\gamma)t(\tau)}-\Gamma^{i(\alpha)}_{k(\gamma)t(\tau)}c^{t(\tau)}_{j(\beta)s(\sigma)}+\Gamma^{t(\tau)}_{k(\gamma)s(\sigma)}c^{i(\alpha)}_{j(\beta)t(\tau)}=0.
	\end{align}
Condition \eqref{eq:lemmaold1diagsameblock} is thus obtained by letting all the indices in \eqref{SymmNablac} to be associated to the same block (which we denote by $\alpha$) and setting one of the lower indices to 1. 
\begin{lemma}[Lemma 4.2 in \cite{LPVGhodograph}] 
    The following are equivalent
    \begin{align}
     & \,    \Gamma^{i(\alpha)}_{j(\alpha)k(\alpha)} = 0, \qquad \text{ for } k \geq 2, \, j \geq i+1; \\
     & \,\Gamma^{i(\alpha)}_{j(\alpha)k(\alpha)} = 0, \qquad \text{ for } j, k \geq 2, \, i-j-k \leq -3. 
    \end{align}
    \label{lemma:old2b}
\end{lemma}

\vspace{0.8em}
\hspace{-1em}An F-manifold with compatible connection and flat unit can be constructed by means of such a connection, when the local vector field $X$ meeting the assumptions of Proposition \ref{LPVG_Prop3.10} belongs to a set of $n$ linearly independent local vector fields defining pairwise commuting flows.
\begin{theorem}\label{LPVG_mainTh}(\cite{LPVG})
	Let $\{X_{(1)},\dots,X_{(n)}\}$ be a set of linearly independent local vector fields on an $n$-dimensional regular F-manifold $(M,\circ,e)$. Let's assume that the corresponding flows
	\begin{align}
		{\bf  u}_{t_i}=X_{(i)}\circ {\bf u}_x,\qquad i\in\{1,\dots,n\},
		\notag
	\end{align}
	pairwise commute, and that there exists a local vector field $X\in\{X_{(1)},\dots,X_{(n)}\}$ such that ${X^{1(\alpha)}\neq X^{1(\beta)}}$ for ${\alpha\neq\beta}$, $\alpha,\beta\in\{1,\dots,r\}$, and $X^{2(\alpha)}\neq0$ for each ${\alpha\in\{1,\dots,r\}}$. Then an F-manifold $(M,\circ,e,\nabla)$ with compatible connection and flat unit is defined by the unique torsionless connection $\nabla$ realising $\nabla e=0$ and $d_\nabla(X\circ)=0$.
\end{theorem}
\hspace{-1em}The proof of Theorem \ref{LPVG_mainTh} relies on the following facts:
\begin{itemize}
\item if $Z$ is a local vector field realising $d_\nabla(Z\circ)=0$ then
\begin{equation}\label{LPVG_eq: 3RCX}
	(R^k_{lmi}c^h_{pk} + R^k_{lip}c^h_{mk} + R^k_{lpm}c^h_{ik})Z^l = 0,\qquad i,p,m,h\inw,
\end{equation}
for any $Z\in\{X_{(1)},\dots,X_{(n)}\}$. The assumption of $\{X_{(i)}\}_{i\in\{1,\dots,n\}}$ providing a basis allows one to conclude that \eqref{LPVG_eq: 3RCX} must hold for any arbitrary local vector field $Z$, implying \eqref{shc-intro}.
\item  if $Z$ is a local vector field realising $d_\nabla(Z\circ)=0$ then
\begin{equation}
		\big(\nabla_jc^i_{ks}-\nabla_kc^i_{js}\big)Z^s=0.
		\label{almostsymmetryofnablac}
	\end{equation}
for any $Z\in\{X_{(1)},\dots,X_{(n)}\}$,  implying the symmetry  of  the tensor field $\nabla_jc^i_{ks}$. In the  Appendix  we show that in fact this second property does not require the existence of a frame of vector fields defining commuting flows.   
\end{itemize}
On the other hand, according to the results of \cite{LPR}, given an F-manifold with compatible connection and flat unit vector field, it is possible to introduce a set of commuting flows of the form
\begin{equation}\label{flowX}
u^i_{t_k}=V^i_ku^k_x=c^i_{jk}X_{(k)}^ju^k_x,\qquad i\inw,\quad k=1,2,\dots
\end{equation}
by requiring the local vector fields $X_{(1)},X_{(2)},\dots$ defining such flows to be solutions of the equation $d_{\nabla}(X\circ)=0$. Moreover, the flows \eqref{flowX} admit common solutions  $\textbf{u}(x,t_1,t_2,\dots)$ obtained, in implicit form, by means of the generalised hodograph formula.

\begin{theorem}
Let $(M,\circ,e,\nabla)$ be an F-manifold with compatible connection and flat unit, then  
 a vector valued  function $\textbf{u}(x,t_1,t_2,\dots)$ satisfying
\begin{equation}\label{ndhm}
x\,e^{i}+\sum_kt_k\,X^{i}_{(k)}(\textbf{u}(x,t_1,t_2,\dots))=0,\qquad i\inw,
\end{equation}
where $X_{(1)},X_{(2)},\dots$ are solutions of of the equation $d_{\nabla}(X\circ)=0$, satisfies simultaneously all the systems \eqref{flowX}.
\end{theorem}
\hspace{-1em}The proof (given in  \cite{LPVGhodograph})  is based  on the following facts:
\begin{enumerate}
\item  The  matrix
\begin{equation}\label{TsMatrix}
M^i_j:=\sum_kt_k\,\d_jX^i_{(k)},\qquad i,j\inw.
\end{equation}
evaluated on the function  $\textbf{u}(x,t_1,t_2,\dots)$ defined by \eqref{ndhm} coincides with  the matrix
\begin{equation}\label{TsMatrix_nabla}
M^i_j(\textbf{u}(x,t_1,t_2,\dots))=\sum_kt_k\,\nabla_jX^i_{(k)},\qquad i,j\inw.
\end{equation}
\item Using this fact it is immediate to check that $M^i_j(\textbf{u}(x,t_1,t_2,\dots))$ satisfies the condition
 \begin{equation}
c^i_{js}M_k^s(\textbf{u}(x,t_1,t_2,\dots))=c^i_{ks}M_j^s(\textbf{u}(x,t_1,t_2,\dots)),\qquad i,j,k\inw,
\label{eq:cond3temporary}
\end{equation}
and this implies that $M(\textbf{u}(x,t_1,t_2,\dots))$ has precisely the form
\[M^i_j(\textbf{u}(x,t_1,t_2,\dots))=c^i_{jk}Z^k,\qquad i,j\inw,\]  
for some vector valued function $Z$.
\item  By differentiating \eqref{ndhm} with respect to $x$ and $t_k$ respectively, we  get
\begin{equation}\label{xtder}
e^i+M^i_j(\textbf{u}(x,t_1,t_2,\dots))u^j_x=0,\qquad X^i_{(k)}+M^i_j(\textbf{u}(x,t_1,t_2,\dots))u^j_{t_k}=0.
\end{equation}
Taking into account that $(M(\textbf{u}(x,t_1,t_2,\dots))^{-1})^i_j=c^i_{jk}(Z^{-1})^k$
 where $Z$ is the unique  solution of the equation $M(\textbf{u}(x,t_1,t_2,\dots))Z^{-1}=e$, from \eqref{xtder} we obtain
\[u^i_x=-c^i_{js}(Z^{-1})^se^j=-(Z^{-1})^i,\quad u^i_{t_k}=-c^i_{js}(Z^{-1})^sX^j_{(k)}=c^i_{js}X^j_{(k)}u^s_x.\]
\end{enumerate}
The above procedure allows one to obtain solutions starting from symmetries (commuting flows). In  \cite{LPVGhodograph} we have proved that, under some additional conditions, all the solutions can be obtained in  this way. We have called this property \emph{completeness  of the symmetries}.

\section{Densities of conservation laws}
\hspace{-1em}Given a system of hydrodynamic type
\[u^i_t=V^i_j({\bf u})u^j_x,\]
densities of  conservation laws are functions $h(u^1 , \dots, u^n)$  such that $h_t$ is a total  $x$-derivative. Under periodic or vanishing (at infinity) boundary conditions, this ensures that the integral of these quantities (over the interval of definition in the periodic case and over the real line in the vanishing case) is constant on  the solutions of the system. Due to the results of \cite{DT}  on the variational bicomplex in the  formal calculus of variations  (see also \cite{DZ}),
 in order to prove that $h_t$ is a total $x$-derivative we need to prove that $\f{\delta \overline{h}_t}{\delta u^i}=0$, where $\overline{h}=\int h(u)\,dx$.
 
\begin{proposition}
Let $\nabla$ be a torsionless connection satisfying  $d_{\nabla}V=0$. Then the vanishing of $\f{\delta \overline{h}_t}{\delta u^l}$ is equivalent to  the condition:
\begin{equation}
V^s_k\nabla_sdh_l-V^s_l\nabla_s dh_k=0.
\end{equation}
\end{proposition}
\begin{proof}
	By straightforward computation we  get
	\begin{equation}
		\f{\delta \overline{h}_t}{\delta u^l}=\left(\d_l\d_i h\,V^i_k-\d_k\d_i h\,V^i_l+\d_i h\left(\d_lV^i_k-\d_kV^i_l\right)\right)u^k_x.
	\end{equation}
	Using the identity
	\[\d_lV^i_k-\d_kV^i_l=(d_{\nabla}V)^i_{lk}+\Gamma^i_{ks}V^s_l-\Gamma^i_{ls}V^s_k,\]
	we obtain
	\begin{eqnarray*}
		\f{\delta h_t}{\delta u^l}&=&\left(\d_l\d_i h\,V^i_k-\d_k\d_i h\,V^i_l+\d_i h\left((d_{\nabla}V)^i_{lk}+\Gamma^i_{ks}V^s_l-\Gamma^i_{ls}V^s_k\right)\right)u^k_x\\
		&=&\left(V^i_k(\d_i\d_l h-\Gamma^s_{li}\d_s h)-V^i_l(\d_i\d_k h-\Gamma^s_{ki}\d_s h)+(d_{\nabla}V)^s_{lk}\d_s h\right)u^k_x\\
		&=&\left(V^i_k\nabla_idh_l-V^i_l\nabla_i dh_k+\left(d_{\nabla}V\right)^s_{lk}dh_s\right)u^k_x.
	\end{eqnarray*}
\end{proof}

\hspace{-1em}As a  consequence of the  above  proposition and of Theorem  \ref{LPVG_mainTh}  we have the following result.
\begin{proposition}\label{sys-cl}
	Let $\{X_{(1)},\dots,X_{(n)}\}$ be a set of linearly independent local vector fields on an $n$-dimensional regular F-manifold $(M,\circ,e)$ defining commuting  flows of the
	form
\[u^i_{t_{(k)}}=(X_{(k)}\circ{\bf u})^i_x,\qquad i=1 , \dots, n,\]	
	satisfying conditions of  Theorem  \ref{LPVG_mainTh} and let $\nabla$ be the associated torsionless compatible connection $\nabla$. Then a function $h$ is  a density of conservation laws
	 for the above flows and for any flow defined by a solution of the equation $d_{\nabla}(X\circ)=0$ if and only  if it satisfies the condition  
	 \begin{equation}\label{dcl1}
	 c^s_{kj}\nabla_sdh_i=c^s_{ij}\nabla_s dh_k.
	 \end{equation} 
\end{proposition} 

\hspace{-1em}Condition  \eqref{dcl1}  can be simplified as follows. 

\begin{proposition}
 Condition \eqref{dcl1} can be reduced to
  \begin{equation}\label{dcl2}
	 \nabla_j\omega_i=c^s_{ij}e^k\nabla_s\omega_k.
	 \end{equation} 
	 \end{proposition}
	\begin{proof}
		Rewriting \eqref{dcl2} as
		\begin{equation}\label{dcl3}
			\partial_j\omega_i=\Gamma_{ji}^s\omega_s+c^s_{ij}e^t\nabla_s\omega_t,
		\end{equation}  
		we observe that 
		\[\partial_j\omega_i-\partial_i\omega_j=(\Gamma_{ji}^s-\Gamma_{ij}^s)\omega_s+(c^s_{ij}-c^s_{ji})e^k\nabla_s\omega_k=0,\]
		implying that locally there exists a function $h$  such that $\omega_i=\partial_i h$.
		Moreover  substituting $\nabla_jdh_i=c^s_{ij}e^k\nabla_s dh_k$ in \eqref{dcl1}  we get the identity
		\[c^s_{kj}c^t_{is}e^m\nabla_tdh_m=c^s_{ki}c^t_{js}e^m\nabla_tdh_m.\]
		Conversely, multiplying  both  sides of \eqref{dcl1} by $e^k$ and taking the sum over $k$, we get \eqref{dcl2} with $\omega$ replaced by $dh$.
	\end{proof} 

\vspace{1em}
    
	\hspace{-1em}In the semisimple  case the system  \eqref{dcl3} reads
	\[ \partial_j\omega_i=\Gamma_{ji}^s\omega_s+\delta^i_je^t\nabla_i\omega_t,\]
	with $e^i=1$, $\Gamma^i_{jk}=0$ for distinct indices, $\Gamma^{i}_{jj}=-\Gamma^i_{ij}$ and $\Gamma^i_{ii}=-\sum_{j\ne i}\Gamma^{i}_{ij}$. For $i=j$ we have 
\begin{eqnarray*}
\partial_i\omega_i&=&\Gamma_{ii}^s\omega_s+e^t\nabla_i\omega_t=\Gamma_{ii}^s\omega_s+\nabla_i\omega_i+\sum_{t\ne  i}\nabla_i\omega_t.
\end{eqnarray*}
Taking into account  the equations for $i\ne j$ we get an identity.
\newline
\newline
For $i\ne j$  we obtain
\[ \partial_j\omega_i=\Gamma_{ji}^i\omega_i+\Gamma_{ji}^j\omega_j.\]
Taking into account that $\omega_i=\partial_i h$ and that $\Gamma_{ji}^i=\frac{\partial_j X^i}{X^j-X^i}$ for $i\neq j$, we obtain the well known system for densities of conservation laws  for diagonal  systems with pairwise distinct  characteristic velocities.
\newline
\newline
The compatibility of such system coincides with Tsarev's integrability condition and ensures
 the  existence of  a  family of solutions depending on $n$ arbitrary  functions of a single variable \cite{ts91}.

\subsection{Compatibility conditions}
Given a solution $(\omega_1 , \dots, \omega_n)$ of the system \eqref{dcl1}, taking into account that
\[\d_k\d_j\omega_i-\d_j\d_k\omega_i=0\]
we should have
\begin{equation}\label{conscond}
\d_k\left(\Gamma_{ji}^s\omega_s+c^s_{ij}e^t\nabla_s\omega_t\right)-\d_j\left(\Gamma_{ki}^s\omega_s+c^s_{ik}e^t\nabla_s\omega_t\right)=0.
\end{equation}

\begin{proposition}\label{compsysdcl}
Let $(M,\nabla,\circ,e)$ an F-manifold with  compatible connection and let 
\[{\bf u}_t=X\circ{\bf u}_x\]
a system  of hydrodynamic type defined by a solution of  the  equation $d_{\nabla}(X\circ)=0$, then the  system \eqref{dcl3} for densities of conservation laws 
satisfies \eqref{conscond}.
\end{proposition}

\begin{proof}
After a long but straightforward computation  one obtains
\[\d_k\left(\Gamma_{ji}^s\omega_s+c^s_{ij}e^t\nabla_s\omega_t\right)-\d_j\left(\Gamma_{ki}^s\omega_s+c^s_{ik}e^t\nabla_s\omega_t\right)=R^s_{ijk}\omega_s+c^s_{ji}\nabla_k\nabla_e\omega_s-c^s_{ki}\nabla_j\nabla_e\omega_s.\]
Using the identity
\[(\nabla_k\nabla_m-\nabla_m\nabla_k)\omega_s=R^t_{skm}\omega_t,\]
we have
\[R^s_{ijk}\omega_s+c^s_{ji}\nabla_k\nabla_e\omega_s-c^s_{ki}\nabla_j\nabla_e\omega_s=\left[(R^t_{sjk}c^s_{mi}+R^t_{skm}c^s_{ji}+R^t_{smj}c^s_{ki})e^m\right]\omega_t.\]

\end{proof}

 \hspace{-1em}In the above  computation one uses  the fact that $\omega$ is a solution of the system \eqref{dcl3}. According  to a classical theorem of  Darboux (see \cite{darboux})  the validity of \eqref{conscond} on  the  solutions  of the system  \eqref{dcl3} implies the  existence  of a family of solutions (depending  on functional parameters)  provided that the system has the form 
\[\frac{\partial\omega_i}{\partial u^j}=f_{ij}(\boldsymbol{\omega},{\bf u}).\]
Not all first order partial derivatives of the unknown functions $\omega_i$ must appear  in the system and the computation of the compatibility  conditions should not introduce the missing ones. Due to the appearance of partial derivatives of the unknown functions in the right-hand side, the system \eqref{dcl3} has not the form considered by Darboux. However, under some extra assumptions, the same considered in \cite{LPVGhodograph} to ensure the completeness of the  symmetries,
 the system \eqref{dcl3} can be reduced to a finite sequence of closed subsystems satisfying Darboux's conditions.  More precisely,  the first subsystem of the sequence already has Darboux's form, while the remaining ones contain in the right-hand sides partial derivatives of functions obtained at the previous steps.

\section{Non-diagonalisable systems of Darboux-Tsarev type}
\hspace{-1em}Let $\circ$ be the  product of a regular F-manifold and let
\begin{equation}\label{DTsys}
	u^i_t=(X\circ u_x)^i=c^i_{jk}X^ju^k_x,\qquad i\inw
\end{equation}
be a system  of hydrodynamic type defined by a vector field $X$ such that, in canonical coordinates, ${X^{1(\alpha)}\neq X^{1(\beta)}}$ for ${\alpha\neq\beta}$, $\alpha,\beta\in\{1,\dots,r\}$, and $X^{2(\alpha)}\neq0$ for each ${\alpha\in\{1,\dots,r\}}$. 

\begin{definition}
	Let us denote
	\begin{align*}
		\mathcal{U}^i:=\{u^{j(\sigma)}\,|\,j\leq i\}_{\sigma\in\{1,\dots,r\}}=\{u^{1(\sigma)},\dots,u^{\min\{i,m_\sigma\}(\sigma)}\}_{\sigma\in\{1,\dots,r\}}.
	\end{align*}
	The system \eqref{DTsys} is said to be of \emph{Darboux-Tsarev type} if the torsionless connection $\nabla$ uniquely defined by the conditions
	\[d_{\nabla}(X\circ)=0,\qquad\nabla e=0,\]
	satisfies the condition \eqref{rc-intri} and, in canonical coordinates, the components   $X^{i(\alpha)}$ of the vector field defining the system depend only on $\mathcal{U}^i$.
\end{definition}

\hspace{-1em}In the case of systems of Darboux-Tsarev type, the  linear system 
\[d_{\nabla}(Y\circ)=0\]
defining the symmetries
\begin{equation}\label{DTsyssym}
	u^i_\tau=(Y\circ u_x)^i=c^i_{jk}Y^ju^k_x,\qquad i\inw,
\end{equation}
of \eqref{DTsys} has a special property which is the main  motivation for their definition. Let us first observe that, in David-Hertling canonical  coordinates, such system reads
\begin{align}
	&\partial_{k(\alpha)}Y^{i(\alpha)}=\partial_{1(\alpha)}Y^{(i-k+1)(\alpha)}+\overset{m_\alpha}{\underset{s=2}{\sum}}\Big(\Gamma^{(i-k+1)(\alpha)}_{1(\alpha)s(\alpha)}-\Gamma^{i(\alpha)}_{k(\alpha)s(\alpha)}\Big)Y^{s(\alpha)},\quad k\in\{2,\dots,m_\alpha\},\label{EqSymm_sameblock_updated}\\
	&\partial_{j(\beta)}Y^{i(\alpha)}=-\overset{i}{\underset{s=1}{\sum}}\Gamma^{i(\alpha)}_{j(\beta)s(\alpha)}Y^{s(\alpha)}-\overset{m_\beta-j+1}{\underset{s=1}{\sum}}\Gamma^{i(\alpha)}_{j(\beta)s(\beta)}Y^{s(\beta)},\qquad\alpha\neq\beta,\,\,j\in\{1,\dots,m_\beta\}.\label{EqSymm_distinctblocks_updated}
\end{align}
For all $m\in\{1,\dots,\underset{{\sigma\in\{1,\dots,r\}}}{\max}m_\sigma\}$,
let us denote $\mathcal{A}_m:=\{\alpha\in\{1,\dots,r\}\,|\,m_\alpha\geq m\}$ and let us call $\mathcal{S}_i$ the subsystem of  \eqref{EqSymm_sameblock_updated}-\eqref{EqSymm_distinctblocks_updated} for $\{Y^{i(\alpha)}\}_{\alpha\in\mathcal{A}_i}$. We have the following result, proved in \cite{LPVGhodograph}.
\begin{theorem}\label{Theorem_Completeness_multiblock}
	In order for the assumptions of Darboux's theorem to be valid for the subsystems $\mathcal{S}_1,\mathcal{S}_2,\dots$, for all $\alpha\in\{1,\dots,r\}$ and $i\in\{1,\dots,m_\alpha\}$, $X^{i(\alpha)}$ must depend on $\mathcal{U}^i$ only. Furthermore, if the system is  of Darboux-Tsarev type then the set of symmetries is complete  and the generalised hodograph formula provides the general solution of the system \eqref{DTsys}.
\end{theorem}

\subsection{Useful lemmas}
In this subsection we give some additional properties of the connection associated to Darboux-Tsarev systems, which will be useful in the subsequent parts. We here assume that $\alpha \neq \beta$, unless specified otherwise. Let's start with some lemmas already present in \cite{LPVGhodograph}.

\begin{lemma}[Remark 4.12 in \cite{LPVGhodograph}]
   If the system  \eqref{DTsys}  is of Darboux-Tsarev type then 
   \begin{equation}
       \Gamma^{i(\alpha)}_{j(\beta)k(\gamma)} = 0, \qquad \text{for } \, i<\text{max}\{j,k\}.
   \end{equation}
   \label{lemma:old4}
\end{lemma}
\begin{remark}
    Note that by combining Lemmas \ref{lemma:old2b} and \ref{lemma:old4} we have that, if the system \eqref{DTsys} is of Darboux-Tsarev type then
    \begin{equation}
        \Gamma^{i(\alpha)}_{j(\alpha)k(\alpha)} = 0, \quad \text{ for } i-j-k \leq -3.
         \label{eq:Rmk:old}
    \end{equation}
    \label{rmk:old}
\end{remark}
    
\begin{lemma}[Lemma 4.3 in \cite{LPVGhodograph}]
    If the system \eqref{DTsys} is of Darboux-Tsarev type then $  \Gamma^{i(\alpha)}_{j(\beta)k(\gamma)}$       depends on $\mathcal{U}^i$ only. 
    \label{lemma:old5}
\end{lemma}

\hspace{-1em}Let us now derive some additional properties of the Christoffel symbols associated to Darboux-Tsarev systems. Writing the condition \eqref{shc-intro} and the symmetry of $\nabla c$ in David-Hertling coordinates using the double-index notation gives
\begin{subequations}
	\begin{equation}
		R^{(j-k+1)(\tau)}_{l(\alpha)m(\beta)i(\gamma)} \delta^{\rho}_{\tau} +   R^{(j-m+1)(\tau)}_{l(\alpha)i(\gamma)k(\rho)}\delta^{\beta}_{\tau} +  R^{(j-i+1)(\tau)}_{l(\alpha)k(\rho)m(\beta)} \delta^{\gamma}_{\tau} = 0,
		\label{eq:3RCmultiinR}
	\end{equation}
	\begin{equation}
		\Gamma^{k(\gamma)}_{l(\rho) m(\sigma)}\delta^{\sigma}_{\alpha}\delta^{\sigma}_{\beta}\delta^{m}_{i+j-1} - \Gamma^{m(\sigma)}_{l(\rho) i(\alpha)}\delta^{\gamma}_{\sigma}\delta^{\gamma}_{\beta}\delta^{k}_{m+j-1} - \Gamma^{k(\gamma)}_{j(\beta) m(\sigma)}\delta^{\sigma}_{\alpha}\delta^{\sigma}_{\rho}\delta^{m}_{i+l-1}+\Gamma^{m(\sigma)}_{j(\beta) i(\alpha)}\delta^{\gamma}_{\sigma}\delta^{\gamma}_{\rho}\delta^{k}_{m+l-1}  = 0,
		\label{eq:nablacsymmulti}
	\end{equation}
\end{subequations}
respectively.  
\newline
\newline
We can then improve upon Lemma \ref{rmk:old} as follows. 

\begin{lemma}
	For any $\rho \in \{1, \dots, r\}$,
	\begin{equation}
		\Gamma^{i(\alpha)}_{j(\rho)k(\beta)} = 0, \qquad \text{for } i-j-k<-1.
	\end{equation}
	\label{lemma:newold}
\end{lemma}

\begin{proof} 
	Notice that if $\alpha \neq \rho \neq \beta$ the lemma holds by \eqref{Chr_a}. Moreover, if $i<j$ or $i<k$ it holds by Lemma \ref{lemma:old4}. Thus we may assume that $i \geq \max\{j,k\}$. By Lemma \ref{lemma:old1} we get \begin{equation}
		\Gamma^{i(\alpha)}_{j(\rho)k(\beta)} = \begin{cases}
			\Gamma^{(i-j+1)(\alpha)}_{1(\alpha)k(\beta)}, & \text{if } \rho = \alpha \text{ and } i \geq j,\vspace{0.2em}\\
			\Gamma^{i(\alpha)}_{(j+k-1)(\beta)1(\beta)}, & \text{if } \rho = \beta \text{ and } j+k \leq m_{\beta}+1,\\
			0, & \text{ otherwise.}
		\end{cases}
	\end{equation}
	However, by Lemma \ref{lemma:old4} all the cases above vanish if $i-j-k<-1$. 
\end{proof}

\begin{lemma}
	Let $j \in \{1, \dots, m_{\alpha}\}$, and $k \in \{1, \dots, m_{\beta}\}$. Then, 
	\begin{equation}
		\partial_{k(\beta)} \Gamma^{j(\alpha)}_{2(\alpha)j(\alpha)} = 0.
	\end{equation}
	\label{lemma:1multidiff}
\end{lemma}
\begin{proof}
	Considering \eqref{eq:3RCmultiinR} with all indices belonging to the same block (which we denote by $\alpha$) except $m$ which belongs to $\beta$, with $\beta \neq \alpha$, letting $k=1$, $i=j$, $m \rightarrow k$, $l=2$ and using Lemma \ref{lemma:old4} gives
	\begin{equation}
		\partial_{k(\beta)}\Gamma^{j(\alpha)}_{2(\alpha)j(\alpha)} - \partial_{j(\alpha)}\Gamma^{j(\alpha)}_{2(\alpha)k(\beta)}+ \underset{f}{\underbrace{\Gamma^{j(\alpha)}_{k(\beta)s(\sigma)} \Gamma^{s(\sigma)}_{2(\alpha)j(\alpha)} - \Gamma^{j(\alpha)}_{j(\alpha) s(\sigma)} \Gamma^{s(\sigma)}_{2(\alpha)k(\beta)}}}. 
	\end{equation}
	Notice that the second term is zero by Lemmas \ref{lemma:old1}, \ref{lemma:old5}. Hence, if $f = 0$ we are done. By Remark \ref{rmk:old} and Lemma \ref{lemma:newold}, the only non-vanishing contribution for the second $s$-sum occurs when $s=2$, $k=1$, and $\sigma = \alpha$. Similarly, for the first $s$-sum, we obtain zero unless $s=j$, $k=1$, and $\sigma = \alpha$. However, then they cancel with each other after using Lemma \ref{lemma:old1}. 
\end{proof}
\begin{lemma}
	\begin{equation}
		\partial_{m(\alpha)} \Gamma^{j(\alpha)}_{l(\alpha)i(\alpha)} = 0, \qquad \text{for } m > j-l+2. 
	\end{equation}
	\label{lemma:diffsameblockzero}
\end{lemma}
\begin{proof}
	Consider \eqref{eq:3RCmultiinR} with all indices belonging to the same block (which we denote by $\alpha$) and let $k=1$ and $m>j-l+2$. Then, by Lemmas \ref{lemma:old5}, and \ref{lemma:newold} together with Remark \ref{eq:Rmk:old}, the result follows. 
\end{proof} 

\begin{lemma}
	For any $\rho \in \{1, \dots, r\}$,
	\begin{equation}
		\partial_{m(\rho)} \Gamma^{i(\alpha)}_{j(\alpha)k(\beta)} = 0, \qquad \forall \, m > i-j+1.
	\end{equation}
	\label{lemma:diffdiffblocks}
\end{lemma}
\begin{proof}
	If $i<j$ the lemma holds trivially  by Lemma \ref{lemma:old4}. Thus, let $i\geq j$. Applying  \eqref{eq:lemmaold1diag} gives
	\begin{equation}
		\partial_{m(\rho)} \Gamma^{i(\alpha)}_{j(\alpha)k(\beta)} = \partial_{m(\rho)}\Gamma^{(i-j+1)(\alpha)}_{1(\alpha) k(\beta)}.
	\end{equation}
	Using Lemma \ref{lemma:old5} concludes the proof.
\end{proof}

\begin{lemma}
	Let $\beta \neq \alpha \neq \gamma$. Then, 
	\begin{equation}
		\partial_{m(\beta)} \Gamma^{j(\alpha)}_{l(\alpha)i(\gamma)} = \partial_{i(\gamma)}\Gamma^{j(\beta)}_{l(\alpha) m(\beta)}.
	\end{equation}
	\label{lemma:diffswapgen}
\end{lemma}
\begin{proof}
	Considering \eqref{eq:3RCmultiinR} with $\alpha = \rho = \tau$, $\beta \neq \tau \neq \gamma$, $k=1$ gives 
	\begin{equation*}
		\partial_{m(\beta)} \Gamma^{j(\alpha)}_{l(\alpha)i(\gamma)} - \partial_{i(\gamma)}\Gamma^{j(\beta)}_{l(\alpha) m(\beta)} + \Gamma^{j(\alpha)}_{m(\beta)t(\sigma)} \Gamma^{t(\sigma)}_{l(\alpha)i(\gamma)} - \Gamma^{j(\alpha)}_{i(\gamma) t(\sigma)} \Gamma^{t(\sigma)}_{m(\beta)l(\alpha)} = 0.
	\end{equation*}
	Note that by \eqref{Chr_a}, we must have $\sigma = \alpha$. Using Lemma \ref{lemma:old1} the final two terms cancel after a relabelling of the summed index (for instance replacing $t$  with $j+l-t$ in the first sum), which concludes the proof. 
\end{proof}
\begin{lemma}
	Let $j \geq 3$, then
	\begin{equation}
		\partial_{j(\alpha)} \Gamma^{k(\alpha)}_{k(\beta)1(\alpha)} = 0.
	\end{equation}
	\label{lemma:newnew}
\end{lemma}
\begin{proof}
	Consider \eqref{eq:3RCmultiinR} with all indices except $m$ belonging to the same block (which we denote by $\alpha$), while $m = m(\beta)$ with $\beta \neq \alpha$. Moreover, let $i = 2, l = 1$, and $(k,  j, m) \rightarrow (j, k+1, k)$, then we get
	\begin{equation*}
		\underset{= \, 0 \, \text{ by Lemma \ref{lemma:old1}}}{\underbrace{\partial_{k(\beta)}\left(\Gamma^{(k-j+2)(\alpha)}_{1(\alpha)2(\alpha)} -\Gamma^{k(\alpha)}_{1(\alpha)j(\alpha)}\right)}} -\underset{= \, 0 \, \text{ by Lemma \ref{lemma:newold}}}{\underbrace{\partial_{2(\alpha)}\Gamma^{(k-j+2)(\alpha)}_{1(\alpha)k(\beta)}}} + \partial_{j(\alpha)}\Gamma^{k(\alpha)}_{1(\alpha)k(\beta)} + f = 0,
	\end{equation*}
	where 
	\begin{equation*}
		f = \Gamma^{(k-j+2)(\alpha)}_{k(\beta)t(\sigma)}\Gamma^{t(\sigma)}_{1(\alpha)2(\alpha)} - \Gamma^{(k-j+2)(\alpha)}_{2(\alpha)t(\sigma)}\Gamma^{t(\sigma)}_{1(\alpha)k(\beta)} + \Gamma^{k(\alpha)}_{j(\alpha)t(\sigma)}\Gamma^{t(\sigma)}_{1(\alpha)k(\beta)} - \Gamma^{k(\alpha)}_{k(\beta)t(\sigma)}\Gamma^{t(\sigma)}_{1(\alpha)j(\alpha)}.
	\end{equation*}
	Note that the second factor in the first term implies, by Lemma \ref{lemma:old4} that $t>k$, but then $k-j+2-t-2 \leq -j \leq -3$, and so the term is zero by Remark \ref{rmk:old}. The same logic implies that $f$ vanishes, which concludes the proof. 
\end{proof}

\subsection{Symmetries for systems of Darboux-Tsarev type} 
In \cite{LPVGhodograph} it was proved that the symmetries of a system \eqref{DTsys} of Darboux-Tsarev type 
are parametrised by $n$ arbitrary functions of a single variable. For instance,  in the  case 
of  a single Jordan block of size $3$, where $X^3=F_1(u^1,u^2)u^3+F_2(u^1,u^2)$,
one gets  
\begin{eqnarray*}
	&&Y^1=Y^1(u^1),\quad Y^2 = (-F'_3\int e^g\,du^2+F_4)e^{-g}\\
	&&Y^3=\frac{e^{-g}u^3(F_1-(X^1)')}{X^2}\left(-F'_3\int e^g\,du^2+F_4\right)+e^{-h}\left(F_5+\int\frac{ke^{h-g}}{(X^2)^2}\,du^2\right)-F'_3u^3
\end{eqnarray*}
where $F_3,F_4,F_5$ are arbitrary functions of $u^1$,
\begin{eqnarray*}
	g=\int\left(\frac{(X^1)'-\frac{\partial X^2}{\partial u^2}}{X^2}\right)\,du^2,\qquad
	h=\int\left(\frac{(X^1)'-F_1}{X^2}\right)\,du^2,
\end{eqnarray*}
and, finally,
\begin{eqnarray*}
	k&=&\left[-(X^2)^2f+X^2\left(\frac{\partial X^2}{\partial u^1}-\frac{\partial F_2}{\partial u^2}\right)+F_2(F_1-(X^1)')\right]\left(F'_3\int e^g\,du^2-F_4\right)\\&&
	+(X^2)^2\left(F'_4+\int e^{g}(fF'_3-F''_3)\,du^2\right),
\end{eqnarray*}
with
\begin{eqnarray*}
	f&=&\int\frac{1}{(X^2)^2}\left(X^2\frac{\partial^2 X^2}{\partial u^1\partial u^2}-X^2(X^1)''+\frac{\partial X^2}{\partial u^1}(X^1)'-\frac{\partial X^2}{\partial u^1}
	\frac{\partial X^2}{\partial u^2}\right)\,du^2.
\end{eqnarray*}

\subsection{Densities of conservation laws for systems of Darboux-Tsarev type} 
In the case of  systems of Darboux-Tsarev type  \eqref{dcl3} can be reduced to a finite sequence of closed subsystems satisfying Darboux's conditions. 
\subsubsection{The case of a single Jordan block}
In David-Hertling canonical coordinates  the system  \eqref{dcl3} reads  
\begin{equation}
	\partial_j\omega_i=\Gamma^s_{ji}\omega_s+\partial_1\omega_{j+i-1}.
\end{equation}
For $j=1$, the system reduces to a trivial identity since  $\Gamma^s_{1i}=0$ (by the flatness of the unit vector field).
The first subsystem of the sequence is the system for $\omega_n$:
\begin{equation}
	\partial_j\omega_n=\Gamma^s_{jn}\omega_s,\qquad j>1,
\end{equation}
which already has Darboux's form, as when $j>1$  we have $j+n-1>n$. Moreover, due to the  condition $\Gamma^s_{jn}=0$ for $s\ge n+j-2$, the  only  non-vanishing contribution  in the right-hand side of the system arises for $j=2$:
\begin{eqnarray*}
	\partial_2\omega_n&=&\Gamma^n_{2n}\omega_n,\\
	\partial_j\omega_n&=&0,\qquad j>2.
\end{eqnarray*}
The second closed subsystem is the system for $\omega_{n-1}$ which reads
\begin{eqnarray*}
	\partial_2\omega_{n-1}&=&\Gamma^{n-1}_{2,n-1}\omega_{n-1}+\Gamma^n_{2,n-1}\omega_n+\partial_1\omega_{n},\\
	\partial_3\omega_{n-1}&=&\Gamma^{n}_{3,n-1}\omega_{n},\\
	\partial_j\omega_{n-1}&=&0,\qquad j>3.
\end{eqnarray*}
In addition to  the unknown function $\omega_{n-1}$,  the right-hand sides of the system  also contain 
the function $\omega_n$ which was obtained in the previous step. In a similar way, the right-hand sides of the remaining $n-2$ subsystems for $\omega_{i}$, $i\in\{n-2,n-3,\dots,1\}$, contain, in addition to the unknown $\omega_i$, the functions $\omega_{i+1},...,\omega_n$  obtained in the previous steps. At each  step, the compatibility of the subsystem follows from Proposition \ref{compsysdcl}.
\newline
\newline
As a  consequence of Darboux's theorem, we have proved the following result.

\begin{theorem}
	In the  case of single Jordan block of size $n$, the general solution of the system \eqref{dcl3} for densities of conservation laws depends on $n$ arbitrary functions of $u^1$. 
\end{theorem}

\subsubsection{Arbitrary Jordan block structure}
The general case can be  treated in a similar way as the case of a single Jordan block. Let the matrix defining the system contain $r$  blocks of sizes $m_1, \dots, m_r$, respectively. Then, without loss of generality, we can assume that 
\[m_1\ge m_2\ge\cdots\ge m_r.\]
Firstly,  observe that in David-Hertling  coordinates the system \eqref{dcl3} reads 
\begin{eqnarray}
	\label{sb1}
	\partial_{j(\alpha)}\omega_{i(\alpha)}&=&\Gamma_{j(\alpha)i(\alpha)}^{s(\sigma)}\omega_{s(\sigma)}+\sum_{\beta=1}^r\partial_{1(\beta)}\omega_{(i+j-1)(\alpha)},\\
	\label{sb2}
	\partial_{j(\beta)}\omega_{i(\alpha)}&=&\Gamma_{j(\beta)i(\alpha)}^{s(\sigma)}\omega_{s(\sigma)},\qquad\beta\ne\alpha.
\end{eqnarray}
When $j=1$, taking  into  account \eqref{sb2},   \eqref{sb1} reduces  to an identity. Indeed,
the equation
\[\partial_{1(\alpha)}\omega_{i(\alpha)}=\Gamma_{1(\alpha)i(\alpha)}^{s(\sigma)}\omega_{s(\sigma)}+\sum_{\beta=1}^r\partial_{1(\beta)}\omega_{(i+j-1)(\alpha)}\]
can be written as
\[\Gamma_{1(\alpha)i(\alpha)}^{s(\sigma)}\omega_{s(\sigma)}+\sum_{\beta\ne\alpha}\partial_{1(\beta)}\omega_{i(\alpha)}=0\]
or, using \eqref{sb2}, as
\[e^{j(\beta)}\Gamma_{j(\beta)i(\alpha)}^{s(\sigma)}\omega_{s(\sigma)}=0,\]
which is an identity due to the  flatness of the unit vector field.
\newline
\newline
The first subsystem of the sequence is the subsystem for the unknowns $\omega_{m_1(1)}, \dots, \omega_{m_1(p_1)}$, where
\[
p_1=\max\big\{\alpha\in\{1,\dots,r\}\,|\,m_\alpha=m_1\big\}.
\]
When $j=1$,  only  the second part of the system is non-trivial and we have 
\[\partial_{1(\beta)}\omega_{m_1(\alpha)}=\Gamma_{1(\beta)m_1(\alpha)}^{s(\sigma)}\omega_{s(\sigma)},\qquad\alpha\in\{1,...,p_1\},\,\beta\ne\alpha.\] 
Using Lemma \ref{lemma:newold}, we obtain
\begin{eqnarray*}
\partial_{1(\beta)}\omega_{m_1(\alpha)}&=&\Gamma_{1(\beta)m_1(\alpha)}^{m_1(\alpha)}\omega_{m_1(\alpha)}+\Gamma_{1(\beta)m_1(\alpha)}^{m_1(\beta)}\omega_{m_1(\beta)}\mathds{1}_{\{\beta\le  p_1\}},\qquad\alpha\in\{1,...,p_1\},\,\beta\ne\alpha.
\end{eqnarray*}
Let us consider the case when $j=2$. THen we have,
\begin{eqnarray}
	\partial_{2(\alpha)}\omega_{m_1(\alpha)}&=&\Gamma_{2(\alpha)m_1(\alpha)}^{s(\sigma)}\omega_{s(\sigma)},\qquad\alpha\in\{1,...,p_1\},\\
	\partial_{2(\beta)}\omega_{m_1(\alpha)}&=&\Gamma_{2(\beta)m_1(\alpha)}^{s(\sigma)}\omega_{s(\sigma)},\qquad\alpha\in\{1,...,p_1\},\,\beta\ne\alpha.
\end{eqnarray}  
Using Remark \ref{rmk:old} and Lemma \ref{lemma:newold}, we obtain 
\begin{eqnarray}
	\partial_{2(\alpha)}\omega_{m_1(\alpha)}&=&\Gamma_{2(\alpha)m_1(\alpha)}^{m_1(\alpha)}\omega_{m_1(\alpha)},\qquad\alpha\in\{1,...,p_1\},\\
	\partial_{2(\beta)}\omega_{m_1(\alpha)}&=&0,\qquad\qquad\qquad\qquad\,\alpha\in\{1,...,p_1\},\,\beta\ne\alpha.
\end{eqnarray}  
Similarly, the remaining equations ($j>2$) for $\omega_{m_1(\alpha)}$ are
\begin{eqnarray}
	\partial_{j(\alpha)}\omega_{m_1(\alpha)}&=&0,\qquad\alpha\in\{1,...,p_1\},\\
	\partial_{j(\beta)}\omega_{m_1(\alpha)}&=&0,\qquad\alpha\in\{1,...,p_1\},\,\beta\ne\alpha.
\end{eqnarray}  
The second subsystem of the sequence is the subsystem for the unknowns 
\[\omega_{(m_1-1)(1)},...,\omega_{(m_1-1)(p_2)},\] 
where
\[
p_2=\max\big\{\alpha\in\{1,\dots,r\}\,|\,m_\alpha\geq m_1-1\big\}.
\]
Clearly $p_2\ge p_1$. When $j=1$, only  the second part of the system is non-trivial and we have 
\begin{eqnarray*}
	\partial_{1(\beta)}\omega_{(m_1-1)(\alpha)}&=&\Gamma_{1(\beta)(m_1-1)(\alpha)}^{(m_1-1)(\alpha)}\omega_{(m_1-1)(\alpha)}+\Gamma_{1(\beta)(m_1-1)(\alpha)}^{m_1(\alpha)}\omega_{m_1(\alpha)}\mathds{1}_{\{\alpha\le  p_1\}}+\\
	&&\Gamma_{1(\beta)(m_1-1)(\alpha)}^{(m_1-1)(\beta)}\omega_{(m_1-1)(\beta)}\mathds{1}_{\{\beta\le  p_2\}}+\Gamma_{1(\beta)(m_1-1)(\alpha)}^{m_1(\beta)}\omega_{m_1(\beta)}\mathds{1}_{\{\beta\le  p_1\}},
\end{eqnarray*}
where $\alpha\in\{1,...,p_2\},\,\beta\ne\alpha.$
\noindent
Let us consider the case where $j=2$. Taking into account the vanishing conditions on the Christoffel symbols we get
\begin{eqnarray*}
	\partial_{2(\alpha)}\omega_{(m_1-1)(\alpha)}&=&\Gamma_{2(\alpha)(m_1-1)(\alpha)}^{(m_1-1)(\alpha)}\omega_{(m_1-1)(\alpha)}+\Gamma_{2(\alpha)(m_1-1)(\alpha)}^{m_1(\alpha)}\omega_{m_1(\alpha)}\mathds{1}_{\{\alpha\le  p_1\}}\\&&+\sum_{\sigma=1}^{p_1}\Gamma_{2(\beta)(m_1-1)(\alpha)}^{m_1(\sigma)}\omega_{m_1(\sigma)},\\
	\partial_{2(\beta)}\omega_{(m_1-1)(\alpha)}&=&\Gamma_{2(\beta)(m_1-1)(\alpha)}^{m_1(\alpha)}\omega_{m_1(\alpha)}\mathds{1}_{\{\alpha\le  p_1\}}+\Gamma_{2(\beta)(m_1-1)(\alpha)}^{m_1(\beta)}\omega_{m_1(\beta)}\mathds{1}_{\{\beta\le  p_1\}},
\end{eqnarray*}
where $\alpha\in\{1,...,p_2\},\,\beta\ne\alpha.$ Similarly, for $j=3$ we get
\begin{eqnarray}
	\partial_{3(\alpha)}\omega_{(m_1-1)(\alpha)}&=&\Gamma_{3(\alpha)(m_1-1)(\alpha)}^{m_1(\alpha)}\omega_{m_1(\alpha)}\mathds{1}_{\{\alpha\le  p_1\}},\qquad\alpha\in\{1,...,p_1\},\\
	\partial_{3(\beta)}\omega_{(m_1-1)(\alpha)}&=&0,\qquad\alpha\in\{1,...,p_1\},\,\beta\ne\alpha,
\end{eqnarray}  
with the remaining equations ($j>3$) for $\omega_{(m_1-1)(\alpha)}$ being
\begin{eqnarray}
	\partial_{j(\alpha)}\omega_{(m_1-1)(\alpha)}&=&0,\qquad\alpha\in\{1,...,p_1\},\\
	\partial_{j(\beta)}\omega_{(m_1-1)(\alpha)}&=&0,\qquad\alpha\in\{1,...,p_1\},\,\beta\ne\alpha.
\end{eqnarray}  
In addition to  the unknowns  $(\omega_{(m_1-1)(1)},...,\omega_{(m_1-1)(p_2)})$, 
the right-hand sides of the system  also contain  
the functions $(\omega_{m_1(1)},...,\omega_{m_1(p_1)})$ obtained at the previous step. Similarly, the right-hand sides of the remaining $m_1-2$ subsystems for $\omega_{i(\alpha)}$, $i\in\{m_1-2,m_1-3,\dots,1\}$, contain, in addition to  the unknown $\omega_{i(\alpha)}$,   the functions $\omega_{(i+1)(\alpha)},...,\omega_{m_1(\alpha)}$ obtained at the previous steps. At each  step, the compatibility of the subsystem follows from Proposition \ref{compsysdcl}.
\newline
\newline
As a  consequence of Darboux's theorem, we have proved the following result.

\begin{theorem}
	In the  case of $r$ Jordan blocks of sizes $m_1,\dots,m_r$, the general solution of the system \eqref{dcl3} for densities of conservation laws depends on $n$ arbitrary functions. More precisely, each  block $\alpha$ contributes with $m_\alpha$ arbitrary functions of the main variable of the block, $u^{1(\alpha)}$ . 
\end{theorem}

\section{Dubrovin-Novikov-Tsarev system}
\hspace{-1em}From now on we will assume that the functions $u^i(x)$ satisfy vanishing boundary conditions (at infinity) or periodic boundary conditions.
\begin{definition}
Functionals of hydrodynamic type are functional  of the form
\[H[u]=\int h(u)\,dx.\]
The domain of integration is the real line for  vanishing boundary conditions and the period in the case of periodic boundary conditions.
\end{definition}

\begin{definition}\cite{DN84}
A Poisson bracket of hydrodynamic type is a Poisson bracket  of the form
\[\{F,G\}
=\iint
\f{\delta F}{\delta u^i(x)}\{u^i(x),u^j(y)\}\f{\delta G}{\delta u^j(y)}\,dxdy,\]
where
\[\{u^i(x),u^j(y)\}=g^{ij}({\bs u}(x))\delta'(x-y)+b^{ij}_k({\bs u}(x))u^k_x\delta(x-y)\]
and $g^{ij},b^{ij}_k$ are smooth functions with $i, j, k \in\{1,\dots,n\}$.
\end{definition}

\hspace{-1em}From the definition of the Dirac delta function and its derivatives, we have 
\begin{eqnarray*}
\int f(y)\delta^{(p)}(x-y)dy&=& f^{(p)}(x).
\end{eqnarray*}
It follows that
\[\{F,G\}
=\int
\f{\delta F}{\delta u^i}P^{ij}\f{\delta G}{\delta u^j}\,dx,\]
where 
\[P^{ij}=g^{ij}({\bs u})\d_x+b^{ij}_k({\bs u})u^k_x.\]
From now on we will assume that $g \coloneqq (g^{ij})$ constitute  a non-degenerate matrix.

\begin{definition}
A hamiltonian system of hydrodynamic type is a  system of the form
\[u^i_t=P^{ij}\frac{\delta H}{\delta u^j}=\left(g^{ij}\f{\d^2 h}{\d u^j\d u^k}+b^{ij}_k\f{\d h}{\d u^j}\right)u^k_x,\qquad i=1,\dots,n\]
where $H=\int h(u)\,dx$ is a functional of hydrodynamic type.
\end{definition}

\begin{theorem}\cite{DN84,DN}
In the non-degenerate case  $\text{det}\,(g)\ne 0$,  the formula
\[\{u^i(x),u^j(y)\}=g^{ij}({\bs u})\delta'(x-y)+b^{ij}_k({\bs u})u^k_x\delta(x-y)\]
 defines a Poisson bracket if and only if $g$ is symmetric, the linear connection defined by the Christoffel symbols $\tilde\Gamma^j_{hk}=-g_{hi}b^{ij}_k$ 
coincides with the Levi-Civita connection of the metric $g$ and the associated Riemann tensor vanishes.
\end{theorem}

\hspace{-1em}It is immediate to check that, in the non degenerate case, a Hamiltonian system of hydrodynamic type has the form
\[u^i_t=g^{ij}\tilde\nabla_j\tilde\nabla_k h\,u^k_x.\]
This leads to the following result.

\begin{proposition}\cite{ts86,DN}
A system of hydrodynamic type 
\[u^i_t=V^i_j({\bs u})u^j_x,\]
is Hamiltonian with respect to a Dubrovin-Novikov structure if and only if there exists a flat metric $(g_{ij})$ such that
\begin{eqnarray}
\label{DN1}
g_{ik}V^k_j&=&g_{jk}V^k_i,\\
\label{DN2}
\tilde\nabla_iV^k_j&=&\tilde\nabla_jV^k_i,
\end{eqnarray}
where $\nabla$ is the covariant derivative defined by the associated Levi-Civita connection.
\end{proposition}

\subsection{Hamiltonian formalism for diagonal systems}
Let us consider a diagonal system
\[V^i_j=v^i\delta^i_j\]
and suppose that $v^i\ne v^j,$ for $i\ne j$. In this case the condition \eqref{DN1} 
states that the metric must be diagonal in the Riemann invariants, while the condition
 \eqref{DN2} reduces  to
\begin{equation}\label{sys_met}
\partial_j g_{ii}=2\f{\d_jv^i}{v^j-v^i}g_{ii},\qquad i\ne j.
\end{equation}
Expanding the compatibility conditions 
\[\d_k\left(2\f{\d_jv^i}{v^j-v^i}g_{ii}\right)=\d_j\left(2\f{\d_kv^i}{v^k-v^i}g_{ii}\right),\qquad i\ne j\ne k\ne i,\]
and using the equations of the system, we obtain Tsarev's integrability condition.

In other words, if a system of hydrodynamic type is \emph{semi-Hamiltonian}, the general  solution of the system
\begin{equation}\label{metrics}
\partial_j g_{ii}=2\f{\d_jv^i}{v^j-v^i}g_{ii},\qquad i\ne j
\end{equation}
depends on $n$ arbitrary functions of a single variable.
\newline
\newline
Indeed, given a solution $(g_{11} , \dots, g_{nn})$ of the system, it is immediate to check that 
\[\left(\varphi_1(r^1)g_{11} , \dots, \varphi_n(r^n)g_{nn}\right)\] 
is a solution as well for any choice of the functions $\varphi_1(r^1) , \dots, \varphi_n(r^n)$.
Summarising, for any semi-Hamiltonian system there exists a family of solutions
 of the system \eqref{metrics}.  The existence of flat  solutions, however, is not guaranteed. As a consequence, only a subset of semi-Hamiltonian systems admits a Hamiltonian structure of Dubrovin-Novikov type.
\newline
\newline
Non-flat solutions of the linear system \eqref{metrics} are possibly  related to non-local Hamiltonian structures of hydrodynamic type.  
Let us consider  
\begin{eqnarray*}
\{u^i(x),u^j(y)\}&=&g^{ij}\delta'(x-y)-g^{is}\tilde\Gamma^j_{sk}u^k_x\delta(x-y)\\
&&+\sum_{\alpha=1}^M\varepsilon_{\alpha}\left(W_{(\alpha)}\right)^i_ku^k_x
\nu(x-y)\left(W_{(\alpha)}\right)^j_hu^h_x,
\end{eqnarray*}
where $\varepsilon_{\alpha}=\pm 1$ and 
\begin{eqnarray*}
\nu(x-y)&=&-\f{1}{2},\quad\text{for}\,\,x<0,\\
\nu(x-y)&=&\f{1}{2},\quad\text{for}\,\,x>0.
\end{eqnarray*}
Conditions ensuring that the above formula defines a Poisson bracket were found by Ferapontov in \cite{F} (see also \cite{CLV}). It  turns out  that $(g^{ij})$ defines a metric and that $\tilde\Gamma^j_{sk}$ are the Christoffel symbols of the associated Levi-Civita connection, as in the local  case. Moreover, $\{W_{(\alpha)}\}$ are $(1,1)$-tensor fields  satisfying the \emph{Gauss-Peterson-Mainardi-Codazzi equations}
\begin{subequations}\label{gmc}
\begin{gather} 
\left[W_{(\alpha)},W_{(\alpha')}\right]=0,\label{GMC1}\\
g_{ik}(W_{(\alpha)})^k_j=g_{jk}(W_{(\alpha)})^k_i,\label{GMC2}\\
\tilde\nabla_k(W_{(\alpha)})^i_j=\tilde\nabla_j(W_{(\alpha)})^i_k,\label{GMC3}\\
\tilde{R}^{ij}_{kh}=\sum_{\alpha=1}^M\varepsilon_{\alpha} \left\{\left(W_{(\alpha)}\right)^i_k\left(W_{(\alpha)}\right)^j_h
-\left(W_{(\alpha)}\right)^j_k\left(W_{(\alpha)}\right)^i_h\right\},\label{GMC4}
\end{gather}
\end{subequations}
where $\tilde{R}^{ij}_{kh}=g^{is}\tilde{R}^{j}_{skh}$. In the case of diagonal systems of hydrodynamic type, the condition \eqref{GMC4} becomes
 \[\tilde{R}^{ij}_{ij}=\sum_{\alpha=1}^M\varepsilon_{\alpha} w^i_{(\alpha)}w^j_{(\alpha)},\]
 where the  functions $w^i_{(\alpha)}$ are the characteristic  velocities  of the symmetries of the  system.  Ferapontov conjectured that any solution of \eqref{metrics} admits such a  quadratic expansion in terms of the symmetries.  
 
\subsection{Hamiltonian formalism for non-diagonalisable systems}
We have seen that  the metric $g$ defining a Hamiltonian structure of hydrodynamic type
  must satisfy the conditions \eqref{DN1} and \eqref{DN2}. For systems defined by a $(1,1)$-tensor field of the form $V^i_j=c^i_{jk}X^k$ associated with an F-manifold  with compatible connection,  the second condition \eqref{DN2} can be written in other, useful, equivalent forms as stated in Theorem \ref{EquivalentDN}. In order to prove this we need a preliminary lemma (similar to Lemma 2.2 in \cite{LPVGhodograph})
\begin{lemma}\label{MToeplitz_implied}
A matrix $M$ satisfying 
\begin{equation}\label{cond3}
c^s_{ji}M_{sk}=c^s_{ki}M_{sj},\qquad i,j,k\inw,
\end{equation}	
can be  written as
\begin{equation}\label{MToeplitz_gen}
M_{ij}=c^s_{ij}\theta_s,\qquad i,j\inw.
\end{equation}
\end{lemma}  
\begin{proof}
	Firstly,  observe that if $M$ is of the form $M_{ij}=c^s_{ij}\theta_s$ then  condition
	\eqref{cond3} is satisfied. Multiplying both sides of  \eqref{cond3}  by $e^j$ and taking the sum over $j$, we get
	\[M_{ik}=c^s_{ki}\theta_s\]
	with $M_{sj}e^j$.
\end{proof}

\begin{theorem}\label{EquivalentDN}
Let  us consider a regular F-manifold with compatible connection and let $g$ be a metric satisfying condition \eqref{DN1}. Assuming  that  the linear system
\[\nabla_kV^i_j-\nabla_jV^i_k=c^i_{js}\nabla_kX^s-c^i_{ks}\nabla_jX^s=0\]
admits $n$ independent solution $X_{(1)} , \dots, X_{(n)}$, the condition \eqref{DN2} is equivalent to one of the following conditions:

\begin{equation}\label{Ch_sym}
\tilde\Gamma^l_{ij}-\Gamma^l_{ij}=-g^{lk}c^m_{ij}(\tilde{\nabla}_k\theta_m-\mathcal{L}_eg_{km})=\frac{1}{2}g^{lk}c^m_{ij}(d\theta_{mk}+\mathcal{L}_eg_{mk}),
\end{equation}
where $d\theta_{mk}=\partial_m\theta_k-\partial_k\theta_m$;

\begin{equation}\label{eq_for_g}
c^r_{hj}\nabla_i\theta_r+c^r_{hi}\nabla_j\theta_r-c^r_{ij}\nabla_h\theta_r=c^r_{ij}c^s_{hr}\nabla_e\theta_s-\theta_r\nabla_hc^r_{ij}.
\end{equation}
\end{theorem}
\begin{proof}
	Let us first prove the equivalence between condition \eqref{Ch_sym} and condition \eqref{DN2}. The fact that \eqref{Ch_sym} implies \eqref{DN2} is straightforward. Indeed,
	\begin{eqnarray*}
		\tilde\nabla_kV^i_j&=&\nabla_kV^i_j-\frac{1}{2}g^{ls}c^m_{jk}(d\theta_{ms}+\mathcal{L}_eg_{ms})V^i_l+\frac{1}{2}g^{is}c^m_{kt}(d\theta_{ms}+\mathcal{L}_eg_{ms})V^t_j\\
		&=&\nabla_jV^i_k-\frac{1}{2}g^{ls}c^m_{kj}(d\theta_{ms}+\mathcal{L}_eg_{ms})V^i_l+\frac{1}{2}g^{is}c^m_{jt}(d\theta_{ms}+\mathcal{L}_eg_{ms})V^t_k=\tilde\nabla_jV^i_k.
	\end{eqnarray*}
	Let us now prove that \eqref{DN2} implies \eqref{Ch_sym}. Let $X$ be one of the solutions $X_{(1)} , \dots, X_{(n)}$. From
	\[\tilde\nabla_kV^i_j=\tilde\nabla_jV^i_k,\qquad\nabla_kV^i_j=\nabla_jV^i_k,\]
	it follows that
	\[(\tilde\nabla_k-\nabla_k)V^i_j=(\tilde\nabla_j-\nabla_j)V^i_k.\]
	In other words, we have the identity 
	\[(\tilde\Gamma^i_{ks}-\Gamma^i_{ks})V^s_j=(\tilde\Gamma^i_{ks}-\Gamma^i_{ks})c^s_{jt}X^t=(\tilde\Gamma^i_{js}-\Gamma^i_{js})c^s_{kt}X^t=
	(\tilde\Gamma^i_{js}-\Gamma^i_{js})V^s_{k}\]
	for any solution $X_{(1)} , \dots, X_{(n)}$, and thus
	\[(\tilde\Gamma^i_{ks}-\Gamma^i_{ks})c^s_{jt}=(\tilde\Gamma^i_{js}-\Gamma^i_{js})c^s_{kt}.\]
	Using Lemma \ref{MToeplitz_implied} we have that 
	\[\tilde\Gamma^i_{ks}-\Gamma^i_{ks}=c^t_{ks}W^i_t,\]
	for some  $(1,1)$-tensor field $W$.  In order to determine $W$ we observe that in  David-Hertling  coordinates:
	\[W^i_s=(\tilde\Gamma^i_{ks}-\Gamma^i_{ks})e^k=\tilde\Gamma^i_{ks}e^k=
	\f{1}{2}g^{il}(\partial_s g_{kl}+\partial_k g_{ls}-\partial_l g_{ks})e^k=\f{1}{2}g^{il}(\mathcal{L}_e g_{ls}+d\theta_{sl}).\]
	Summarising, we have
	\[\tilde\Gamma^i_{ks}-\Gamma^i_{ks}=\f{1}{2}c^t_{ks}g^{il}(\mathcal{L}_e g_{lt}+d\theta_{tl}).\]
	Moreover, it is immediate to check that
	\[\frac{1}{2}g^{lk}c^m_{ij}(d\theta_{mk}+\mathcal{L}_eg_{mk})=-g^{lk}c^m_{ij}(\tilde{\nabla}_k\theta_m-\mathcal{L}_eg_{km}).\]
	This  proves \eqref{Ch_sym}. From \eqref{Ch_sym} it follows that 
	\[e^h\tilde\Gamma^l_{hm}=-g^{lk}(\tilde\nabla_k\theta_m-\mathcal{L}_eg_{mk}).\]
	We have
	\[\tilde\Gamma^l_{ij}-\Gamma^l_{ij}=c^m_{ij}e^h\tilde\Gamma^l_{hm},\]
	that is
	\begin{equation}\label{start-id}
		\tilde\Gamma^k_{ij}-c^s_{ij}e^t\tilde\Gamma^k_{ts}=\Gamma^k_{ij}.
	\end{equation}
	The left-hand side of the above  system can be now written in terms of the  metric as
	\begin{eqnarray*}
		& &\frac{1}{2}g^{km}(\partial_ig_{mj}+\partial_jg_{im}-\partial_mg_{ij})-\frac{1}{2}c^s_{ij}e^tg^{km}(\partial_tg_{ms}+\partial_sg_{tm}-\partial_mg_{ts})\\
		&=&\frac{1}{2}g^{km}(c^r_{mj}\partial_i\theta_r+c^r_{im}\partial_j\theta_r-c^r_{ij}\partial_m\theta_r)-\frac{1}{2}c^s_{ij}e^tg^{km}(c^r_{ms}\partial_t\theta_r+c^r_{tm}\partial_s\theta_r-c^r_{ts}\partial_m\theta_r)\\
		&=&\frac{1}{2}g^{km}(c^r_{mj}\partial_i\theta_r+c^r_{im}\partial_j\theta_r)-\frac{1}{2}c^s_{ij}e^tg^{km}(c^r_{ms}\partial_t\theta_r)
		-\frac{1}{2}c^s_{ij}g^{km}\partial_s\theta_m.
	\end{eqnarray*}
	Using the last expression for  the  left-hand side and multiplying both sides  for $g_{hk}$, we obtain
	\begin{equation}\label{eq_for_g_bis}
		c^r_{hj}\partial_i\theta_r+c^r_{ih}\partial_j\theta_r-c^s_{ij}\partial_s\theta_h=2c^r_{hs}\Gamma^s_{ij}\theta_r+c^s_{ij}c^r_{hs}e^t\partial_t\theta_r.
	\end{equation}
	In canonical coordinates the product and the unit vector field are constant and   \eqref{eq_for_g} reduces to \eqref{eq_for_g_bis}. It remains to prove that 
	\eqref{start-id}  implies \eqref{Ch_sym}. From \eqref{start-id} we have that
	\begin{eqnarray*}
		\tilde\Gamma^k_{ij}-\Gamma^k_{ij}=c^s_{ij}e^t\tilde\Gamma^k_{ts}=
		\frac{1}{2}c^s_{ij}e^tg^{km}(\partial_tg_{ms}+\partial_sg_{tm}-\partial_mg_{ts}),
	\end{eqnarray*}
	and
	\[e^t(\partial_tg_{ms}+\partial_sg_{tm}-\partial_mg_{ts})=
	\mathcal{L}_eg_{ms}+d\theta_{sm}.\]
\end{proof}

\begin{remark}
In the  semisimple case the system \eqref{eq_for_g} reduces to the system \eqref{sys_met}. Indeed, for $i\ne j\ne h\ne i$ the system is trivially satisfied. For $i=j\ne h$ it reduces to \eqref{eq_for_g}:
\begin{equation*}
-\partial_i\theta_h=2\Gamma^h_{ii}\theta_h=-2\Gamma^h_{ih}\theta_h.
\end{equation*}
The cases $i=h\ne j$ and $j=h\ne i$ can be treated in a similar way. Finally, for $i=j=k$ we have
\begin{equation*}
\partial_i\theta_i=2\Gamma^i_{ii}\theta_i+e^t\partial_t\theta_i=-2\sum_{t\ne i}\Gamma^i_{it}\theta_i+e^t\partial_t\theta_i=-\sum_{t\ne i}\partial_t \theta_i+e^t\partial_t\theta_i
.
\end{equation*}
\end{remark}

\begin{remark}
Multiplying  both sides of \eqref{eq_for_g_bis} by $e^h$ and taking the sum over $h$, we get
\begin{equation}\label{red_eq_for_g}
(\partial_i\theta_j+\partial_j\theta_i)-c^s_{ij}e^h(\partial_h\theta_s+\partial_s\theta_h)=2\Gamma^k_{ij}\theta_{k}.
\end{equation}
It is easy to check  that the system \eqref{red_eq_for_g} provides only necessary conditions. 
\end{remark}

\subsection{Hamiltonian formalism in Darboux-Tsarev case}
In this subsection we study the Dubrovin-Novikov-Tsarev system in the Darboux-Tsarev case.
\begin{proposition}\label{DT_system}
In Darboux-Tsarev case the system \eqref{eq_for_g_bis} is equivalent to the system
\begin{subequations}
	\begin{equation}\label{eq:completeness_thetazeros}
		\partial_{j(\beta)}\theta_{i(\alpha)}= 0, \qquad \text{for } i+j \geq m_{\alpha} + \begin{cases}
			2 & \text{ for } \beta \neq \alpha,\\
			3 & \text{ for } \beta = \alpha,
		\end{cases}
	\end{equation}
    \begin{equation}
        \partial_{2(\alpha)}\theta_{m_{\alpha}(\alpha)} = \dfrac{2}{m_{\alpha}}\left(\sum_{k=2}^{m_{\alpha}}\Gamma^{k(\alpha)}_{2(\alpha) k(\alpha)}\right)\theta_{m_{\alpha}(\alpha)},
        \label{eq:multisyst_top2sameblock}
    \end{equation}
    \begin{equation}
        \begin{aligned}\partial_{j(\alpha)}\theta_{i(\alpha)} = \, & \, \partial_{(j-1)(\alpha)}\theta_{(i+1)(\alpha)} + \partial_{2(\alpha)} \theta_{(i+j-2)(\alpha)} - \partial_{1(\alpha)} \theta_{(i+j-1)(\alpha)} \\ \, & \, + 2 \sum_{s=j-1}^{m_{\alpha}-i+1} \left(\Gamma^{(s-j+1)(\alpha)}_{1(\alpha)1(\alpha)} - \Gamma^{s(\alpha)}_{2(\alpha) (j-1)(\alpha)}\right)\theta_{(i+s-1)(\alpha)},\\&\qquad\qquad\qquad\qquad 3\leq j\leq m_\alpha-i+2,\quad i\leq m_\alpha-1,
    \end{aligned}
    \label{eq:multisyst_sameblock}
    \end{equation}
    \begin{equation}
        \partial_{j(\beta)} \theta_{i(\alpha)} = 2 \sum_{s=j}^{m_{\alpha}-i+1} \Gamma^{s(\alpha)}_{j(\beta)1(\alpha)} \theta_{(s+i-1)(\alpha)}, \qquad j\leq m_\alpha-i+1,\quad \alpha \neq \beta.
        \label{eq:multisyst_diffblock}
    \end{equation}
    \label{eq:multisyst}
\end{subequations}
\end{proposition}

\begin{proof}
In order to prove that the Dubrovin-Novikov-Tsarev  system (\ref{DN1},\ref{DN2}) implies system \eqref{eq:multisyst}, we  will exploit the equivalence between (\ref{DN1},\ref{DN2}), \eqref{Ch_sym} and \eqref{eq_for_g_bis} (Theorem \ref{EquivalentDN}). Using double-index notation, we can write \eqref{Ch_sym} as
\begin{equation}\label{Ch_sym_bis}
\Gamma^{i(\alpha)}_{j(\beta)k(\gamma)}=\tilde{\Gamma}^{i(\alpha)}_{j(\beta)k(\gamma)}-\frac{1}{2}g^{i(\alpha)s(\sigma)}c^{m(\mu)}_{j(\beta)k(\gamma)}\big((d\theta)_{m(\mu)s(\sigma)}+e^{t(\tau)}\partial_{t(\tau)}g_{s(\sigma)m(\mu)}\big).
\end{equation}
From Lemma \ref{MToeplitz_implied}  it follows that, in the regular case and in canonical coordinates,
	\begin{equation}\label{blockHankel}
		g_{i(\alpha)j(\beta)}=\delta_{\alpha\beta}\,\theta_{(i+j-1)(\alpha)},\qquad i\in\{1,\dots,m_\alpha\},\,j\in\{1,\dots,m_\beta\},\,\alpha,\beta\in\{1,\dots,r\}.
	\end{equation}
	 Moreover, the inverse of a matrix with entries \eqref{blockHankel} has entries
	\begin{equation}\label{blockHankelinverse}
		g^{i(\alpha)j(\beta)}=\delta_{\alpha\beta}\,\overline{g}^{(i+j-m_\alpha)(\alpha)},\qquad i\in\{1,\dots,m_\alpha\},\,j\in\{1,\dots,m_\beta\},\,\alpha,\beta\in\{1,\dots,r\},
	\end{equation}
	where the functions $\overline{g}^{i(\alpha)}$ are the unique solutions of the linear system
\[V^{i(\alpha)}_{j(\beta)}\overline{g}^{j(\beta)}=e^{i(\alpha)},\]
with	 $V$ being the invertible matrix defined by
\[V^{i(\alpha)}_{j(\beta)}=\delta_{\alpha\beta}\delta^{i+k-m_{\alpha}}_j\theta_{k(\alpha)}.\]
Using \eqref{blockHankel} and  \eqref{blockHankelinverse}, the Christoffel symbols of the Levi-Civita connection of $g$ read
\begin{equation*}
	\tilde{\Gamma}^{i(\alpha)}_{j(\beta)k(\gamma)}=\frac{1}{2}\overline{g}^{(i+s-m_\alpha)(\alpha)}\big(\delta_{\alpha\gamma}\partial_{j(\beta)}\theta_{(s+k-1)(\alpha)}+\delta_{\alpha\beta}\partial_{k(\gamma)}\theta_{(j+s-1)(\alpha)}-\delta_{\beta\gamma}\partial_{s(\alpha)}\theta_{(j+k-1)(\beta)}\big).
\end{equation*}
Then, using \eqref{Ch_sym_bis}, the Christoffel symbols of $\nabla$ are 
\begin{align}\notag
	\Gamma^{i(\alpha)}_{j(\beta)k(\gamma)}&=\frac{1}{2}\overline{g}^{(i+s-m_\alpha)(\alpha)}\big(\delta_{\alpha\gamma}\partial_{j(\beta)}\theta_{(s+k-1)(\alpha)}+\delta_{\alpha\beta}\partial_{k(\gamma)}\theta_{(j+s-1)(\alpha)}-\delta_{\beta\gamma}\partial_{s(\alpha)}\theta_{(j+k-1)(\beta)}\notag\\
	&\quad- \delta_{\beta \gamma} \partial_{(j+k-1)(\gamma) }\theta_{s(\alpha)} + \delta_{\gamma \beta}\partial_{s(\alpha)} \theta_{(j+k-1)(\gamma)} - \sum_{\tau = 1}^r \delta_{\alpha \beta} \delta_{\alpha \gamma}\partial_{1(\tau)}\theta_{(s+j+k-2)(\alpha)}\big).
	\label{bfromGamma_multiblock}
\end{align}

\hspace{-1em}System \eqref{eq:multisyst} follows from the following Lemmas.
\begin{lemma}\label{Prop_Jr_vanishingChr_theta_1}
	The following conditions are equivalent:
	\begin{itemize}
		\item[a.] $\partial_{i(\alpha)}\theta_{j(\alpha)}=0$ for $i,j\in\{1,\dots,m_\alpha\}$ such that $i+j\geq m_\alpha+3$;
		\item[b.] $\Gamma^{i(\alpha)}_{j(\alpha)k(\alpha)}=0$ for $i,j,k\in\{1,\dots,m_\alpha\}$ such that $i-j-k\leq-3$, $j, k \geq2$.
	\end{itemize}		
\end{lemma}
\begin{proof}
	Let us first assume that $\partial_{i(\alpha)}\theta_{j(\alpha)}=0$ when $i+j\geq m_\alpha+3$. Then, by considering $\alpha=\beta=\gamma$ in \eqref{bfromGamma_multiblock}, we have
	\begin{align*}
		\Gamma^{i(\alpha)}_{j(\alpha)k(\alpha)}
		&=\frac{1}{2}\overset{m_\alpha}{\underset{s=m_\alpha-i+1}{\sum}}\overline{g}^{(i+s-m_\alpha)(\alpha)}\big(\partial_{j(\alpha)}\theta_{(s+k-1)(\alpha)}+\partial_{k(\alpha)}\theta_{(j+s-1)(\alpha)}\notag\\
		&-\partial_{(j+k-1)(\alpha)}\theta_{s(\alpha)}-\sum_{\tau=1}^{r}\partial_{1(\tau)}\theta_{(s+j+k-2)(\alpha)}\big).
	\end{align*}
	For $i-j-k\leq-3$, $j, k \geq2$, we get
	\begin{align*}
		\Gamma^{i(\alpha)}_{j(\alpha)k(\alpha)}&=\frac{1}{2}\overset{m_\alpha}{\underset{s=m_\alpha-i+1}{\sum}}\overline{g}^{(i+s-m_\alpha)(\alpha)}\big(\partial_{j(\alpha)}\theta_{(s+k-1)(\alpha)}+\partial_{k(\alpha)}\theta_{(j+s-1)(\alpha)}-\partial_{(j+k-1)(\alpha)}\theta_{s(\alpha)}\big),
	\end{align*}
	where each of the three instances of partial derivatives vanishes, being of the form $\partial_{a(\alpha)}\theta_{b(\alpha)}$ with $a+b=s+j+k-1\geq m_\alpha-i+j+k\geq m_\alpha+3$.
	\newline
    \newline
	Let us now assume that $\Gamma^{i(\alpha)}_{j(\alpha)k(\alpha)}=0$ when $i-j-k\leq-3$, $j, k \geq2$. Then, for such indices, we have
	\begin{align}\label{Prop_Jr_vanishingChr_theta_1_temp1}
		\overset{m_\alpha}{\underset{s=m_\alpha-i+1}{\sum}}\overline{g}^{(i+s-m_\alpha)(\alpha)}\big(\partial_{j(\alpha)}\theta_{(s+k-1)(\alpha)}+\partial_{k(\alpha)}\theta_{(j+s-1)(\alpha)}-\partial_{(j+k-1)(\alpha)}\theta_{s(\alpha)}\big)=0.
	\end{align}
	By choosing $i=1$ in \eqref{Prop_Jr_vanishingChr_theta_1_temp1}, we get $\partial_{(j+k-1)(\alpha)}\theta_{m_\alpha(\alpha)}=0$ (since $\bar{g}^{1(\alpha)} \neq 0$, by nondegeneracy of the metric), proving
	\begin{align*}
		\partial_{l(\alpha)}\theta_{m_\alpha(\alpha)}=0,\qquad l\geq3.
	\end{align*}
	Let us fix $N\in\{m_\alpha-j-k+4,\dots,m_\alpha-1\}$, and inductively assume that
	\begin{align*}
		\partial_{l(\alpha)}\theta_{t(\alpha)}=0,\qquad l+t\geq m_\alpha+3,
	\end{align*}
	for each $t>N$. By choosing $i=m_\alpha-N+1$ in \eqref{Prop_Jr_vanishingChr_theta_1_temp1}, we get
	\begin{align*}
		0&=\overset{m_\alpha}{\underset{s=N}{\sum}}\overline{g}^{(s-N+1)(\alpha)}\big(\partial_{j(\alpha)}\theta_{(s+k-1)(\alpha)}+\partial_{k(\alpha)}\theta_{(j+s-1)(\alpha)}-\partial_{(j+k-1)(\alpha)}\theta_{s(\alpha)}\big)\\
		&=-\overline{g}^{1(\alpha)}\partial_{(j+k-1)(\alpha)}\theta_{N(\alpha)},
	\end{align*}
	yielding $\partial_{(j+k-1)(\alpha)}\theta_{N(\alpha)}=0$ and thus proving
	\begin{align*}
		\partial_{l(\alpha)}\theta_{N(\alpha)}=0,\qquad l\geq m_\alpha-N+3.
	\end{align*}
\end{proof}
\begin{lemma}\label{Prop_Jr_vanishingChr_theta_2}
	Let $\alpha\neq\beta$. The following conditions are equivalent:
	\begin{itemize}
		\item[a.] $\partial_{j(\beta)}\theta_{i(\alpha)}=0$ for $i\in\{1,\dots,m_\alpha\}$, $j\in\{1,\dots,m_\beta\}$ such that $i+j\geq m_\alpha+2$;
		\item[b.] $\Gamma^{i(\alpha)}_{j(\beta)1(\alpha)}=0$ for $i\in\{1,\dots,m_\alpha\}$, $j\in\{1,\dots,m_\beta\}$ such that $i<j$.
	\end{itemize}
\end{lemma}
\begin{proof}
	Let us first assume that $\Gamma^{i(\alpha)}_{j(\beta)1(\alpha)}=0$ when $i<j$. Then, by considering ${\alpha=\gamma\neq\beta}$ and $k=1$ in \eqref{bfromGamma_multiblock}, we have
	\begin{align}\label{Prop_Jr_vanishingChr_theta_2_temp2}
		\Gamma^{i(\alpha)}_{j(\beta)1(\alpha)}&=\frac{1}{2}\overset{m_\alpha}{\underset{s=m_\alpha-i+1}{\sum}}\overline{g}^{(i+s-m_\alpha)(\alpha)}\partial_{j(\beta)}\theta_{s(\alpha)}.
	\end{align}
	By choosing $i=1$ in \eqref{Prop_Jr_vanishingChr_theta_2_temp2}, we get $\partial_{j(\beta)}\theta_{m_\alpha(\alpha)}=0$ for each $j\geq2$. Let us fix $N< m_\alpha$, and inductively assume that
	\begin{align*}
		\partial_{j(\beta)}\theta_{t(\alpha)}=0,\qquad j\geq m_\alpha-t+2,
	\end{align*}
	for each $t\geq N+1$. By choosing $i=m_\alpha-N+1$ in \eqref{Prop_Jr_vanishingChr_theta_2_temp2}, for each $j\geq m_\alpha-N+2$ we get
	\begin{align*}
		0=\Gamma^{(m_\alpha-N+1)(\alpha)}_{j(\beta)1(\alpha)}&=\frac{1}{2}\overset{m_\alpha}{\underset{s=N}{\sum}}\overline{g}^{(s-N+1)(\alpha)}\partial_{j(\beta)}\theta_{s(\alpha)}=\frac{1}{2}\overline{g}^{1(\alpha)}\partial_{j(\beta)}\theta_{N(\alpha)},
	\end{align*}
	yielding $\partial_{j(\beta)}\theta_{N(\alpha)}=0$ and thus proving point a.
	\newline
    \newline
	Let us now assume that $\partial_{j(\beta)}\theta_{i(\alpha)}=0$ when $i+j\geq m_\alpha+2$. By means of \eqref{Prop_Jr_vanishingChr_theta_2_temp2}, we immediately deduce that that b holds.
\end{proof}
\begin{corollary}\label{lemma:completeness_thetazeros}
		If the system is of Darboux-Tsarev type then \eqref{eq:completeness_thetazeros} holds for each $\alpha,\beta\in\{1,\dots,r\}$.
	\end{corollary}
\begin{proof}
	Corollary  \ref{lemma:completeness_thetazeros} follows from Lemmas \ref{lemma:old4}, \ref{Prop_Jr_vanishingChr_theta_1}, and  \ref{Prop_Jr_vanishingChr_theta_2}.
\end{proof}
\begin{lemma}\label{LemmaxEqsys_partial2malpha}
	\eqref{eq_for_g_bis} implies \eqref{eq:multisyst_top2sameblock}.
\end{lemma}
\begin{proof}
	By considering \eqref{eq_for_g_bis} with all indices referring to the same block ($\alpha$) and $h=m_\alpha-1$, $i=j=2$, we get
	\begin{align*}
		2\partial_{2(\alpha)}\theta_{m_\alpha(\alpha)}=\partial_{3(\alpha)}\theta_{(m_\alpha-1)(\alpha)}+2\theta_{m_\alpha(\alpha)}\Gamma^{2(\alpha)}_{2(\alpha)2(\alpha)},
	\end{align*}
    by using Lemma \ref{lemma:old4}
	Let us fix $N\in\{3,\dots,m_\alpha\}$, and inductively assume that
	\begin{align}\label{LemmaxEqsys_partial2malpha_indstep}
		s\partial_{2(\alpha)}\theta_{m_\alpha(\alpha)}=\partial_{(s+1)(\alpha)}\theta_{(m_\alpha-s+1)(\alpha)}+2\theta_{m_\alpha(\alpha)}\sum_{k=2}^{s}\Gamma^{k(\alpha)}_{2(\alpha)k(\alpha)}
	\end{align}
	for each $s=N-1$. By considering \eqref{eq_for_g_bis} with all indices referring to the same block ($\alpha$) and $h=m_\alpha-N+1$, $i=N$, $j=2$, we get
	\begin{align}\label{LemmaxEqsys_partial2malpha_auxiliary}
		\partial_{2(\alpha)}\theta_{m_\alpha(\alpha)}=\partial_{(N+1)(\alpha)}\theta_{(m_\alpha-N+1)(\alpha)}-\partial_{N(\alpha)}\theta_{(m_\alpha-N+2)(\alpha)}+2\theta_{m_\alpha(\alpha)}\Gamma^{N(\alpha)}_{2(\alpha)N(\alpha)}.
	\end{align}
	Using \eqref{LemmaxEqsys_partial2malpha_indstep} with $s=N-1$, and Lemma \ref{lemma:old4}, gives
	\begin{equation*}
N\partial_{2(\alpha)}\theta_{m_\alpha(\alpha)} =\partial_{(N+1)(\alpha)}\theta_{(m_\alpha-N+1)(\alpha)}+2\theta_{m_\alpha(\alpha)}\sum_{k=2}^{N}\Gamma^{k(\alpha)}_{2(\alpha)k(\alpha)},
	\end{equation*}
	proving \eqref{LemmaxEqsys_partial2malpha_indstep} for any $s\in\{2,\dots,m_\alpha\}$. In particular, for $s=m_\alpha$ we retrieve \eqref{eq:multisyst_top2sameblock}.	
\end{proof}
\begin{lemma}\label{LemmaxEqsys_partialdiffblock}
	\eqref{eq_for_g_bis} implies \eqref{eq:multisyst_diffblock}.
\end{lemma}
\begin{proof}
	Considering \eqref{eq_for_g_bis} with $(h,i,j) \mapsto (h(\beta), i(\alpha), 1(\alpha))$, and $\alpha\neq\beta$ gives
	\begin{align*}
	\partial_{i(\alpha)}\theta_{h(\beta)}&=-2\sum_{s=i}^{m_\beta-h+1}\Gamma^{s(\beta)}_{i(\alpha)1(\alpha)}\theta_{(h+s-1)(\beta)}=2\sum_{s=i}^{m_\beta-h+1}\Gamma^{s(\beta)}_{i(\alpha)1(\beta)}\theta_{(h+s-1)(\beta)},
	\end{align*}
    where we have used \eqref{Chr_nablae_1} 
 and Lemma \ref{lemma:old4}.
\end{proof}
\begin{lemma}\label{LemmaxEqsys_partialsameblockgeneral}
	\eqref{eq_for_g_bis} implies \eqref{eq:multisyst_sameblock}.
\end{lemma}
\begin{proof}
	By considering \eqref{eq_for_g_bis} with all indices referring to the same block ($\alpha$) with $j=2$, and relabelling $i=k-1$ for $k\geq3$, we get
	\begin{align*}
		\partial_{k(\alpha)}\theta_{h(\alpha)}&=\partial_{(k-1)(\alpha)}\theta_{(h+1)(\alpha)}+\partial_{2(\alpha)}\theta_{(h+k-2)(\alpha)}-\partial_{1(\alpha)}\theta_{(h+k-1)(\alpha)}\\
		&-2\sum_{s=k-1}^{m_\alpha-h+1}\Gamma^{s(\alpha)}_{2(\alpha)(k-1)(\alpha)}\theta_{(h+s-1)(\alpha)}-\sum_{\beta\neq\alpha}\partial_{1(\beta)}\theta_{(h+k-1)(\alpha)},
	\end{align*}
	where, by \eqref{eq:multisyst_diffblock}, 
	\begin{align*}
		-\sum_{\beta\neq\alpha}\partial_{1(\beta)}\theta_{(h+k-1)(\alpha)}&=-2\sum_{\beta\neq\alpha}\sum_{s=1}^{m_{\alpha}-h-k+2} \Gamma^{s(\alpha)}_{1(\beta)1(\alpha)} \theta_{(s+h+k-2)(\alpha)}\\&\underset{\eqref{Chr_nablae_2}}{=}2\sum_{s=1}^{m_{\alpha}-h-k+2} \Gamma^{s(\alpha)}_{1(\alpha)1(\alpha)} \theta_{(s+h+k-2)(\alpha)}=2\sum_{t=k}^{m_{\alpha}-h+1} \Gamma^{(t-k+1)(\alpha)}_{1(\alpha)1(\alpha)} \theta_{(t+h-1)(\alpha)}.
	\end{align*}
	Thus,
	\begin{align*}
		\partial_{k(\alpha)}\theta_{h(\alpha)}&=\partial_{(k-1)(\alpha)}\theta_{(h+1)(\alpha)}+\partial_{2(\alpha)}\theta_{(h+k-2)(\alpha)}-\partial_{1(\alpha)}\theta_{(h+k-1)(\alpha)}\\
		&-2\sum_{s=k-1}^{m_\alpha-h+1}\bigg(\Gamma^{s(\alpha)}_{2(\alpha)(k-1)(\alpha)}- \Gamma^{(s-k+1)(\alpha)}_{1(\alpha)1(\alpha)}\bigg)\theta_{(h+s-1)(\alpha)}.
	\end{align*}
\end{proof}

\hspace{-1em}Hence, we have retrieved system \eqref{eq:multisyst} from \eqref{eq_for_g_bis}. We are now going to show that  system \eqref{eq:multisyst} contains the whole information of \eqref{eq_for_g_bis}, that is it implies \eqref{eq_for_g_bis}. In  double-index notation, \eqref{eq_for_g_bis} reads
	\begin{align}\label{eqforg_multiblock}
		&\delta_{\beta\gamma}\partial_{i(\alpha)}\theta_{(h+j-1)(\beta)}+\delta_{\alpha\gamma}\partial_{j(\beta)}\theta_{(h+i-1)(\alpha)}-\delta_{\alpha\beta}\partial_{(i+j-1)(\alpha)}\theta_{h(\gamma)}\\&-\delta_{\alpha\beta}\delta_{\alpha\gamma}\sum_{\tau=1}^{r}\partial_{1(\tau)}\theta_{(h+i+j-2)(\alpha)}-2\Gamma^{s(\gamma)}_{i(\alpha)j(\beta)}\theta_{(h+s-1)(\gamma)}=0\notag.
	\end{align}
	If $\alpha\neq\beta\neq\gamma\neq\alpha$ then \eqref{eqforg_multiblock} holds trivially. Moreover, in each of the cases $\alpha=\beta\neq\gamma$, $\alpha=\gamma\neq\beta$, $\alpha\neq\beta=\gamma$, \eqref{eqforg_multiblock} easily follows from \eqref{eq:completeness_thetazeros} and \eqref{eq:multisyst_diffblock} by using Lemma \ref{lemma:old1} and the flatness of the unit vector field. We are then left with proving the case where $\alpha=\beta=\gamma$, where \eqref{eqforg_multiblock} reads
	\begin{align}\label{eqforg_multiblock_alpha}
		&\partial_{i(\alpha)}\theta_{(h+j-1)(\alpha)}+\partial_{j(\alpha)}\theta_{(h+i-1)(\alpha)}-\partial_{(i+j-1)(\alpha)}\theta_{h(\alpha)}\\&-\sum_{\tau=1}^{r}\partial_{1(\tau)}\theta_{(h+i+j-2)(\alpha)}-2\Gamma^{s(\alpha)}_{i(\alpha)j(\alpha)}\theta_{(h+s-1)(\alpha)}=0\notag.
	\end{align}
	Since \eqref{eqforg_multiblock_alpha} holds trivially for $i=1$ or $j=1$ by means of \eqref{eq:multisyst_diffblock}, Lemma \ref{lemma:old1} and flatness of the unit, let us consider $i,j\geq2$. Moreover, by \eqref{eq:completeness_thetazeros}, \eqref{eqforg_multiblock_alpha} trivially holds for $h=m_\alpha$ as well. Let us fix $H\in\{1,\dots,m_\alpha-1\}$ and inductively assume that \eqref{eqforg_multiblock_alpha} holds for each ${h\in\{H+1,\dots,m_\alpha\}}$. By \eqref{eq:completeness_thetazeros}, we consider the left-hand side of \eqref{eqforg_multiblock_alpha} for $h=H$ with $i+j\leq m_\alpha-H+3$:
	\begin{align*}
		&\partial_{i(\alpha)}\theta_{(H+j-1)(\alpha)}+\partial_{j(\alpha)}\theta_{(H+i-1)(\alpha)}-\partial_{(i+j-1)(\alpha)}\theta_{H(\alpha)}\\&-\sum_{\tau=1}^{r}\partial_{1(\tau)}\theta_{(H+i+j-2)(\alpha)}-2\Gamma^{s(\alpha)}_{i(\alpha)j(\alpha)}\theta_{(H+s-1)(\alpha)},
	\end{align*}
	which reduces to
	\begin{align}\label{eqforg_multiblock_alpha_mustvanishLHS}
		&\partial_{j(\alpha)}\theta_{(H+i-1)(\alpha)}-\partial_{(i+j-1)(\alpha)}\theta_{H(\alpha)}-2\Gamma^{s(\alpha)}_{i(\alpha)j(\alpha)}\theta_{(H+s-1)(\alpha)}\\&-\partial_{(j-1)(\alpha)}\theta_{(H+i)(\alpha)}+\partial_{(i+j-2)(\alpha)}\theta_{(H+1)(\alpha)}+2\Gamma^{s(\alpha)}_{i(\alpha)(j-1)(\alpha)}\theta_{(H+s)(\alpha)}\notag,
	\end{align}
	as
	\begin{align*}
		\partial_{i(\alpha)}\theta_{(H+j-1)(\alpha)}-\sum_{\tau=1}^{r}\partial_{1(\tau)}\theta_{(H+i+j-2)(\alpha)}=&-\partial_{(j-1)(\alpha)}\theta_{(H+i)(\alpha)}+\partial_{(i+j-2)(\alpha)}\theta_{(H+1)(\alpha)}\\&+2\Gamma^{s(\alpha)}_{i(\alpha)(j-1)(\alpha)}\theta_{(H+s)(\alpha)}
	\end{align*}
	by the inductive assumption for $h=H+1$ (with $j-1$ playing the role of $j$). By employing \eqref{eq:multisyst_sameblock} on the first (with $H+i-1$ playing the role of $i$) and on the second (with $H$, $i+j-1$ playing the roles of $i$, $j$ respectively) elements in each row of \eqref{eqforg_multiblock_alpha_mustvanishLHS}, we get
	\begin{align*}
		&2\sum_{s=i+j-2}^{m_\alpha-h+1}\big(-\Gamma^{(s-i+1)(\alpha)}_{2(\alpha)(j-1)(\alpha)}+\Gamma^{s(\alpha)}_{2(\alpha)(i+j-2)(\alpha)}-\Gamma^{s(\alpha)}_{i(\alpha)j(\alpha)}+\Gamma^{(s-1)(\alpha)}_{i(\alpha)(j-1)(\alpha)}\big)\theta_{(h+s-1)(\alpha)}
	\end{align*}
	which vanishes by \eqref{SymmNablac}.
\end{proof}

\subsection{Compatibility}
This subsection is devoted to proving the following proposition. 
\begin{proposition}
    The system \eqref{eq:multisyst} is compatible in the sense that  
    \begin{equation}
		\big(\partial_{j(\beta)}\partial_{k(\gamma)}-\partial_{k(\gamma)}\partial_{j(\beta)}\big)\theta_{i(\alpha)}=0,
        \label{eq:multicomprel}
	\end{equation}
    for $\theta_i(\alpha)$ satisfying the system \eqref{eq:multisyst}, with $i \in \{1, \cdots, m_{\alpha}\}$, and all $j, k$ as present in the system \eqref{eq:multisyst}. 
    \label{prop:comp}
\end{proposition}
\begin{proof}
We shall prove Proposition \ref{prop:comp} by induction over $i$ starting with $i = m_{\alpha}$. We consider separately the three cases:
\begin{enumerate}
    \item $\beta = \gamma = \alpha$;
    \item $\beta = \alpha \neq \gamma$ (which, by symmetry, includes the case $\gamma = \alpha \neq \beta$);
    \item $\beta \neq \alpha \neq \gamma$.
\end{enumerate}
Let $i=m_{\alpha}$.
\vspace{-0.8em}
\newline
\newline
(1) Let $\beta = \gamma = \alpha$. Note that if $j=k$, or $j, k \geq 3$ the compatibility relation holds trivially by \eqref{eq:completeness_thetazeros}. Thus, without loss of generality, let $j = 2$ and $k \geq 3$. Then the left-hand side of \eqref{eq:multicomprel} reads
    \begin{align*}
        - \partial_{k(\alpha)}\partial_{2(\alpha)} \theta_{m_{\alpha}(\alpha)} = \, & \,  \dfrac{2}{m_{\alpha}}\partial_{k(\alpha)}\left(\sum_{s=2}^{m_{\alpha}} \Gamma^{s(\alpha)}_{2(\alpha) s(\alpha)}\theta_{m_{\alpha}(\alpha)} \right)  = \dfrac{2}{m_{\alpha}}\left(\sum_{s=2}^{m_{\alpha}}\partial_{k(\alpha)}\Gamma^{s(\alpha)}_{2(\alpha)s(\alpha)}\right)\theta_{m_{\alpha}(\alpha)},
  \end{align*}
    where we have used \eqref{eq:multisyst_top2sameblock} and  \eqref{eq:completeness_thetazeros}. However, by Lemma \ref{lemma:diffsameblockzero}, $\partial_{k(\alpha)} \Gamma^{s(\alpha)}_{2(\alpha) s(\alpha)} = 0$ for $k \geq 3$, which concludes the proof. 
    \newline
    \newline
    (2) Let $\beta = \alpha \neq \gamma$. Suppose $j=2$. Using  \eqref{eq:multisyst}, the left-hand side of \eqref{eq:multicomprel} reads
    \begin{align*}
      \,  & \,    2 \partial_{2(\alpha)} \left(\delta_1^k \Gamma^{1(\alpha)}_{k(\gamma)1(\alpha)}\theta_{m_{\alpha}(\alpha)} \right) - \partial_{k(\gamma)}\left(\dfrac{2}{m_{\alpha}}\left(\sum_{s=2}^{m_{\alpha}} \Gamma^{s(\alpha)}_{2(\alpha) s(\alpha)}\right)\theta_{m_{\alpha}(\alpha)} \right)\\
     \underset{\underset{ \eqref{eq:multisyst_top2sameblock}, \eqref{eq:multisyst_diffblock}}{\text{Leibniz,}}}{=} \, & \, 2\delta^k_{1}  \underset{= \,  0 \text{ by Lemma \ref{lemma:old5}}}{\underbrace{\left(\partial_{2(\alpha)} \Gamma^{1(\alpha)}_{k(\gamma)1(\alpha)}\right)}}\theta_{m_{\alpha}(\alpha)} + 2\delta^k_1 \Gamma^{1(\alpha)}_{k(\gamma)1(\alpha)}\partial_{2(\alpha)}\theta_{m_{\alpha} (\alpha)}\\  & \, - \dfrac{2}{m_{\alpha}}\left( \sum_{s=2}^{m_{\alpha}}\underset{ = \,  0 \text{ by Lemma \ref{lemma:1multidiff}}}{\underbrace{\partial_{k(\gamma)}\Gamma^{s(\alpha)}_{2(\alpha)s(\alpha)}}}\right)\theta_{m_{\alpha}(\alpha)} - \dfrac{2}{m_{\alpha}}\left(\sum_{s=2}^{m_{\alpha}} \Gamma^{s(\alpha)}_{2(\alpha)s(\alpha)}\right)\partial_{k(\gamma)}\theta_{m_{\alpha}(\alpha)}\\
     \underset{\eqref{eq:multisyst_top2sameblock}, \eqref{eq:multisyst_diffblock}}{=}  \, & \,  \dfrac{4}{m_{\alpha}}\delta^k_1 \Gamma^{1(\alpha)}_{k(\gamma)1(\alpha)}\left(\sum_{s=2}^{m_{\alpha}}\Gamma^{s(\alpha)}_{2(\alpha) s(\alpha)}\right)\theta_{m_{\alpha} (\alpha)} - \dfrac{4}{m_{\alpha}}\delta^k_1 \Gamma^{1(\alpha)}_{k(\gamma)1(\alpha)}\left(\sum_{s=2}^{m_{\alpha}}\Gamma^{s(\alpha)}_{2(\alpha) s(\alpha)}\right)\theta_{m_{\alpha} (\alpha)} = 0.
    \end{align*}
   Now, let $j \geq  3$. Using \eqref{eq:completeness_thetazeros} and \eqref{eq:multisyst_diffblock},  the left-hand side of  \eqref{eq:multicomprel} reads
    \begin{equation*}
        2 \delta^k_1 \partial_{j(\alpha)} \left(\Gamma^{1(\alpha)}_{k(\gamma)1(\alpha)}\theta_{m_{\alpha}(\alpha)}\right) \underset{\underset{\text{ \eqref{eq:completeness_thetazeros}}}{\text{Leibniz},}}{=} 2 \partial_{j(\alpha)}\left(\Gamma^{1(\alpha)}_{k(\gamma) 1(\alpha)}\right)\theta_{m_{\alpha}(\alpha)} \underset{\text{Lemma } \ref{lemma:1multidiff}}{=} 0.  
    \end{equation*}
    (3) Let $\beta \neq \alpha \neq \gamma$. Using \eqref{eq:multisyst_diffblock}, the left-hand side \eqref{eq:multicomprel} reads
    \begin{align*}
\, & \,         \delta^k_1 \partial_{j(\beta)}\left(\Gamma^{1(\alpha)}_{k(\gamma)1(\alpha)}\theta_{m_{\alpha}(\alpha)} \right) - \delta^j_1 \partial_{k(\gamma)} \left(\Gamma^{1(\alpha)}_{j(\beta) 1(\alpha)} \theta_{m_{\alpha}(\alpha)} \right)\\ \,  & \, \underset{\underset{\eqref{eq:multisyst_diffblock}}{\text{Leibniz}}}{=} \delta^k_1 \delta^j_1\left(\partial_{j(\beta)}\Gamma^{1(\alpha)}_{k(\gamma)1(\alpha)}  -\partial_{k(\gamma)}\Gamma^{1(\alpha)}_{j(\beta) 1(\alpha)}\right)\theta_{m_{\alpha}(\alpha)} \underset{\text{Lemma } \ref{lemma:diffswapgen}}{=} 0.
    \end{align*}
    \vspace{-1.5em}
    \newline
    \newline
 Let us fix $i\in\{1,\dots,m_\alpha-1\}$ and inductively assume that
	\begin{align*}
		\big(\partial_{j(\beta)}\partial_{k(\gamma)}-\partial_{k(\gamma)}\partial_{j(\beta)}\big)\theta_{I(\alpha)}=0
	\end{align*}
	for every $j(\beta)$, $k(\gamma)$, when  $I> i$.
\vspace{-0.5em}
\newline
\newline
(1) Let $\beta = \gamma = \alpha$. Without loss of generality we may assume that $j < k$, which implies $3 \leq j \leq m_{\alpha} - i + 2$, as otherwise the proposition holds trivially by \eqref{eq:completeness_thetazeros}. Using \eqref{eq:multisyst}, the left-hand side of \eqref{eq:multicomprel} reads
    \begin{align*}
   \qquad \quad    \, & \,   \partial_{j(\alpha)}\partial_{k(\alpha)} \theta_{i(\alpha)} - \partial_{(j-1)(\alpha)} \partial_{k(\alpha)} \theta_{(i+1)(\alpha)} - \partial_{2(\alpha)} \partial_{k(\alpha)} \theta_{(i+j-2)(\alpha)} + \partial_{1(\alpha)}\partial_{k(\alpha)} \theta_{(i+j-1)(\alpha)} \\ \, & \,  - 2 \sum_{s=j-1}^{m_{\alpha}-i+1}\left(\Gamma^{(s-j+1)(\alpha)}_{1(\alpha)1(\alpha)} - \Gamma^{s(\alpha)}_{2(\alpha)(j-1)(\alpha)}\right)\partial_{k(\alpha)}\theta_{(s+i-1)(\alpha)} \\ \, & \, 
 - 2 \sum_{s=j-1}^{m_{\alpha}-i+1}\partial_{k(\alpha)}\left(\Gamma^{(s-j+1)(\alpha)}_{1(\alpha)1(\alpha)} - \Gamma^{s(\alpha)}_{2(\alpha)(j-1)(\alpha)}\right) \theta_{(s+i-1)(\alpha)},
    \end{align*}
    where we have used the inductive hypothesis to swap partial derivatives. 
Note that if $k \geq m_{\alpha}-i+3$ then, by \eqref{eq:completeness_thetazeros}, only the final $s$-sum remains. However, by Lemma \ref{lemma:old5} this vanishes as well. 
\newline
\newline
    Thus, we may assume that $k \leq m_{\alpha} -i +2$. Expanding further, relabelling the indices to collect $\theta$'s and  using  \eqref{eq:lemmaold1diagsameblock}  we get
    \begin{equation}
        F + G,
    \end{equation}
    where,
    \begin{align*}
        F = \, & \,   4 \sum_{s=k-1}^{m_{\alpha}-i+1}\sum_{t=j-1}^{m_{\alpha} - i - s +2}\left(\Gamma^{(s-k+1)(\alpha)}_{1(\alpha)1(\alpha)} - \Gamma^{s(\alpha)}_{2(\alpha)(k-1)(\alpha)}\right)\left(\Gamma^{(t-j+1)(\alpha)}_{1(\alpha)1(\alpha)} - \Gamma^{t(\alpha)}_{2(\alpha)(j-1)(\alpha)}\right)\theta_{(t+s+i-2)(\alpha)} \\ \, & \, - 4\sum_{s=j-1}^{m_{\alpha}-i+1}\sum_{t=k-1}^{m_{\alpha} - i - s +2}\left(\Gamma^{(s-j+1)(\alpha)}_{1(\alpha)1(\alpha)} - \Gamma^{s(\alpha)}_{2(\alpha)(j-1)(\alpha)}\right)\left(\Gamma^{(t-k+1)(\alpha)}_{1(\alpha)1(\alpha)} - \Gamma^{t(\alpha)}_{2(\alpha)(k-1)(\alpha)}\right)\theta_{(t+s+i-2)(\alpha)},
    \end{align*}
    \begin{align*}
        G =  2 \sum_{s=i+j+k-3}^{m_{\alpha}}\Bigg[\Bigg( \, & \, \partial_{(k-1)(\alpha)}(\Gamma^{(s-i-j+1)(\alpha)}_{1(\alpha)1(\alpha)} - \Gamma^{(s-i)(\alpha)}_{2(\alpha)(j-1)(\alpha)}) + \partial_{j(\alpha)}(\Gamma^{(s-i-k+2)(\alpha)}_{1(\alpha)1(\alpha)} - \Gamma^{(s-i+1)(\alpha)}_{2(\alpha)(k-1)(\alpha)})\\ \, & \,  -\partial_{2(\alpha)}\Gamma^{(s-i-k+3)(\alpha)}_{2(\alpha)(j-1)(\alpha)}  + \partial_{1(\alpha)}\Gamma^{(s-i-k+2)(\alpha)}_{2(\alpha)(j-1)(\alpha)}\Bigg)  - (j \leftrightarrow k) \Bigg]\theta_{s(\alpha)},
    \end{align*}
To obtain $G$ we have truncated the sums to the range for which $\theta$ is defined, and used Remark \ref{rmk:old} and Lemma \ref{lemma:diffsameblockzero}  to obtain the lower limit of the sum. We will show that $F$ and $G$ vanish separately. 
\newline
\newline
\underline{F.}\\
This is zero simply by swapping the sums as follows. Swapping the sums in the first line gives
\begin{equation}
    4 \sum_{t=j-1}^{m_{\alpha} - i - k+3} \sum_{s=k-1}^{m_{\alpha} - i - t + 2}\left(\Gamma^{(s-k+1)(\alpha)}_{1(\alpha)1(\alpha)} - \Gamma^{s(\alpha)}_{2(\alpha)(k-1)(\alpha)}\right)\left(\Gamma^{(t-j+1)(\alpha)}_{1(\alpha)1(\alpha)} - \Gamma^{t(\alpha)}_{2(\alpha)(j-1)(\alpha)}\right)\theta_{(t+s+i-2)(\alpha)}.
    \label{eq:temp1iF}
\end{equation}
Notice that for the second line in the definition of $F$, if $s>m_{\alpha} - i - k +3$, then $t+s-i-2> m_{\alpha} +1$, using the smallest possible value for $t$; $k-1$. Thus $\theta_{(t+s-i-2)(\alpha)}$ is undefined  and we may set the upper $s$-limit to $m_{\alpha}-i-k+3$.
Hence, relabelling $s \leftrightarrow t$ in \eqref{eq:temp1iF} gives precisely the negative of the second line, as required. 
\newline
\newline
\underline{G.}\\
The condition, \eqref{eq:3RCmultiinR}, with all indices belonging to the same block (which we denote by $\alpha$), $k=1$ and \begin{itemize}
    \item $l=2$, $m \rightarrow k-1$, $i \rightarrow j-1$, $j \rightarrow s-i$ gives
    \begin{align}
        \, & \,  \label{eq:comptempG1}   \left( \partial_{(k-1)(\alpha)}\left(\Gamma^{(s-i-j+2)(\alpha)}_{1(\alpha)2(\alpha)} - \Gamma^{(s-i)(\alpha)}_{2(\alpha)(j-1)(\alpha)} \right) + \partial_{1(\alpha)}\Gamma^{(s-i-k+2)(\alpha)}_{2(\alpha)(j-1)(\alpha)}\right) - \left(j \leftrightarrow k \right) =  \\ \notag
         \, & \, \left(\Gamma^{(s-i)(\alpha)}_{(k-1)(\alpha)t(\alpha)} \Gamma^{t(\alpha)}_{2(\alpha)(j-1)(\alpha)}  + \Gamma^{(s-i-k+2)(\alpha)}_{(j-1)(\alpha)t(\alpha)} \Gamma^{t(\alpha)}_{1(\alpha)2(\alpha)}  - \Gamma^{(s-i-k+2)(\alpha)}_{1(\alpha)t(\alpha)} \Gamma^{t(\alpha)}_{2(\alpha)(j-1)(\alpha)}\right) - \left(j \leftrightarrow k\right) ,
    \end{align}
    \item $i=2$, $j \rightarrow s-i+1$, $l \rightarrow k-1$, $m \rightarrow j$ gives
    \begin{align}
    \, & \,    \label{eq:comptempG2} \partial_{j(\alpha)}\left(\Gamma^{(s-i)(\alpha)}_{1(\alpha)(k-1)(\alpha)} - \Gamma^{(s-i+1)(\alpha)}_{2(\alpha)(k-1)(\alpha)}\right) = \\ \notag
    \, & \, \partial_{2(\alpha)}\left(\Gamma^{(s-i-j+2)(\alpha)}_{1(\alpha)(k-1)(\alpha)}-\Gamma^{(s-i+1)(\alpha)}_{j(\alpha)(k-1)(\alpha)}\right) + \partial_{1(\alpha)}\left(\Gamma^{(s-i)(\alpha)}_{j(\alpha)(k-1)(\alpha)} - \Gamma^{(s-i-j+2)(\alpha)}_{2(\alpha)(k-1)(\alpha)}\right)\\ \notag
 \, & \,  +  \Gamma^{(s-i+1)(\alpha)}_{j(\alpha)t(\alpha)}\Gamma^{t(\alpha)}_{2(\alpha)(k-1)(\alpha)} - \Gamma^{(s-i+1)(\alpha)}_{2(\alpha)t(\alpha)}\Gamma^{t(\alpha)}_{j(\alpha)(k-1)(\alpha)} + \Gamma^{(s-i-j+2)(\alpha)}_{2(\alpha)t(\alpha)}\Gamma^{t(\alpha)}_{1(\alpha)(k-1)(\alpha)}\\ \notag
  \, & \, - \Gamma^{(s-i-j+2)(\alpha)}_{1(\alpha)t(\alpha)}\Gamma^{t(\alpha)}_{2(\alpha)(k-1)(\alpha)} + \Gamma^{(s-i)(\alpha)}_{1(\alpha)t(\alpha)}\Gamma^{t(\alpha)}_{j(\alpha)(k-1)(\alpha)} - \Gamma^{(s-i)(\alpha)}_{j(\alpha)t(\alpha)}\Gamma^{t(\alpha)}_{1(\alpha)(k-1)(\alpha)}.
    \end{align}
\end{itemize} 
 In the above we have used Lemma \ref{lemma:old1} together with \eqref{Chr_a} to show that the non-differentiated part of \eqref{eq:3RCmultiinR} cancel with each other if $\sigma = \beta$. Putting \eqref{eq:comptempG1}, \eqref{eq:comptempG2}, and \eqref{eq:comptempG2} with $j \leftrightarrow k$ into $G$ gives that $G$ is a sum of the terms
\begin{align}\label{wtemphy}
     & \, \partial_{1(\alpha)}\left(\Gamma^{(s-i)(\alpha)}_{j(\alpha) (k-1)(\alpha)}-\Gamma^{(s-i-j+2)(\alpha)}_{2(\alpha)(k-1)(\alpha)}-\Gamma^{(s-i)(\alpha)}_{k(\alpha)(j-1)(\alpha)}+\Gamma^{(s-i-k+2)(\alpha)}_{2(\alpha)(j-1)(\alpha)}\right)\\
     \, & \, - \partial_{2(\alpha)}\left(\Gamma^{(s-i+1)(\alpha)}_{j(\alpha) (k-1)(\alpha)}-\Gamma^{(s-i-j+3)(\alpha)}_{2(\alpha)(k-1)(\alpha)}-\Gamma^{(s-i+1)(\alpha)}_{k(\alpha)(j-1)(\alpha)}+\Gamma^{(s-i-k+3)(\alpha)}_{2(\alpha)(j-1)(\alpha)}\right)\notag\\
       \, & \, + \Big(\Gamma^{(s-i)(\alpha)}_{(k-1)(\alpha)t(\alpha)} \Gamma^{t(\alpha)}_{2(\alpha)(j-1)(\alpha)}  + \Gamma^{(s-i-k+2)(\alpha)}_{(j-1)(\alpha)t(\alpha)} \Gamma^{t(\alpha)}_{1(\alpha)2(\alpha)}  - \Gamma^{(s-i-k+2)(\alpha)}_{1(\alpha)t(\alpha)} \Gamma^{t(\alpha)}_{2(\alpha)(j-1)(\alpha)}\notag\\
        \, & \, +  \Gamma^{(s-i+1)(\alpha)}_{j(\alpha)t(\alpha)}\Gamma^{t(\alpha)}_{2(\alpha)(k-1)(\alpha)} - \Gamma^{(s-i+1)(\alpha)}_{2(\alpha)t(\alpha)}\Gamma^{t(\alpha)}_{j(\alpha)(k-1)(\alpha)} + \Gamma^{(s-i-j+2)(\alpha)}_{2(\alpha)t(\alpha)}\Gamma^{t(\alpha)}_{1(\alpha)(k-1)(\alpha)}\notag\\ \notag
  \, & \, - \Gamma^{(s-i-j+2)(\alpha)}_{1(\alpha)t(\alpha)}\Gamma^{t(\alpha)}_{2(\alpha)(k-1)(\alpha)} + \Gamma^{(s-i)(\alpha)}_{1(\alpha)t(\alpha)}\Gamma^{t(\alpha)}_{j(\alpha)(k-1)(\alpha)} - \Gamma^{(s-i)(\alpha)}_{j(\alpha)t(\alpha)}\Gamma^{t(\alpha)}_{1(\alpha)(k-1)(\alpha)}\Big) \notag\\ \, & \, - \left(j \leftrightarrow k\right).\notag
\end{align}
The non-differentiated terms in \eqref{wtemphy} can be reorganised into
\begin{align*}
	&\Gamma^{(s-i)(\alpha)}_{(k-1)(\alpha)t(\alpha)} \Gamma^{t(\alpha)}_{2(\alpha)(j-1)(\alpha)}+  \Gamma^{(s-i+1)(\alpha)}_{j(\alpha)t(\alpha)}\Gamma^{t(\alpha)}_{2(\alpha)(k-1)(\alpha)} - \Gamma^{(s-i+1)(\alpha)}_{2(\alpha)t(\alpha)}\Gamma^{t(\alpha)}_{j(\alpha)(k-1)(\alpha)}-\left(j \leftrightarrow k\right)\\
	&=\Gamma^{t(\alpha)}_{2(\alpha)(j-1)(\alpha)}\big(
	\Gamma^{(s-i)(\alpha)}_{(k-1)(\alpha)t(\alpha)}-\Gamma^{(s-i+1)(\alpha)}_{k(\alpha)t(\alpha)}\big)-\Gamma^{t(\alpha)}_{2(\alpha)(k-1)(\alpha)}\big(
	\Gamma^{(s-i)(\alpha)}_{(j-1)(\alpha)t(\alpha)}-\Gamma^{(s-i+1)(\alpha)}_{j(\alpha)t(\alpha)}\big)\\
	&-\Gamma^{(s-i+1)(\alpha)}_{2(\alpha)t(\alpha)}\big(\Gamma^{t(\alpha)}_{j(\alpha)(k-1)(\alpha)}-\Gamma^{t(\alpha)}_{(j-1)(\alpha)k(\alpha)}\big)\\
	&\overset{\eqref{SymmNablac}}{=}\Gamma^{t(\alpha)}_{2(\alpha)(j-1)(\alpha)}\big(
	\Gamma^{(s-i-t+2)(\alpha)}_{2(\alpha)(k-1)(\alpha)}-\Gamma^{(s-i+1)(\alpha)}_{2(\alpha)(t+k-2)(\alpha)}\big)-\Gamma^{t(\alpha)}_{2(\alpha)(k-1)(\alpha)}\big(
	\Gamma^{(s-i-t+2)(\alpha)}_{2(\alpha)(j-1)(\alpha)}-\Gamma^{(s-i+1)(\alpha)}_{2(\alpha)(t+j-2)(\alpha)}\big)\\
	&-\Gamma^{(s-i+1)(\alpha)}_{2(\alpha)t(\alpha)}\big(\Gamma^{(t-j+2)(\alpha)}_{2(\alpha)(k-1)(\alpha)}-\Gamma^{(t-k+2)(\alpha)}_{2(\alpha)(j-1)(\alpha)}\big),
\end{align*}
which vanishes after a relabelling of the indices in half of the sums, and
\begin{align*}
	&\Gamma^{t(\alpha)}_{1(\alpha)1(\alpha)}\big(
	\Gamma^{(s-i-k+2)(\alpha)}_{(j-1)(\alpha)(t+1)(\alpha)}- \Gamma^{(s-i-k-t+3)(\alpha)}_{2(\alpha)(j-1)(\alpha)}+ \Gamma^{(s-i-j+2)(\alpha)}_{2(\alpha)(t+k-2)(\alpha)}\\
	& \, - \Gamma^{(s-i-j-t+3)(\alpha)}_{2(\alpha)(k-1)(\alpha)} + \Gamma^{(s-i-t+1)(\alpha)}_{j(\alpha)(k-1)(\alpha)} - \Gamma^{(s-i)(\alpha)}_{j(\alpha)(t+k-2)(\alpha)}
	\big)-\left(j \leftrightarrow k\right)\\
	&=\Gamma^{t(\alpha)}_{1(\alpha)1(\alpha)}\big(
	\Gamma^{(s-i-k+2)(\alpha)}_{(j-1)(\alpha)(t+1)(\alpha)}+ \Gamma^{(s-i-j+2)(\alpha)}_{2(\alpha)(t+k-2)(\alpha)}+\Gamma^{(s-i-t+1)(\alpha)}_{j(\alpha)(k-1)(\alpha)} - \Gamma^{(s-i)(\alpha)}_{j(\alpha)(t+k-2)(\alpha)}\\
	&-\Gamma^{(s-i-j+2)(\alpha)}_{(k-1)(\alpha)(t+1)(\alpha)}- \Gamma^{(s-i-k+2)(\alpha)}_{2(\alpha)(t+j-2)(\alpha)}-\Gamma^{(s-i-t+1)(\alpha)}_{(j-1)(\alpha)k(\alpha)} + \Gamma^{(s-i)(\alpha)}_{k(\alpha)(t+j-2)(\alpha)}
	\big)\\
	&\overset{\eqref{SymmNablac}}{=}\Gamma^{t(\alpha)}_{1(\alpha)1(\alpha)}\big(
	\Gamma^{(s-i)(\alpha)}_{(k-1)(\alpha)(t+j-1)(\alpha)}-\Gamma^{(s-i)(\alpha)}_{(j-1)(\alpha)(t+k-1)(\alpha)}+\Gamma^{(s-i-k+2)(\alpha)}_{(j-1)(\alpha)(t+1)(\alpha)}-\Gamma^{(s-i-j+2)(\alpha)}_{(k-1)(\alpha)(t+1)(\alpha)}
	\big)\\
	&\overset{\eqref{SymmNablac}}{=}0.
\end{align*}
Thus, $G$ vanishes by the symmetry of $\nabla c$ (which can be seen by letting $(i, l, j, k) \rightarrow (2, j-1, k-1, s-i)$ and $(2, j-1, k-1, s-i+1)$ in \eqref{eq:nablacsymmulti} for the partial derivative with respect to $u^1$ and $u^2$, respectively).
\newline
\newline (2) Let $\beta = \alpha \neq \gamma$, and relabel $\gamma \rightarrow \beta$.  Let us first assume that $3\leq j \leq m_{\alpha} - i+2$, and $k \leq m_{\alpha} - i + 1$. After successive use of \eqref{eq:multisyst}, the left-hand side of \eqref{eq:multicomprel} reads 
\begin{equation}
    \eqref{eq:multicomprel} = G + H, 
\end{equation}
where, 
\begin{align*}
    G = 2 \sum_{s=i+j+k-2}^{m_{\alpha}}\Bigg[\, & \, \partial_{j(\alpha)}\Gamma^{(s-i+1)(\alpha)}_{k(\beta)1(\alpha)}-\partial_{(j-1)(\alpha)}\Gamma^{(s-i)(\alpha)}_{k(\beta)1(\alpha)} -  \partial_{2(\alpha)}\Gamma^{(s-i-j+3)(\alpha)}_{k(\beta)1(\alpha)} +\partial_{1(\alpha)}\Gamma^{(s-i-j+2)(\alpha)}_{k(\beta)1(\alpha)} \\ \, & \, - \partial_{k(\beta)}\left(\Gamma^{(s-i-j+2)(\alpha)}_{1(\alpha)1(\alpha)} - \Gamma^{(s-i+1)(\alpha)}_{2(\alpha)(j-1)(\alpha)}\right) \Bigg]\theta_{s(\alpha)},
\end{align*}
\begin{align*}
    H = 
 & \, 2 \underset{ = \, 0 \text{ by Lemma \ref{lemma:newnew}}}{\underbrace{\left(\partial_{j(\alpha)}\Gamma^{k(\alpha)}_{k(\beta)1(\alpha)}\right)\theta_{(k+i-1)(\alpha)}}} + 2 \underset{a)}{\underbrace{\sum_{s=k+i}^{i+j+k-3}\left(\partial_{j(\alpha)}\Gamma^{(s-i+1)(\alpha)}_{k(\beta)1(\alpha)} - \partial_{(j-1)(\alpha)}\Gamma^{(s-i)(\alpha)}_{k(\beta)1(\alpha)}\right) \theta_{s(\alpha)}}}  \\ & \, -2  \underset{b)}{\underbrace{\sum_{s=i+j-2}^{i+j+k-3}\partial_{k(\beta)}\left(\Gamma^{(s-i-j+2)(\alpha)}_{1(\alpha)1(\alpha)}-\Gamma^{(s-i+1)(\alpha)}_{2(\alpha)(j-1)(\alpha)}\right)\theta_{s(\alpha)}}} - 2  \left(\partial_{2(\alpha)}\Gamma^{k(\alpha)}_{k(\beta) 1(\alpha)}\right)\theta_{(i+j+k-3)(\alpha)},
\end{align*}
    where we have truncated the upper limits of the sums to the range for which $\theta$ is defined, and performed a relabelling of the indices, (analogously to that of term F in the previous case), in order to cancel the double sums.  We will show that G and H vanish separately.
    \vspace{-0.5em}
    \newline
    \newline
\underline{G.}\\
Considering \eqref{eq:3RCmultiinR} with $\tau = \rho = \gamma = \alpha \neq \beta$, and $i=2$, $l=1$, $(k, j, m) \rightarrow (j, s-i+2, k)$, together with \eqref{eq:lemmaold1diagsameblock}, gives
\begin{align}
    \partial_{j(\alpha)}\Gamma^{(s-i+1)(\alpha)}_{1(\alpha)k(\beta)} - \partial_{2(\alpha)} \Gamma^{(s-i-j+3)(\alpha)}_{1(\alpha)k(\beta)} = \, & \, \notag   \Gamma^{(s-i-j+3)(\alpha)}_{2(\alpha)t(\sigma)} \Gamma^{t(\sigma)}_{k(\beta)1(\alpha)} - \Gamma^{(s-i-j+3)(\alpha)}_{k(\beta)t(\sigma)} \Gamma^{t(\sigma)}_{1(\alpha)2(\alpha)} \\ \, & \,  + \Gamma^{(s-i+1)(\alpha)}_{t(\sigma) k(\beta)} \Gamma^{t(\sigma)}_{j(\alpha)1(\alpha)} - \Gamma^{(s-i+1)(\alpha)}_{t(\sigma)j(\alpha)} \Gamma^{t(\sigma)}_{1(\alpha)k(\beta)}.
         \label{eq:3RCGa}
    \end{align}
Letting $l=2$, $i=1$ and $(k, j, m) \rightarrow (j-1, s-i+1, k)$ instead in \eqref{eq:3RCmultiinR}  with $\tau = \rho = \gamma = \alpha \neq \beta$  gives
\begin{align}
    \partial_{k(\beta)}\left(\Gamma^{(s-i-j+3)(\alpha)}_{2(\alpha)1(\alpha)} - \Gamma^{(s-i+1)(\alpha)}_{2(\alpha) (j-1)(\alpha)} \right) = \, & \, \notag \partial_{1(\alpha)}\Gamma^{(s-i-j+3)(\alpha)}_{2(\alpha)k(\beta)} - \partial_{(j-1)(\alpha)}\Gamma^{(s-i+1)(\alpha)}_{2(\alpha)k(\beta)} \\ \, & \, \notag  + \Gamma^{(s-i-j+3)(\alpha)}_{t(\sigma)1(\alpha)} \Gamma^{t(\sigma)}_{k(\beta)2(\alpha)} - \Gamma^{(s-i-j+3)(\alpha)}_{t(\sigma)k(\beta)} \Gamma^{t(\sigma)}_{2(\alpha)1(\alpha)}  \\ \, & \,  + \Gamma^{(s-i+1)(\alpha)}_{t(\sigma)k(\beta)} \Gamma^{t(\sigma)}_{2(\alpha) (j-1)(\alpha)}  - \Gamma^{(s-i+1)(\alpha)}_{(j-1)(\alpha) t(\sigma)} \Gamma^{t(\sigma)}_{2(\alpha) k(\beta)}.
    \label{eq:3RCGb}
\end{align}
Putting \eqref{eq:3RCGa} and \eqref{eq:3RCGb} into G yields
\begin{align}
 \sum_{s=i+j+k-2}^{m_{\alpha}}\Bigg( \, & \, \notag \Gamma^{(s-i-j+3)(\alpha)}_{2(\alpha)t(\sigma)} \Gamma^{t(\sigma)}_{k(\beta)1 (\alpha)} + \Gamma^{(s-i+1)(\alpha)}_{k(\beta) t(\sigma)} \Gamma^{t(\sigma)}_{1(\alpha) j(\alpha)} - \Gamma^{(s-i+1)(\alpha)}_{j(\alpha) t(\sigma)} \Gamma^{t(\sigma)}_{1(\alpha)k(\beta)} \\ \, &  \, - \Gamma^{(s-i-j+3)(\alpha)}_{1(\alpha) t(\sigma)} \Gamma^{t(\sigma)}_{k(\beta) 2(\alpha)}  + \Gamma^{(s-i+1)(\alpha)}_{(j-1)(\alpha) t(\sigma)} \Gamma^{t(\sigma)}_{2(\alpha) k(\beta)} - \Gamma^{(s-i+1)(\alpha)}_{k(\beta) t(\sigma)} \Gamma^{t(\sigma)}_{2(\alpha) (j-1)(\alpha)}\Bigg),
    \label{eq:Gtemp2} 
\end{align}
with the use of \eqref{eq:lemmaold1diagsameblock}. Notice that if $\sigma \neq \alpha$, then \eqref{eq:Gtemp2} vanishes immediately by Lemma \ref{lemma:old1}. Thus,  after performing a relabelling of the index $\tilde{t}(t)$, together with  \eqref{eq:lemmaold1diagsameblock}, we obtain
\begin{equation*}
   G =  2\sum_{s=i+j+k-2}^{m_{\alpha}} \Gamma^{\tilde{t}}_{k(\beta) 1(\alpha)}\left(\Gamma^{(s-i+1)(\alpha)}_{(j-1)(\alpha) (\tilde{t}+1)(\alpha)} - \Gamma^{(s-i-\tilde{t}+2)(\alpha)}_{2(\alpha) (j-1)(\alpha)} + \Gamma^{(s-i-j+3)(\alpha)}_{2(\alpha) \tilde{t}(\alpha)} - \Gamma^{(s-i+1)(\alpha)}_{j(\alpha) \tilde{t}(\alpha)} \right)\theta_{s(\alpha)},
\end{equation*}
 which vanishes by the symmetry of $\nabla c$ (this can be easily seen by taking $\alpha = \beta = \gamma = \rho$ in \eqref{eq:nablacsymmulti} and letting $(k, l, i, j) \rightarrow (s-i+1, j-1, 1, \tilde{t}+1)$. 
 \vspace{-0.5em}
\newline
\newline
\underline{H.}
We will now simplify a), and b)  as given in the definition  of H. 
\newline
\newline
a) Consider  \eqref{eq:3RCmultiinR}, with $\beta \neq \rho = \sigma = \gamma  = \alpha$,   $i=k-1$, $l=1$,  
 and relabelling $(k, m, j) \rightarrow (j, k, s-i+m-1)$. Using Lemmas \ref{lemma:old1},  \ref{lemma:old4}, \ref{lemma:newold} and  \eqref{eq:Rmk:old} (together with the fact that $s \leq i+j+k-3$) we get
\begin{equation}
    \partial_{j(\alpha)}\Gamma^{(s-i+1)(\alpha)}_{k(\beta)1(\alpha)} - \partial_{(j-1)(\alpha)}\Gamma^{(s-i)(\alpha)}_{k(\beta)1(\alpha)} = \Gamma^{k(\alpha)}_{1(\alpha)k(\beta)} \left(\Gamma^{(s-i)(\alpha)}_{(j-1)(\alpha) k(\gamma)} - \Gamma^{(s-i+1)(\alpha)}_{j(\alpha) k(\alpha)} \right). 
    \label{eq:tempcompmultiga}
\end{equation}
Notice that the right-hand side of \eqref{eq:tempcompmultiga} is zero, by  \eqref{eq:Rmk:old}, when $s \leq i+j+k-4$.   Thus,  a) yields
\begin{equation*}
2 \, \Gamma^{k(\alpha)}_{1(\alpha)k(\beta)}\left(\Gamma^{(j+k-3)(\alpha)}_{(j-1)(\alpha)k(\alpha)} - \Gamma^{(j+k-2)(\alpha)}_{j(\alpha)k(\alpha)}\right)\theta_{(i+j+k-3)(\alpha)}.
\end{equation*}

b) Consider \eqref{eq:3RCmultiinR}, with $\beta \neq \rho = \sigma = \gamma  = \alpha$,  $i=1$, $k=2$, and the relabelling: $(m, j, l) \rightarrow (k, s-i+1, j-1)$. Using \eqref{Chr_a} and Remark \ref{rmk:old} together with Lemmas  \ref{lemma:old1}, \ref{lemma:old4}, and \ref{lemma:newold}   (and with the fact that $s \leq i+j+k-3$) we get 
\begin{align}
   \partial_{k(\beta)} \left(\Gamma^{(s-i-j+2)(\alpha)}_{1(\alpha)1(\alpha)} - \Gamma^{(s-i+1)(\alpha)}_{(j-1)(\alpha)2(\alpha)}\right) \notag = \, & \,  - \partial_{2(\alpha)}\Gamma^{(s-i-j+3)(\alpha)}_{k(\beta)1(\alpha)} + \Gamma^{(s-i-j+3)(\alpha)}_{1(\alpha)k(\beta)} \Gamma^{(j-1)(\alpha)}_{(j-1)(\alpha)2(\alpha)} \\ \, & \, - \Gamma^{(s-i+1)(\alpha)}_{2(\alpha)(j+k-2)(\alpha)} \Gamma^{k(\alpha)}_{1(\alpha)k(\beta)}.
   \label{eq:tempcompmultigb}
\end{align}
Notice that the right-hand side of \eqref{eq:tempcompmultigb} is zero, by Lemma \ref{lemma:old1}, when $s \leq i+j+k-3$ and thus b) is equivalent to
\begin{equation*}
2\,\left(\partial_{2(\alpha)}\Gamma^{k(\alpha)}_{1(\alpha) k(\beta)} -  \, \Gamma^{k(\alpha)}_{1(\alpha)k(\beta)} \left( \Gamma^{(j-1)(\alpha)}_{2(\alpha)(j-1)(\alpha)} - \Gamma^{(j+k-2)(\alpha)}_{2(\alpha) (j+k-2)(\alpha)}\right)\right)\theta_{(i+j+k-3)(\alpha)}.
\end{equation*}
Hence $H$, and consequently   the left-hand side of  \eqref{eq:multicomprel}, becomes
\begin{equation*}
    2 \, \Gamma^{k(\alpha)}_{1(\alpha)k(\beta)}\left(\Gamma^{(j+k-3)(\alpha)}_{(j-1)(\alpha)k(\alpha)} - \Gamma^{(j+k-2)(\alpha)}_{j(\alpha)k(\alpha)} - \Gamma^{(j-1)(\alpha)}_{2(\alpha)(j-1)(\alpha)} + \Gamma^{(j+k-2)(\alpha)}_{2(\alpha) (j+k-2)(\alpha)}\right)\theta_{(i+j+k-3)(\alpha)},
\end{equation*}
which vanishes by the symmetry of $\nabla c$ (which can be seen by taking $\alpha = \beta=\gamma=\rho$ in \eqref{eq:nablacsymmulti} and letting $(i, j, k, l) \rightarrow (j-1, k, j+k-2, 2)$).    
\newline
\newline
In the above we assumed that $3 \leq j \leq m_{\alpha}-i +2$, and $k \leq m_{\alpha} -i+1$.  Note that if both $j>m_{\alpha}-i+2$, and $k>m_{\alpha}-i+1$, then the proposition holds trivially by Corollary   \ref{lemma:completeness_thetazeros}. Suppose $j>m_{\alpha} - i +2$, $k \leq m_{\alpha}-i+1$. Then the left-hand side of  \eqref{eq:multicomprel} reads
\begin{equation}
    \partial_{j(\alpha)}\partial_{k(\beta)}\theta_{i(\alpha)} = 2 \sum_{s=k}^{m_{\alpha}-i+1} \partial_{j(\alpha)} \left(\Gamma^{s(\alpha)}_{1(\alpha)k(\beta)}\right) \theta_{(s+i-1)(\alpha)},
\end{equation}
after successive use of \eqref{eq:multisyst}, together with the induction hypothesis. Each term in the sum vanishes, by Lemma  \ref{lemma:old5}, since $j>m_{\alpha}-i+1\geq s$. 
Now suppose $j\leq m_{\alpha} - i +2$, $k > m_{\alpha}-i+1$. Then the left-hand side of \eqref{eq:multicomprel} reads
\begin{equation}
    - \partial_{k(\beta)} \partial_{j(\alpha)} \theta_{i(\alpha)} = 2 \sum_{s=j-1}^{m_{\alpha}-i+1}\partial_{k(\beta)}\left(\Gamma^{(s-j+1)(\alpha)}_{1(\alpha)1(\alpha)} - \Gamma^{s(\alpha)}_{2(\alpha)(j-1)(\alpha)}\right)\theta_{(s+i-1)(\alpha)}.
\end{equation}
Again, this vanishes due to Lemma \ref{lemma:old5} since $k > m_{\alpha}-i+1 \geq s > s-j+1$.
\vspace{-0.5em}
\newline
\newline 
(3) Let $\beta \neq \alpha \neq \gamma$. Without loss of generality, we may assume that $j \leq m_{\alpha}-i+1$, as otherwise the proposition holds trivially, by \eqref{eq:completeness_thetazeros}. First, let's assume that $k \geq m_{\alpha} - i +2$. Then, the left-hand side of \eqref{eq:multicomprel} becomes
    \begin{equation}
        -2\sum_{s=k}^{m_{\alpha}-i+1}\partial_{k(\gamma)}\left(\Gamma^{s(\alpha)}_{1(\alpha)j(\beta)}\right)\theta_{(i+s-1)(\alpha)},
    \end{equation}
    where we have truncated the upper limit of the sum to the range for which $\theta_{(i+s-1)(\alpha)}$ is defined. However, the above is zero by Lemma \ref{lemma:diffdiffblocks} since $s-1-k \leq m_{\alpha}-i-k \leq -2 <-1$. Thus we are free to assume that $3 \leq j < k \leq m_{\alpha}-i+1$.  Expanding the left-hand side of \eqref{eq:multicomprel}, with successive use of \eqref{eq:multisyst}
    \begin{align*}
         \, & \,  2 \sum_{s=k}^{m_{\alpha}-i+1} \left(\partial_{j(\beta)}\Gamma^{s(\alpha)}_{1(\alpha)k(\gamma)} - \partial_{k(\gamma)}\Gamma^{s(\alpha)}_{1(\alpha)j(\beta)}\right)\theta_{(s+i-1)(\alpha)} \underset{\text{by Lemma \ref{lemma:diffswapgen}}}{=} 0,
    \end{align*}
     after cancelling the two double sums by performing a relabelling of the indices analogously to the previous cases. 
     \vspace{-0.5em}
\newline
\newline
    In the above expansion we assumed that $k \leq m_{\alpha}-s-i+2$. If this is not the case we must consider $j \leq m_{\alpha}-s-i+2$ as we otherwise get zero trivially by \eqref{eq:completeness_thetazeros}. Thus, after the use of Lemma \ref{lemma:diffswapgen}, the left-hand side of \eqref{eq:multicomprel} reads
    \begin{equation}
         4\sum_{s=k}^{m_{\alpha}-i+1}\sum_{t=j}^{m_{\alpha}-s-i+2} \Gamma^{s(\alpha)}_{1(\alpha)k(\gamma)} \Gamma^{t(\alpha)}_{1(\alpha)j(\beta)}\theta_{(s+t+i-2)(\alpha)},
    \end{equation}
     If $t-1-j < -1$ then $\Gamma^{t(\alpha)}_{1(\alpha)j(\beta)}$ is zero, by Lemma \ref{lemma:newold}. Moreover $t-j-1 \leq m_{\alpha}-s-i+1-j$ is strictly less than $-1$ if $m_{\alpha}-s-i+2<j$. Hence, since $j<k$ the expression vanishes by Lemma \ref{lemma:newold}.
\newline
\end{proof}

\hspace{-1em}As a consequence of the above result  we have the following theorem.

\begin{theorem}\label{DT_sols}
The general  solution of the system \eqref{eq:multisyst}  depends  on $r$ arbitrary functions of a single variable
\[f_{1(1)}(u^{1(\alpha)}), \dots, f_{1(r)}(u^{1(r)})\]
and $n-r$ functions of two variables 
\[f_{2(1)}(u^{1(1)},u^{2(1)}) ,..., f_{m_1(1)}(u^{1(1)},u^{2(1)}) , ..., f_{2(r)}(u^{1(r)},u^{2(r)}),..., f_{m_r(r)}(u^{1(r)},u^{2(r)}).\]
\end{theorem}

\begin{proof}
For any block $\alpha$ let us start from the subsystem for $\theta_{m_{\alpha}(\alpha)}$. It is a closed compatible subsystem in  Darboux's form containing all partial derivatives  of the unknown function apart from $\partial_{1(\alpha)}\theta_{m_{\alpha}(\alpha)}$. From Darboux's theorem its general solution depends on an arbitrary  function $f_{1(\alpha)}(u^{1(\alpha)})$. At the next steps one has to consider the subsystems for $\theta_{(m_{\alpha}-i)(\alpha)}$ where $i=1,...,m_{\alpha}-1$. Besides  the unknown, the right-hand sides of each subsystem  also contain  the functions $\theta_{m_{\alpha}(\alpha)},...,\theta_{(m_{\alpha}-i+1)(\alpha)}$ obtained at the previous steps. They are closed compatible subsystems in  Darboux's form containing all partial derivatives  of the unknown function apart from $\partial_{1(\alpha)}\theta_{(m_{\alpha}-i)(\alpha)}$ and $\partial_{2(\alpha)}\theta_{(m_{\alpha}-i)(\alpha)}$. From Darboux's theorem the general solution of the subsystem for $\theta_{(m_{\alpha}-i)(\alpha)}$ depends on an arbitrary  function $f_{(1+i)(\alpha)}(u^{1(\alpha)},u^{2(\alpha)})$. Repeating this procedure for each block we get the result. 
\end{proof}

\subsubsection{An example}
Let us consider the case of  a single Jordan block of size $3$. In this case
\[V=X\circ=\begin{bmatrix}
		X^{1} & 0  & 0\cr
		X^{2} & X^{1} & 0\cr
		X^{3} & X^{2} & X^{1}
	\end{bmatrix},\qquad  g=\begin{bmatrix}
		\theta_{1} & \theta_2  & \theta_3\cr
		\theta_{2} & \theta_{3} & 0\cr
		\theta_3 & 0 & 0.
	\end{bmatrix}.
\]
The associated torsionless connection $\nabla $ is defined by the Christoffel symbols
\begin{eqnarray*}
\Gamma^1_{22} &=& -\frac{1}{(X^2)^2}\left(\frac{\partial X^1}{\partial u^2}X^2-\frac{\partial X^1}{\partial u^3}X^3\right),\\ 
\Gamma^1_{23} &=& \Gamma^1_{32}\,\,=\,\,\Gamma^2_{33}\,\,=\,\,-\frac{1}{X^2}\frac{\partial X^1}{\partial u^3},\\
\Gamma^2_{22} &=& \frac{1}{(X^2)^3}\left((X^2)^2\frac{\partial X^1}{\partial u^1}-(X^2)^2\frac{\partial X^2}{\partial u^2}+X^2X^3\frac{\partial X^2}{\partial u^3}-(X^3)^2\frac{\partial X^1}{\partial u^3}\right),\\
\Gamma^2_{23} &=& -\frac{1}{(X^2)^2}\left(\frac{\partial X^2}{\partial u^3}X^2-\frac{\partial X^1}{\partial u^3}X^3\right),\\ 
\Gamma^3_{22} &=& \frac{1}{(X^2)^3}\left((X^2)^2\frac{\partial X^2}{\partial u^1}-(X^2)^2\frac{\partial X^3}{\partial u^2}-X^2X^3\frac{\partial X^1}{\partial u^1}+X^2X^3\frac{\partial X^3}{\partial u^3}\right.\\
&&\left.+(X^3)^2\frac{\partial X^1}{\partial u^2}-(X^3)^2\frac{\partial X^2}{\partial u^3}\right),\\
\Gamma^3_{23} &=& \Gamma^3_{32}\,\,=\,\,\frac{1}{(X^2)^2}\left(X^2\frac{\partial X^1}{\partial u^1}-X^2\frac{\partial X^3}{\partial u^3}-X^3\frac{\partial X^1}{\partial u^2}+X^3\frac{\partial X^2}{\partial u^3}\right),\\  
\Gamma^3_{33} &=& \frac{1}{X^2}\left(\frac{\partial X^1}{\partial u^2}-\frac{\partial X^2}{\partial u^3}\right).
\end{eqnarray*}
The Dubrovin-Novikov-Tsarev system for the metric $g$ contains the condition
\[\frac{\partial  X^1}{\partial u^3}\theta_3=0,\]
which implies that $X^1$ does not  depend on $u^3$ as otherwise $g$ would be degenerate (as already observed in \cite{VF}). Taking into account this fact, the remaining system of equations contains the equations
\begin{eqnarray*}
\frac{\partial\theta_3}{\partial u^3}X^2-2\frac{\partial  X^1}{\partial u^2}\theta_3=0,\qquad\frac{\partial\theta_3}{\partial u^3}X^2+2\frac{\partial  X^2}{\partial u^3}\theta_3=0,
\end{eqnarray*}
which provide the condition
 \[\frac{\partial  X^2}{\partial u^3}=-\frac{\partial  X^1}{\partial u^2}.\]
 \newline
 \newline
 We now focus on the case  $X=E-a_0\,e$, where $a_0=\sum_{i=1}^3\epsilon_iu^i$, i.e.
 \[X^1=u^1-a_0,\quad X^2=u^2,\quad X^3=u^2.\]
  This belongs to the class of systems studied in \cite{LPVG}. For general values of the constants $\epsilon_i$, the Dubrovin-Novikov-Tsarev system does not admit  a solution. The necessary conditions $\frac{\partial  X^1}{\partial u^3}=0$ and $\frac{\partial  X^2}{\partial u^3}=-\frac{\partial  X^1}{\partial u^2}$ imply that $\epsilon_2=\epsilon_3=0$, that is $a_0=\epsilon_1u^1$. With this choice, the system belongs to  the Darboux-Tsarev class with the non-vanishing Christoffel symbols of the associated connection given by
\begin{eqnarray*}
\Gamma^2_{22} &=& \frac{1}{X^2}\left(\frac{\partial X^1}{\partial u^1}-\frac{\partial X^2}{\partial u^2}\right)=-\frac{\epsilon_1}{u_2},\\
\Gamma^3_{22} &=& \frac{X^3}{(X^2)^2}\left(-\frac{\partial X^1}{\partial u^1}+\frac{\partial X^3}{\partial u^3}\right)=\epsilon_1\frac{u_3}{u_2^2},\\
\Gamma^3_{23} &=& \Gamma^3_{32}\,\,=\,\,\frac{1}{X^2}\left(\frac{\partial X^1}{\partial u^1}-\frac{\partial X^3}{\partial u^3}\right)=-\frac{\epsilon_1}{u_2},
\end{eqnarray*}  
and the Dubrovin-Novikov-Tsarev system reads
\begin{eqnarray*}
&&\frac{\partial\theta_3}{\partial u_3}=0,\qquad u_2\left(\frac{\partial\theta_3}{\partial u_2}+\frac{\partial\theta_2}{\partial u_3}\right)+2\epsilon_1\theta_3=0,\\
&&u_2^2\left(2\frac{\partial\theta_2}{\partial u_2}-\frac{\partial\theta_1}{\partial u_3}-
\frac{\partial\theta_3}{\partial u_1}\right)+2\epsilon_1(u_2\theta_2-u_3\theta_3)=0.
\end{eqnarray*} 
As expected, the general solution depends on two arbitrary functions of two variables, $F_2(u_1, u_2)$ and $F_3(u_1, u_2)$, and one arbitrary function of a single variable, $F_1(u_1)$:
\begin{eqnarray*}
 \theta_1&=& \frac{\epsilon_1}{9}(2\epsilon_1-3)F_1u_2^{-2-\f{4}{3}\epsilon_1}u_3^2+\left(-F'_1u_2^{-\f{4}{3}\epsilon_1}+2(\epsilon_1+1)\frac{\partial F_2}{\partial u_2}\right)u_3+F_3,\\
 \theta_2&=& -\frac{2}{3}\epsilon_1F_1u_2^{-1-\f{4}{3}\epsilon_1}u_3+F_2,\\
\theta_3&=&F_1u_2^{-\f{4}{3}\epsilon_1}.
\end{eqnarray*} 
Let us study flat solutions. There  are only two values of $\epsilon_1$ such that the metric
 might be flat. Indeed, for instance, we have 
\[R^1_{212}=\frac{1}{9}\frac{\epsilon_1(\epsilon_1-3)}{u_2^2},\]
 which implies $\epsilon_1=3$ or $\epsilon_1=0$. Focusing on the first option, we have
 \[R^1_{112} = \frac{1}{2u_2F_1}\left(2u_2^5\frac{\partial^2 F_2}{\partial u_2^2}
 +14\frac{\partial F_2}{\partial u_2}u_2^4+16F_2u_2^3-F'_1\right),\]
which vanishes if and only if
\[F_2=\frac{F_4}{u_2^2}+\frac{F_5}{u_2^4}-\frac{1}{2}\frac{F'_1}{u_2^3},\]
where $F_4$  and $F_5$ are arbitrary functions of  $u_1$.  
All the remaining components of the Riemann tensor vanish if and only if
\[F_3=\frac{F_7}{u_2^2}+\frac{F_6}{u_2}+\frac{F_4^2}{F_1}-\frac{2F'_5}{u_2^3}+\frac{3F_5F'_1}{F_1u_2^3}+\frac{F_5^2}{F_1u_2^4},\]
where $F_6$  and $F_7$ are arbitrary functions of  $u_1$.  
\newline
\newline
Let us also compute the  general solution  of the  linear system for densities of  conservation laws:
\begin{eqnarray}
\label{subs1}
&&\frac{\partial\omega_3}{\partial u_2}=-\frac{\epsilon_1}{u_2}\omega_3,\qquad\frac{\partial\omega_3}{\partial u_3}=0,\\
\label{subs2}
&&\frac{\partial\omega_2}{\partial u_2}=-\frac{\epsilon_1}{u_2}\omega_2+\frac{\epsilon_1u_3}{u_2^2}\omega_3+\frac{\partial\omega_3}{\partial u_1},\qquad\frac{\partial\omega_2}{\partial u_3}=-\frac{\epsilon_1}{u_2}\omega_3,\\
\label{subs3}
&&\frac{\partial\omega_1}{\partial u_2}=\frac{\partial\omega_2}{\partial u_1},\qquad\frac{\partial\omega_1}{\partial u_3}=\frac{\partial\omega_3}{\partial u_1}.
\end{eqnarray}
The general  solution of the subsystem \eqref{subs1} is
\[\omega_3=F_1(u_1)u_2^{-\epsilon_1},\]
where $F_1$ is an arbitrary function. Substituting  such $\omega_3$ in the subsystem \eqref{subs2} we get
\[\omega_2=-\epsilon_1F_1(u_1)u_2^{-(1+\epsilon_1)}u_3+(F'_1(u_1)u_2+F_2(u_1))u_2^{-\epsilon_1},\]
where $F_2$ is a new arbitrary function. Substituting  such $\omega_2$ and $\omega_3$ in the subsystem \eqref{subs3} we get
\[\omega_1=F'_1(u_1)u_2^{-\epsilon_1}u_3-\frac{F'_2(u_1)u_2^{1-\epsilon_1}}{\epsilon_1-1}-\frac{F''_1(u_1)u_2^{2-\epsilon_1}}{\epsilon_1-2}+F_3(u_1),\]
where $F_3$ is a new arbitrary function.  Finally, it is immediate to see that $\omega =dh$
 with
\[h=F_1(u_1)u_2^{-\epsilon_1}u_3-\frac{F'_1(u_1)u_2^{2-\epsilon_1}}{\epsilon_1-2}-\frac{F_2(u_1)u_2^{1-\epsilon_1}}{\epsilon_1-1}+\tilde{F}_3(u_1),\]
where $\tilde{F}'_3=F_3$.
 
\section{Conclusions}
\hspace{-1em}In the present paper we have studied the system (Dubrovin-Novikov-Tsarev system) of first  order PDEs for the pseudo-Riemannian metrics governing the Hamiltonian formalism  for block-diagonal systems of hydrodynamic type defined by a block-diagonal matrix, whose generic block has a lower-triangular Toeplitz form. The diagonal/semisimple setting corresponds to the limiting case of $n$ blocks of  size  $1$. In this case the Dubrovin-Novikov-Tsarev system admits a family of solutions depending on $n$  arbitrary functions of a single variable. The existence of such a family follows  from
 Tsarev's integrability conditions and from a classical theorem of Darboux. Remarkably,
  Tsarev's integrability conditions provide simultaneously the compatibility  conditions for the linear system  for the symmetries and for the linear system  for the densities of conservation laws. 
\newline
\newline
In two previous papers (\cite{LPVG} and \cite{LPVGhodograph}), we started extending Tsarev's theory to the  non-diagonalisable setting. Following \cite{LPR}, we framed the theory in the context of regular F-manifolds, identifying the characteristic velocities of the system with the  components of a  vector field in canonical coordinates. Then we showed that, under mild  technical assumptions, the vector field defining the system selects a unique torsionless connection and the  integrability conditions translate  into a geometric condition relating the curvature of the connection and the structure functions of the product (extending previous results in the semisimple case in a non-trivial way). Moreover, exploiting the theory of  F-manifolds with compatible connection,
 we showed that any symmetry of the system defines a solution by means of the generalised hodograph formula (remarkably this latter result does not rely on the regularity assumption).
\newline
\newline
In the non-diagonalisable case some novelties arise: a first novelty, observed in \cite{LPVGhodograph}, is that the classical theorem of Darboux which ensures the existence of family of solutions and the solvability of the  Cauchy problem  by means of the generalised hodograph formula, cannot be applied in general.
 This requires further conditions leading to the notion of Darboux-Tsarev systems of  hydrodynamic type. A second  novelty is that, in general,  the Dubrovin-Novikov-Tsarev system does not admit solutions. For instance, using a different language, it was observed  in \cite{VF} that for a single Jordan block, the first component of the vector field defining the system should not depend on the ``last''  canonical coordinate. In this paper, studying the case of a single Jordan block of size $3$, we faced an additional condition and we expect that additional conditions appear in  general.
\newline
\newline
The main result of the paper is that, like in the semisimple/diagonal case, Darboux-Tsarev's conditions are not only sufficient for the completeness of the symmetries (as proved in  \cite{LPVG})  but also ensure the existence of family of solutions of the Dubrovin-Novikov-Tsarev system and of the linear system of densities of conservation laws, as  summarised in Table 1. In other words,  Darboux-Tsarev  systems are the non-diagonalisable counterpart, in the regular setting, of Tsarev's semi-Hamiltonian systems.
\begin{table}[h!]
\begin{center}{
\begin{tabular}{ |c|c|  }
\hline
Equations  & Functional parameters\\
\hline
\hline
$d_{\nabla}(X\circ)=0$  &  $n$ of  $1$ variable \\
\hline
$c^s_{kj}\nabla_sdh_i=c^s_{ij}\nabla_s dh_k$  & $n$ of  $1$ variable \\
\hline
$c^r_{hj}\nabla_i\theta_r+c^r_{hi}\nabla_j\theta_r-c^r_{ij}\nabla_h\theta_r=c^r_{ij}c^s_{hr}\nabla_e\theta_s-\theta_r\nabla_hc^r_{ij}$ & $r$ ($1$ var.) $+$ $n-r$ ($2$ vars.)\\
\hline
\end{tabular}
}
\end{center}
\caption{Symmetries, densities of conservation laws, and metrics for  Darboux-Tsarev systems with  $n$ components and  $r$ blocks}
\label{tab:Summary}
\end{table}

\appendix
\section{Symmetry of $\nabla c$}
\hspace{-1em}This appendix is devoted, once known that $(\nabla_jc^i_{ks}-\nabla_kc^i_{js})X^s=0$ for every $i,j,k$, to showing that $\nabla_jc^i_{ks}-\nabla_kc^i_{js}=0$ for every $i,j,k,s$. Let us denote
\begin{equation*}
	A^i_{jks}:=\nabla_jc^i_{ks}-\nabla_kc^i_{js}=\Gamma^i_{jt}c^t_{ks}-\Gamma^t_{js}c^i_{kt}-\Gamma^i_{kt}c^t_{js}+\Gamma^t_{ks}c^i_{jt}.
\end{equation*}	
According to the double-index notation, we have
\begin{align}\label{FormulaM_0}
	A^{i(\alpha)}_{j(\beta)k(\gamma)s(\sigma)}&=\nabla_{j(\beta)}c^{i(\alpha)}_{k(\gamma)s(\sigma)}-\nabla_{k(\gamma)}c^{i(\alpha)}_{j(\beta)s(\sigma)}\\
	&=\Gamma^{i(\alpha)}_{j(\beta)t(\tau)}c^{t(\tau)}_{k(\gamma)s(\sigma)}-\Gamma^{t(\tau)}_{j(\beta)s(\sigma)}c^{i(\alpha)}_{k(\gamma)t(\tau)}-\Gamma^{i(\alpha)}_{k(\gamma)t(\tau)}c^{t(\tau)}_{j(\beta)s(\sigma)}+\Gamma^{t(\tau)}_{k(\gamma)s(\sigma)}c^{i(\alpha)}_{j(\beta)t(\tau)}.\notag
\end{align}
Our assumption reads
\begin{align}\label{SymmNablacHPXs}
	\overset{r}{\underset{\sigma=1}{\sum}}\,\overset{m_\sigma}{\underset{s=1}{\sum}}\,A^{i(\alpha)}_{j(\beta)k(\gamma)s(\sigma)}\,X^{s(\sigma)}=0.
\end{align}
In order to deal with \eqref{FormulaM_0}, we first observe some useful properties.	
\begin{proposition}\label{Prop1}
	Let $\alpha\neq\beta$. Then $d_\nabla V=0$ implies
	\begin{align}\label{EqProp1}
		\Gamma^{i(\alpha)}_{j(\beta)k(\alpha)}=\Gamma^{(i-k+1)(\alpha)}_{j(\beta)1(\alpha)}\,\mathds{1}_{\{i\geq k\}}.
	\end{align}
\end{proposition}
\begin{proof}
	Let us first consider $i<k$. By Remark \ref{ChristoffelExplicit}, we have
	\begin{align*}
		\Gamma^{i(\alpha)}_{j(\beta)k(\alpha)}&=-\frac{1}{X^{1(\alpha)}-X^{1(\beta)}}\Bigg(\overset{m_\alpha}{\underset{s=k+1}{\sum}}\Gamma^{i(\alpha)}_{j(\beta)s(\alpha)}X^{(s-k+1)(\alpha)}-\overset{m_\beta}{\underset{s=j+1}{\sum}}\Gamma^{i(\alpha)}_{k(\alpha)s(\beta)}X^{(s-j+1)(\beta)}\Bigg)
	\end{align*}
	which, by an inductive argument over $k$ (decreasing, starting from $k=m_\alpha$) and $j$ (decreasing, starting from $j=m_\beta$), implies $\Gamma^{i(\alpha)}_{j(\beta)k(\alpha)}=0$ for all $j\in\{1,\dots,m_\beta\}$ and $i,k\in\{1,\dots,m_\alpha\}$ such that $i<k$. Let us now consider $i\geq k$. By Remark \ref{ChristoffelExplicit}, we have
	\begin{align*}
		\Gamma^{i(\alpha)}_{j(\beta)k(\alpha)}-\Gamma^{(i-k+1)(\alpha)}_{j(\beta)1(\alpha)}&=-\frac{1}{X^{1(\alpha)}-X^{1(\beta)}}\Bigg(\partial_{j(\beta)}X^{(i-k+1)(\alpha)}+\overset{m_\alpha}{\underset{s=k+1}{\sum}}\Gamma^{i(\alpha)}_{j(\beta)s(\alpha)}X^{(s-k+1)(\alpha)}\\&-\overset{m_\beta}{\underset{s=j+1}{\sum}}\Gamma^{i(\alpha)}_{k(\alpha)s(\beta)}X^{(s-j+1)(\beta)}-\partial_{j(\beta)}X^{(i-k+1)(\alpha)}\\&-\overset{m_\alpha}{\underset{s=2}{\sum}}\Gamma^{(i-k+1)(\alpha)}_{j(\beta)s(\alpha)}X^{s(\alpha)}+\overset{m_\beta}{\underset{s=j+1}{\sum}}\Gamma^{(i-k+1)(\alpha)}_{1(\alpha)s(\beta)}X^{(s-j+1)(\beta)}\Bigg)\\
		&=-\frac{1}{X^{1(\alpha)}-X^{1(\beta)}}\Bigg(\overset{i}{\underset{s=k+1}{\sum}}\Gamma^{i(\alpha)}_{j(\beta)s(\alpha)}X^{(s-k+1)(\alpha)}-\overset{m_\beta}{\underset{s=j+1}{\sum}}\Gamma^{i(\alpha)}_{k(\alpha)s(\beta)}X^{(s-j+1)(\beta)}\\&-\overset{i-k+1}{\underset{s=2}{\sum}}\Gamma^{(i-k+1)(\alpha)}_{j(\beta)s(\alpha)}X^{s(\alpha)}+\overset{m_\beta}{\underset{s=j+1}{\sum}}\Gamma^{(i-k+1)(\alpha)}_{1(\alpha)s(\beta)}X^{(s-j+1)(\beta)}\Bigg)\\
		&=-\frac{1}{X^{1(\alpha)}-X^{1(\beta)}}\Bigg(\overset{i-k+1}{\underset{s=2}{\sum}}\big(\Gamma^{i(\alpha)}_{j(\beta)(s+k-1)(\alpha)}-\Gamma^{(i-k+1)(\alpha)}_{j(\beta)s(\alpha)}\big)X^{s(\alpha)}\\&-\overset{m_\beta}{\underset{s=j+1}{\sum}}\big(\Gamma^{i(\alpha)}_{k(\alpha)s(\beta)}-\Gamma^{(i-k+1)(\alpha)}_{1(\alpha)s(\beta)}\big)X^{(s-j+1)(\beta)}\Bigg)
	\end{align*}
	which, by an inductive argument over $i-k$ (increasing, starting from $i-k=0$) and $j$ (decreasing, starting from $j=m_\beta$), implies $\Gamma^{i(\alpha)}_{j(\beta)k(\alpha)}-\Gamma^{(i-k+1)(\alpha)}_{j(\beta)1(\alpha)}=0$ for all $j\in\{1,\dots,m_\beta\}$ and $i,k\in\{1,\dots,m_\alpha\}$ such that $i\geq k$.
\end{proof}
\begin{proposition}\label{Prop2}
	Let $\alpha\neq\beta$. Then $d_\nabla V=0$ implies
	\begin{align}\label{EqProp2}
		\Gamma^{i(\alpha)}_{j(\beta)k(\beta)}=\Gamma^{i(\alpha)}_{(j+k-1)(\beta)1(\beta)}\,\mathds{1}_{\{j+k\leq m_\beta+1\}}.
	\end{align}
\end{proposition}
\begin{proof}
	Since \eqref{EqProp2} is trivial for $j=1$ or $k=1$, we assume $j,k\geq2$. Without loss of generality, $j\leq k$. Let us first consider $j+k>m_\beta+1$. By Remark \ref{ChristoffelExplicit}, we have
	\begin{align*}
		\Gamma^{i(\alpha)}_{j(\beta)k(\beta)}&=\frac{1}{X^{2(\beta)}}\Bigg(\overset{m_\beta}{\underset{s=k+1}{\sum}}\Gamma^{i(\alpha)}_{(j-1)(\beta)s(\beta)}X^{(s-k+1)(\beta)}-\overset{m_\beta}{\underset{s=j+1}{\sum}}\Gamma^{i(\alpha)}_{k(\beta)s(\beta)}X^{(s-j+2)(\beta)}\Bigg)
	\end{align*}
	which, by an inductive argument over $k$ (decreasing, starting from $k=m_\beta$) and $j$ (decreasing, starting from $j=k$), implies $\Gamma^{i(\alpha)}_{j(\beta)k(\beta)}=0$ for all $i\in\{1,\dots,m_\alpha\}$ and $j,k\in\{1,\dots,m_\beta\}$ such that $j+k>m_\beta+1$. Let us now consider $j+k\leq m_\beta+1$ (in particular, $2\leq j\leq k\leq m_\beta-1$). By Remark \ref{ChristoffelExplicit}, we have
	\begin{align*}
		\Gamma^{i(\alpha)}_{j(\beta)k(\beta)}&=\frac{1}{X^{2(\beta)}}\Bigg(\overset{m_\beta-j+2}{\underset{s=k+1}{\sum}}\Gamma^{i(\alpha)}_{(j-1)(\beta)s(\beta)}X^{(s-k+1)(\beta)}-\overset{m_\beta-k+1}{\underset{s=j+1}{\sum}}\Gamma^{i(\alpha)}_{k(\beta)s(\beta)}X^{(s-j+2)(\beta)}\Bigg)\\
		&=\Gamma^{i(\alpha)}_{(j-1)(\beta)(k+1)(\beta)}+\frac{1}{X^{2(\beta)}}\Bigg(\overset{m_\beta-j+2}{\underset{s=k+2}{\sum}}\Gamma^{i(\alpha)}_{(j-1)(\beta)s(\beta)}X^{(s-k+1)(\beta)}-\overset{m_\beta-k+1}{\underset{s=j+1}{\sum}}\Gamma^{i(\alpha)}_{k(\beta)s(\beta)}X^{(s-j+2)(\beta)}\Bigg)
	\end{align*}
	which, by an inductive argument over $j+k$ (decreasing, starting from ${j+k=m_\beta+1}$) and $j$ (increasing, starting from $j=2$), implies $\Gamma^{i(\alpha)}_{j(\beta)k(\beta)}=\Gamma^{i(\alpha)}_{(j+k-1)(\beta)1(\beta)}$ for all $i\in\{1,\dots,m_\alpha\}$ and ${j,k\in\{1,\dots,m_\beta\}}$ such that $j+k\leq m_\beta+1$.
\end{proof}
\begin{remark}\label{Remsgeq2}
	As a consequence of Propositions \ref{Prop1} and \ref{Prop2}, we have
	\begin{equation*}
		A^{i(\alpha)}_{j(\beta)k(\gamma)1(\sigma)}=0
	\end{equation*}
	for all suitable indices. In fact,
	\begin{align*}
		A^{i(\alpha)}_{j(\beta)k(\gamma)1(\sigma)}&=\Gamma^{i(\alpha)}_{j(\beta)t(\tau)}c^{t(\tau)}_{k(\gamma)1(\sigma)}-\Gamma^{t(\tau)}_{j(\beta)1(\sigma)}c^{i(\alpha)}_{k(\gamma)t(\tau)}-\Gamma^{i(\alpha)}_{k(\gamma)t(\tau)}c^{t(\tau)}_{j(\beta)1(\sigma)}+\Gamma^{t(\tau)}_{k(\gamma)1(\sigma)}c^{i(\alpha)}_{j(\beta)t(\tau)}\\
		&=\Gamma^{i(\alpha)}_{j(\beta)k(\gamma)}(\delta_{\gamma\sigma}-\delta_{\beta\sigma})-\Gamma^{(i-k+1)(\alpha)}_{j(\beta)1(\sigma)}\delta^\alpha_\gamma+\Gamma^{(i-j+1)(\alpha)}_{k(\gamma)1(\sigma)}\delta^\alpha_\beta
	\end{align*}
	which for $\beta=\gamma$ becomes
	\begin{align*}
		A^{i(\alpha)}_{j(\beta)k(\beta)1(\sigma)}&=\delta^\alpha_\beta\big(\Gamma^{(i-j+1)(\alpha)}_{k(\alpha)1(\sigma)}-\Gamma^{(i-k+1)(\alpha)}_{j(\alpha)1(\sigma)}\big),
	\end{align*}
	vanishing by \eqref{EqProp1} when $\sigma\neq\alpha$ and then by flatness of the unit when $\sigma=\alpha$, while for $\beta\neq\gamma$ becomes
	\begin{align*}
		A^{i(\alpha)}_{j(\beta)k(\gamma)1(\sigma)}&=\Gamma^{i(\alpha)}_{j(\beta)k(\gamma)}(\delta_{\gamma\sigma}-\delta_{\beta\sigma})-\Gamma^{(i-k+1)(\alpha)}_{j(\beta)1(\sigma)}\delta^\alpha_\gamma+\Gamma^{(i-j+1)(\alpha)}_{k(\gamma)1(\sigma)}\delta^\alpha_\beta,
	\end{align*}
	vanishing trivially when $\alpha,\beta,\gamma$ are pairwise distinct and reading, when $\alpha=\gamma\neq\beta$ (covering also the case where $\alpha=\beta\neq\gamma$), 
	\begin{align*}
		A^{i(\alpha)}_{j(\beta)k(\alpha)1(\sigma)}&=\Gamma^{i(\alpha)}_{j(\beta)k(\alpha)}(\delta_{\alpha\sigma}-\delta_{\beta\sigma})-\Gamma^{(i-k+1)(\alpha)}_{j(\beta)1(\sigma)},
	\end{align*}
	vanishing trivially when $\sigma\neq\alpha\neq\beta\neq\sigma$ and by \eqref{EqProp1} when either $\sigma=\alpha\neq\beta$ or $\sigma=\beta\neq\alpha$.		
\end{remark}
Now, we possess the tools to tackle \eqref{FormulaM_0}. In light of Remark \ref{Remsgeq2}, without loss of generality, we consider $s\geq2$. We have
\begin{align*}
	A^{i(\alpha)}_{j(\beta)k(\gamma)s(\sigma)}&=\Gamma^{i(\alpha)}_{j(\beta)t(\tau)}c^{t(\tau)}_{k(\gamma)s(\sigma)}-\Gamma^{t(\tau)}_{j(\beta)s(\sigma)}c^{i(\alpha)}_{k(\gamma)t(\tau)}-\Gamma^{i(\alpha)}_{k(\gamma)t(\tau)}c^{t(\tau)}_{j(\beta)s(\sigma)}+\Gamma^{t(\tau)}_{k(\gamma)s(\sigma)}c^{i(\alpha)}_{j(\beta)t(\tau)}
\end{align*}
which trivially vanishes when both $\alpha$ and $\sigma$ are different from both $j$ and $k$. We must then consider the following cases:
\begin{itemize}
	\item[a.] $\alpha=\beta=\gamma=\sigma$
	\item[b.] $\alpha=\sigma=\gamma\neq\beta$ (covering $\alpha=\sigma=\beta\neq\gamma$ as well)
	\item[c.] $\alpha\neq\beta=\gamma=\sigma$
	\item[d.] $\alpha=\beta=\gamma\neq\sigma$
	\item[e.] $\alpha=\beta\neq\gamma=\sigma$ (covers $\alpha=\gamma\neq\beta=\sigma$ as well)
	\item[f.] $\sigma\neq\alpha=\beta\neq\gamma\neq\sigma$ (covering $\sigma\neq\alpha=\gamma\neq\beta\neq\sigma$ as well)
	\item[g.] $\alpha\neq\sigma=\gamma\neq\beta\neq\alpha$ (covering $\alpha\neq\sigma=\beta\neq\gamma\neq\alpha$ as well).
\end{itemize}
\textbf{Case b}: $\alpha=\sigma=\gamma\neq\beta$.
\begin{align*}
	A^{i(\alpha)}_{j(\beta)k(\alpha)s(\alpha)}&=\Gamma^{i(\alpha)}_{j(\beta)t(\alpha)}c^{t(\alpha)}_{k(\alpha)s(\alpha)}-\Gamma^{t(\alpha)}_{j(\beta)s(\alpha)}c^{i(\alpha)}_{k(\alpha)t(\alpha)}=\Gamma^{i(\alpha)}_{j(\beta)(k+s-1)(\alpha)}-\Gamma^{(i-k+1)(\alpha)}_{j(\beta)s(\alpha)}\overset{\eqref{EqProp1}}{=}0.
\end{align*}
\textbf{Case c}: $\alpha\neq\beta=\gamma=\sigma$.
\begin{align*}
	A^{i(\alpha)}_{j(\beta)k(\beta)s(\beta)}&=\Gamma^{i(\alpha)}_{j(\beta)t(\beta)}c^{t(\beta)}_{k(\beta)s(\beta)}-\Gamma^{i(\alpha)}_{k(\beta)t(\beta)}c^{t(\beta)}_{j(\beta)s(\beta)}=\Gamma^{i(\alpha)}_{j(\beta)(k+s-1)(\beta)}-\Gamma^{i(\alpha)}_{k(\beta)(j+s-1)(\beta)}\overset{\eqref{EqProp2}}{=}0.
\end{align*}
\textbf{Case d}: $\alpha=\beta=\gamma\neq\sigma$.
\begin{align*}
	A^{i(\alpha)}_{j(\alpha)k(\alpha)s(\sigma)}&=-\Gamma^{t(\alpha)}_{j(\alpha)s(\sigma)}c^{i(\alpha)}_{k(\alpha)t(\alpha)}+\Gamma^{t(\alpha)}_{k(\alpha)s(\sigma)}c^{i(\alpha)}_{j(\alpha)t(\alpha)}=-\Gamma^{(i-k+1)(\alpha)}_{j(\alpha)s(\sigma)}+\Gamma^{(i-j+1)(\alpha)}_{k(\alpha)s(\sigma)}\overset{\eqref{EqProp1}}{=}0.
\end{align*}
\textbf{Case e}: $\alpha=\beta\neq\gamma=\sigma$.
\begin{align*}
	A^{i(\alpha)}_{j(\alpha)k(\gamma)s(\gamma)}&=\Gamma^{i(\alpha)}_{j(\alpha)t(\gamma)}c^{t(\gamma)}_{k(\gamma)s(\gamma)}+\Gamma^{t(\alpha)}_{k(\gamma)s(\gamma)}c^{i(\alpha)}_{j(\alpha)t(\alpha)}=\Gamma^{i(\alpha)}_{j(\alpha)(k+s-1)(\gamma)}+\Gamma^{(i-j+1)(\alpha)}_{k(\gamma)s(\gamma)}\\
	&\overset{\eqref{EqProp1}}{\underset{\eqref{EqProp2}}{=}}\Gamma^{(i-j+1)(\alpha)}_{1(\alpha)(k+s-1)(\gamma)}+\Gamma^{(i-j+1)(\alpha)}_{1(\gamma)(k+s-1)(\gamma)}
\end{align*}
which vanishes by means of flatness of the unit.
\\\textbf{Case f}: $\sigma\neq\alpha=\beta\neq\gamma\neq\sigma$.
\begin{align*}
	A^{i(\alpha)}_{j(\alpha)k(\gamma)s(\sigma)}&=\Gamma^{t(\alpha)}_{k(\gamma)s(\sigma)}c^{i(\alpha)}_{j(\alpha)t(\alpha)}
\end{align*}
which trivially vanishes as the blocks appearing in $\Gamma^{t(\alpha)}_{k(\gamma)s(\sigma)}$ are pairwise distinct.
\\\textbf{Case g}: $\alpha\neq\sigma=\gamma\neq\beta\neq\alpha$.
\begin{align*}
	A^{i(\alpha)}_{j(\beta)k(\gamma)s(\gamma)}&=\Gamma^{i(\alpha)}_{j(\beta)t(\gamma)}c^{t(\gamma)}_{k(\gamma)s(\gamma)}
\end{align*}
which trivially vanishes as the blocks appearing in $\Gamma^{i(\alpha)}_{j(\beta)t(\gamma)}$ are pairwise distinct.
\\In order to deal with \emph{case a}, we need additional results. By choosing $\alpha=\beta=\gamma$ in \eqref{SymmNablacHPXs}, we get
\begin{align*}
	&\overset{r}{\underset{\sigma=1}{\sum}}\,\overset{m_\sigma}{\underset{s=1}{\sum}}\,A^{i(\alpha)}_{j(\alpha)k(\alpha)s(\sigma)}\,X^{s(\sigma)}=0
\end{align*}
where, due to the above \emph{cases b-g} and Remark \ref{Remsgeq2}, only $\sigma=\alpha$ and $s\geq2$ survive, yielding
\begin{align}\label{SymmNablacHPXs_bis}
	&\overset{m_\alpha}{\underset{s=2}{\sum}}\,A^{i(\alpha)}_{j(\alpha)k(\alpha)s(\alpha)}\,X^{s(\alpha)}=0.
\end{align}

\begin{proposition}\label{Prop5}
	The condition \eqref{SymmNablacHPXs} implies
	\begin{align}\label{EqProp5}
		\Gamma^{i(\alpha)}_{j(\alpha)(k+1)(\alpha)}-\Gamma^{i(\alpha)}_{(j+1)(\alpha)k(\alpha)}=\Gamma^{(i-k+1)(\alpha)}_{2(\alpha)j(\alpha)}-\Gamma^{(i-j+1)(\alpha)}_{2(\alpha)k(\alpha)},\qquad j,k\in\{2,\dots,m_\alpha\}.
	\end{align}
\end{proposition}
\begin{proof}
	The condition \eqref{SymmNablacHPXs_bis} reads
	\begin{align*}
		0&=\overset{m_\alpha}{\underset{s=2}{\sum}}\,A^{i(\alpha)}_{j(\alpha)k(\alpha)s(\alpha)}\,X^{s(\alpha)}\\
		&=\overset{m_\alpha}{\underset{s=2}{\sum}}\,\big(\Gamma^{i(\alpha)}_{j(\alpha)(k+s-1)(\alpha)}-\Gamma^{(i-k+1)(\alpha)}_{j(\alpha)s(\alpha)}-\Gamma^{i(\alpha)}_{k(\alpha)(j+s-1)(\alpha)}+\Gamma^{(i-j+1)(\alpha)}_{k(\alpha)s(\alpha)}\big)\,X^{s(\alpha)}\\
		&=\big(\Gamma^{i(\alpha)}_{j(\alpha)(k+1)(\alpha)}-\Gamma^{(i-k+1)(\alpha)}_{j(\alpha)2(\alpha)}-\Gamma^{i(\alpha)}_{k(\alpha)(j+1)(\alpha)}+\Gamma^{(i-j+1)(\alpha)}_{k(\alpha)2(\alpha)}\big)\,X^{2(\alpha)}\\
		&+\overset{m_\alpha}{\underset{s=3}{\sum}}\,\big(\Gamma^{i(\alpha)}_{j(\alpha)(k+s-1)(\alpha)}-\Gamma^{(i-k+1)(\alpha)}_{j(\alpha)s(\alpha)}-\Gamma^{i(\alpha)}_{k(\alpha)(j+s-1)(\alpha)}+\Gamma^{(i-j+1)(\alpha)}_{k(\alpha)s(\alpha)}\big)\,X^{s(\alpha)}
	\end{align*}
	which yields $\Gamma^{i(\alpha)}_{j(\alpha)(k+1)(\alpha)}-\Gamma^{(i-k+1)(\alpha)}_{j(\alpha)2(\alpha)}-\Gamma^{i(\alpha)}_{k(\alpha)(j+1)(\alpha)}+\Gamma^{(i-j+1)(\alpha)}_{k(\alpha)2(\alpha)}=0$ for $j+k=2m_\alpha$, as $X^{2(\alpha)}\neq0$, and \eqref{EqProp5}, by an inductive argument over $j+k$.
\end{proof}
\begin{corollary}\label{Cor5}
	The condition \eqref{EqProp5} implies
	\begin{align}\label{EqCor5}
		\Gamma^{i(\alpha)}_{j(\alpha)(k+s-1)(\alpha)}-\Gamma^{i(\alpha)}_{(j+s-1)(\alpha)k(\alpha)}=\Gamma^{(i-k+1)(\alpha)}_{s(\alpha)j(\alpha)}-\Gamma^{(i-j+1)(\alpha)}_{s(\alpha)k(\alpha)},\qquad j,k,s\in\{2,\dots,m_\alpha\}.
	\end{align}
\end{corollary}
\begin{proof}
	For $s=2$, \eqref{EqCor5} trivially reduces to \eqref{EqProp5}. Let us fix $s\in\{2,\dots,m_\alpha\}$, and inductively assume that
	\begin{align}\label{EqCor5_indhp}
		\Gamma^{i(\alpha)}_{j(\alpha)(k+t-1)(\alpha)}-\Gamma^{i(\alpha)}_{(j+t-1)(\alpha)k(\alpha)}=\Gamma^{(i-k+1)(\alpha)}_{t(\alpha)j(\alpha)}-\Gamma^{(i-j+1)(\alpha)}_{t(\alpha)k(\alpha)},\, j,k\in\{2,\dots,m_\alpha\},\,t\in\{2,\dots,s-1\}.
	\end{align}
	By adding and subtracting the term $\Gamma^{i(\alpha)}_{(j+1)(\alpha)(k+s-2)(\alpha)}$, we have
	\begin{align*}
		\Gamma^{i(\alpha)}_{j(\alpha)(k+s-1)(\alpha)}-\Gamma^{i(\alpha)}_{(j+s-1)(\alpha)k(\alpha)}&=\underbrace{\Gamma^{i(\alpha)}_{j(\alpha)(k+s-1)(\alpha)}-\Gamma^{i(\alpha)}_{(j+1)(\alpha)(k+s-2)(\alpha)}}_{\text{Use \eqref{EqProp5} replacing $(i,j,k)$ with $(i,j,k+s-2)$}}\\&+\underbrace{\Gamma^{i(\alpha)}_{(j+1)(\alpha)(k+s-2)(\alpha)}-\Gamma^{i(\alpha)}_{(j+s-1)(\alpha)k(\alpha)}}_{\text{Use \eqref{EqCor5_indhp} replacing $(i,j,k,t)$ with $(i,j+1,k,s-1)$}}\\
		&=\Gamma^{(i-k-s+3)(\alpha)}_{2(\alpha)j(\alpha)}\underbrace{-\Gamma^{(i-j+1)(\alpha)}_{2(\alpha)(k+s-2)(\alpha)}}_{\substack{\text{Use \eqref{EqCor5_indhp} replacing $(i,j,k,t)$}\\\text{with $(i-j+1,2,k,s-1)$}}}\\
		&+\underbrace{\Gamma^{(i-k+1)(\alpha)}_{(s-1)(\alpha)(j+1)(\alpha)}}_{\substack{\text{Use \eqref{EqProp5} replacing $(i,j,k)$}\\\text{with $(i-k+1,j,s-1)$}}}-\Gamma^{(i-j)(\alpha)}_{(s-1)(\alpha)k(\alpha)}\\
		&=\Gamma^{(i-k-s+3)(\alpha)}_{2(\alpha)j(\alpha)}-\Gamma^{(i-j+1)(\alpha)}_{s(\alpha)k(\alpha)}-\Gamma^{(i-j-k+2)(\alpha)}_{(s-1)(\alpha)2(\alpha)}+\Gamma^{(i-j)(\alpha)}_{(s-1)(\alpha)k(\alpha)}\\
		&+\Gamma^{(i-k+1)(\alpha)}_{j(\alpha)s(\alpha)}-\Gamma^{(i-k-s+3)(\alpha)}_{2(\alpha)j(\alpha)}+\Gamma^{(i-k-j+2)(\alpha)}_{2(\alpha)(s-1)(\alpha)}-\Gamma^{(i-j)(\alpha)}_{(s-1)(\alpha)k(\alpha)}\\
		&=\Gamma^{(i-k+1)(\alpha)}_{j(\alpha)s(\alpha)}-\Gamma^{(i-j+1)(\alpha)}_{s(\alpha)k(\alpha)}.
	\end{align*}
\end{proof}
We are finally ready to tackle \emph{case a}.
\\\textbf{Case a}: $\alpha=\beta=\gamma=\sigma$.
\begin{align*}
	A^{i(\alpha)}_{j(\alpha)k(\alpha)s(\alpha)}&=\Gamma^{i(\alpha)}_{j(\alpha)t(\alpha)}c^{t(\alpha)}_{k(\alpha)s(\alpha)}-\Gamma^{t(\alpha)}_{j(\alpha)s(\alpha)}c^{i(\alpha)}_{k(\alpha)t(\alpha)}-\Gamma^{i(\alpha)}_{k(\alpha)t(\alpha)}c^{t(\alpha)}_{j(\alpha)s(\alpha)}+\Gamma^{t(\alpha)}_{k(\alpha)s(\alpha)}c^{i(\alpha)}_{j(\alpha)t(\alpha)}
\end{align*}
which, by flatness of the unit and \eqref{EqProp1}, trivially vanishes whenever $j=1$ or $k=1$. Let us then consider $j,k\geq2$. For each $s\geq2$, we have
\begin{align*}
	A^{i(\alpha)}_{j(\alpha)k(\alpha)s(\alpha)}&=\Gamma^{i(\alpha)}_{j(\alpha)(k+s-1)(\alpha)}-\Gamma^{i(\alpha)}_{k(\alpha)(j+s-1)(\alpha)}-\Gamma^{(i-k+1)(\alpha)}_{j(\alpha)s(\alpha)}+\Gamma^{(i-j+1)(\alpha)}_{k(\alpha)s(\alpha)}\overset{\eqref{EqCor5}}{=}0.
\end{align*}

\end{document}